\documentclass[aos]{imsart}
\RequirePackage{amsthm,amsmath,amsfonts,amssymb}
\RequirePackage[authoryear]{natbib}%% uncomment this for author-year citations
\RequirePackage[colorlinks,citecolor=blue,urlcolor=blue]{hyperref}%% uncomment this for coloring bibliography citations and linked URLs

\startlocaldefs
\usepackage{thmtools}
\usepackage{thm-restate}

\theoremstyle{plain}
\newtheorem{thm}{Theorem}[section] % reset theorem numbering for each section
\newtheorem{lem}[thm]{Lemma}
\newtheorem{prop}[thm]{Proposition}
\newtheorem{cor}[thm]{Corollary}
\newtheorem{lemma}[thm]{Lemma}
\theoremstyle{remark}
\newtheorem{defn}[thm]{Definition} % definition numbers are dependent on theorem numbers

\newtheorem{example}[thm]{Example}
\newtheorem{remark}[thm]{Remark}
\usepackage{paralist, color, graphicx}
\usepackage{subcaption}
\usepackage{verbatim}
\usepackage{mathtools}

\usepackage{cleveref}

\usepackage{mathtools}
\usepackage{color}
\usepackage{todonotes}
\usepackage[ruled,vlined,linesnumbered]{algorithm2e}

\usepackage{mathrsfs}%fancy fonts like mscr
\usepackage{tikz}

\usepackage{enumerate}

\newcommand{\mscr}[1]{\mathcal{#1}}

\newcommand{\twid}[1]{\widetilde{#1}}

\newcommand{\ZZ}{\mathbb{Z}}
\newcommand{\RR}{\mathbb{R}}

\newcommand{\EE}{\mathbb{E}}

\newcommand{\al}{\alpha}

\newcommand{\de}{\delta}
\newcommand{\ep}{\epsilon}

\newcommand{\ta}{\theta}
\newcommand{\Ta}{\Theta}

\renewcommand{\l}{\left}
\renewcommand{\r}{\right}

\newcommand{\defeq}{\vcentcolon=}

\newcommand{\iid}{\overset{\text{iid}}{\sim}}

\DeclareMathOperator{\var}{Var}

\newcommand{\ol}{\overline}
\newcommand{\ul}{\underline}

\newcommand{\thx}{\theta_X}
\newcommand{\thy}{\theta_Y}

\endlocaldefs

\begin{document}

\begin{frontmatter}
\title{Canonical Noise Distributions and Private Hypothesis Tests}
%\title{A sample article title with some additional note\thanksref{T1}}
\runtitle{Canonical Noise Distributions}
%\thankstext{T1}{A sample of additional note to the title.}

\begin{aug}
%%%%%%%%%%%%%%%%%%%%%%%%%%%%%%%%%%%%%%%%%%%%%%%
%% Only one address is permitted per author. %%
%% Only division, organization and e-mail is %%
%% included in the address.                  %%
%% Additional information can be included in %%
\author[A]{\fnms{Jordan} \snm{Awan}\ead[label=e1]{jawan@purdue.edu}} 
%\author[B]{\fnms{???} \snm{???}\ead[label=e2,mark]{???@???}}
\and
\author[B]{\fnms{Salil} \snm{Vadhan}\ead[label=e2]{salil\_vadhan@harvard.edu}}
%%%%%%%%%%%%%%%%%%%%%%%%%%%%%%%%%%%%%%%%%%%%%%
%% Addresses                                %%
%%%%%%%%%%%%%%%%%%%%%%%%%%%%%%%%%%%%%%%%%%%%%%
\address[A]{Purdue University, \printead{e1}}
\address[B]{Harvard University, \printead{e2}}
\end{aug}
\begin{abstract}
$f$-DP has recently been proposed as a generalization of differential privacy allowing a lossless analysis of composition, post-processing, and privacy amplification via subsampling. In the setting of $f$-DP, we propose the concept of a \emph{canonical noise distribution} (CND), {the first mechanism designed for an arbitrary $f$-DP guarantee}. The notion of CND captures whether an additive privacy mechanism {perfectly matches the privacy guarantee of}
%is appropriately tailored for
a given $f$. {We prove that a CND always exists, and} give a construction that produces a CND {for any $f$}. We show that private hypothesis tests are intimately related to CNDs, allowing for the release of private $p$-values at no additional privacy cost, as well as the construction of uniformly most powerful (UMP) tests for binary data, {within the general $f$-DP framework.}

We apply our techniques to the problem of difference-of-proportions testing, and construct a UMP unbiased (UMPU) ``semi-private'' test which upper bounds the performance of any {$f$-DP} test. Using this as a benchmark, we propose a private test based on the inversion of characteristic functions, which allows for optimal inference for the two population parameters and is nearly as powerful as the semi-private UMPU. When specialized to the case of $(\ep,0)$-DP, we show empirically that our proposed test is more powerful than any $(\ep/\sqrt 2)$-DP test and has more accurate type I errors than the classic normal approximation test.
\end{abstract}



\end{frontmatter}
%%%%%%%%%%%%%%%%%%%%%%%%%%%%%%%%%%%%%%%%%%%%%%
%% Please use \tableofcontents for articles %%
%% with 50 pages and more                   %%
%%%%%%%%%%%%%%%%%%%%%%%%%%%%%%%%%%%%%%%%%%%%%%
%\tableofcontents

%%%%%%%%%%%%%%%%%%%%%%%%%%%%%%%%%%%%%%%%%%%%%%
%%%% Main text entry area:

\thispagestyle{empty}
\section{Introduction}
The concept of differential privacy (DP) was introduced in \citet{dwork2006calibrating}, which offered a framework for the construction of private mechanisms and a rigorous notion of what it means to limit privacy loss when performing statistical releases on sensitive data. DP requires that the randomized algorithm $M$ performing the release has the property that for any two datasets $X$ and $X'$ which differ in one individual's data (\emph{adjacent datasets}), the distributions of $M(X)$ and $M(X')$ are ``close.'' Since this seminal paper, many variants of differential privacy have been proposed; the variants primarily differ in how they formulate the notion of closeness. For example, pure and approximate DP are phrased in terms of bounding the probabilities of sets of outputs, according to $M(X)$ versus $M(X')$ \citep{dwork2014algorithmic}, whereas concentrated \citep{bun2016concentrated} and Renyi \citep{mironov2017renyi} DP are based on bounding a divergence between $M(X)$ and $M(X')$. 

\citet{Wasserman2010:StatisticalFDP} and \citet{Kairouz2017} showed that pure and approximate DP can be expressed as imposing constraints on the type I and type II errors of hypothesis tests which seek to discriminate between two adjacent databases. Recently \citet{dong2022gaussian} expanded this view, defining $f$-DP which allows for an arbitrary bound to be placed on the receiver-operator curve (ROC) or \emph{tradeoff} function when testing between two adjacent databases. It is shown in \citet{dong2022gaussian} that $f$-DP retains many of the useful properties of DP such as post-processing, composition, and subsampling and allows for loss-less calculation of the privacy cost of each of these operations. Furthermore, as special cases, $f$-DP contains both pure and approximate DP, and contains relatives of divergence-based notions of DP as well (e.g., Gaussian DP (GDP) is slightly stronger than zero-concentrated DP). %For these reasons, we view $f$-DP as the most natural definition of privacy, and is the version of DP that we work with in this paper. 

In this paper, we study two basic and fundamental privacy questions in the framework of $f$-DP. The first is based on {optimizing} the basic mechanism of adding independent noise to a real-valued statistic, and the second is about constructing hypothesis tests under the constraint of DP. We show that in fact, the two problems are intricately related, where the ``canonical additive noise distribution'' enables private $p$-values ``for free,'' and gives a closed form construction of certain optimal hypothesis tests. 

One of the most basic and fundamental types of privacy mechanisms is noise addition, where independent noise is added to a real-valued statistic. {Additive mechanisms are not only widely used by themselves, but are also often a key ingredient to more complex mechanisms such as functional mechanism \citep{zhang2012functional}, objective perturbation \citep{chaudhuri2011differentially}, stochastic gradient descent \citep{abadi2016deep}, and the sparse vector technique \citep{dwork2009complexity}, to name a few. The oldest and most widely used additive mechanisms are the Laplace and Gaussian mechanisms, but there have since been many proposed distributions which satisfy different definitions of DP.} 
%The Laplace mechanism appeared along with the original definition of DP, and the Gaussian mechanism was another one of the earliest privacy mechanisms designed for approximate DP. 
A natural question is what noise distributions are ``optimal" or ``canonical" for a given definition of privacy. The geometric mechanism/discrete Laplace mechanism is optimal for $\ep$-DP counts, in terms of maximizing Bayesian utility \citep{ghosh2012universally}, the staircase mechanism is optimal for $\ep$-DP in terms of $\ell_1$ or $\ell_2$-error \citep{geng2015optimal}, and the truncated-uniform-Laplace (Tulap) distribution generalizes both the discrete Laplace and staircase mechanisms and is optimal for $(\ep,\de)$-DP in terms of generating uniformly most powerful (UMP) hypothesis tests and uniformly most accurate (UMA) confidence intervals for Bernoulli data \citep{awan2018differentially,awan2020differentially}. With divergence-based definitions of privacy, Gaussian noise is argued to be canonical for (zero) concentrated DP \citep{bun2016concentrated}, and the sinh-normal distribution is argued to be canonical for truncated concentrated DP \citep{bun2018composable}. 

In this paper, we give {the first} formal definition of a \emph{canonical noise distribution} (CND) which {captures the notion of whether a distribution tightly matches a privacy guarantee $f$-DP.} 
%is {the first mechanism developed for an arbitrary} $f$-DP notion of privacy, {and captures whether the distribution tightly matches the privacy guarantee}. 
We show that the Gaussian distribution is canonical for Gaussian differential privacy (GDP), and the Tulap distribution is canonical for $(\ep,\de)$-DP. We prove that a CND always exists for any {nontrivial} symmetric tradeoff function $f$, and give a general construction to generate a CND given {any} tradeoff function $f$. {This construction results in the first general mechanism for an arbitrary $f$-DP guarantee.} {In the special case of $(\ep,\de)$-DP, our} construction results in the Tulap distribution.% {proposed by \citet{awan2018differentially}.}%in the case of $(\ep,\de)$-DP. 

Another basic privacy question is on the nature of DP hypothesis tests. \citet{awan2018differentially}  showed that for independent Bernoulli data, there exists uniformly most powerful (UMP) $(\ep,\de)$-DP tests which are based on the Tulap distribution, enabling ``free'' private $p$-values, at no additional cost to privacy. 

We show that for an arbitrary tradeoff function $f$ {and} any $f$-DP test, a free private $p$-value can always be generated in terms of a CND for $f$. We also extend the main results of \citet{awan2018differentially} from $(\ep,\de)$-DP to $f$-DP as well as from i.i.d. Bernoulli variables to exchangeable binary data. This extension shows that the CND is the proper generalization of the Tulap distribution, and gives an explicit construction of the most powerful $f$-DP test for binary data, {the first DP hypothesis test for a general $f$-DP guarantee.}

We end with an extensive application to private difference-of-proportions testing. Testing two population proportions is a very basic and common hypothesis testing setting that arises when there are two groups with binary responses, such as A/B testing, clinical trials, and observational studies. As such, the techniques for testing these hypotheses are standardized and included in most introductory statistics textbooks. However, there currently lacks a theoretically based private test with accurate sensitivity and specificity. \citet{karwa2018correspondence} was the first attempt at tackling the private difference-of-proportions testing problem, and recently \citet{awan2020one} used a novel asymptotic method to calibrate the type I errors of a related DP test in large sample sizes. Our application builds off of these prior works, with a much improved analysis and strong theoretical basis in the $f$-DP framework.

%\citet{awan2021approximate} also 

%The only prior works that has attempted to tackle private difference-of-proportions testing are \citet{karwa2018correspondence} and \citet{awan2021approximate}, but neither was able to accurately control the type I errors in moderate or small sample sizes. 

We show that in general, there does not exist a UMP unbiased $f$-DP test for this problem, but using our earlier results on most powerful $f$-DP tests for binary data, we show that there does exist a UMP unbiased ``semi-private'' test, which satisfies a weakened version of $f$-DP. While this test does not satisfy $f$-DP, it does provide an upper bound on the power of any $f$-DP test, and gives intuition on the structure of a good $f$-DP test for this problem. We then design a novel $f$-DP test for the testing problem, based on using CNDs and an expression of the sampling distribution in terms of characteristic functions, enabling efficient computation via Gil-Pelaez inversion. Using theory of the parametric bootstrap, we argue that the test is asymptotically unbiased and has asymptotically accurate type I errors.  Empirically, we show that the test has more accurate type I errors and $p$-values than the popularly used normal approximation test, and that the power of our proposed test is nearly as powerful as the semi-private UMP unbiased test. In the case of $\ep$-DP, we demonstrate through simulations that our test has higher power than any $(\ep/\sqrt 2)$-DP  test, indicating that it is near optimal. Furthermore,  our test has the benefit of allowing for optimal hypothesis tests and confidence intervals for each of the population proportions, using the techniques of \citet{awan2020differentially}, as the proposed test is based on the same DP summary statistics. 

{\bfseries \large Organization} In Section \ref{s:background} we review background on hypothesis tests and differential privacy. In Section \ref{s:canonical}, we introduce the concept of a \emph{canonical noise distribution}, give some basic properties of CNDs, and provide a general construction of a CND for any $f$-DP privacy notion. In Section \ref{s:HT}, we show that any $f$-DP hypothesis test must satisfy constraints based on the function $f$, we give a general result for ``free'' DP $p$-values given an $f$-DP test function, and develop most powerful $f$-DP tests for exchangeable binary data. In Section \ref{s:application}, we consider the problem of privately testing the difference of population proportions. Specifically in Section \ref{s:semiprivate}, we develop a uniformly most powerful unbiased ``semi-private'' test, which gives an upper bound on the power of any $f$-DP test, in Section \ref{s:inversion} we propose an $f$-DP test based on the inversion of characteristic functions, and in Section \ref{s:simulations} we evaluate the type I error and power of our two sample test in simulations. Proofs and technical details are deferred to the supplementary materials.

{\bfseries \large  Related work}
Private hypothesis testing was first tackled by \citet{vu2009differential}, developing DP tests for population proportions as well as independence tests for $2\times 2$ tables. These tests use additive Laplace noise, and use a normal approximation to the sampling distribution to calibrate the type I errors.  \citet{Solea2014} develop tests for normally distributed data using similar techniques.  \citet{Wang2015:Revisiting} and \citet{Gaboardi2016} expanded on \citet{vu2009differential}, developing additional tests for multinomials. \citet{Wang2015:Revisiting} developed asymptotic sampling distributions for their tests, verifying the type I errors via simulations, whereas  \citet{Gaboardi2016} use Monte Carlo methods to estimate and control the type I error. \citet{Uhler2013} develop DP $p$-values for chi-squared tests of GWAS data, and derive the exact sampling distribution of the noisy statistic. \citet{kifer2016new} develop private $\chi^2$ tests for goodness-of-fit and identity problems which are designed to have the same asymptotic properties as the non-private tests. 

Under ``local differential privacy,'' a notion of DP where even the data curator does not have access to the original dataset, \citet{Gaboardi2018} develop multinomial tests based on asymptotic distributions.

The first uniformly most powerful hypothesis tests under DP for the testing of i.i.d. Bernoulli data were developed by \citet{awan2018differentially}. Their tests were based on the Tulap distribution, an extension of the discrete Laplace and Staircase mechanisms. \citet{awan2020differentially} expanded on these results to offer UMP unbiased two-sided DP tests as well as optimal DP confidence intervals and confidence distributions for Bernoulli data. %The results and methods of this paper are inspired by these last two works. 

Given a DP output, \citet{Sheffet2017} and \citet{Barrientos2019} develop significance tests for regression coefficients.  \citet{wang2018statistical} develop general approximating distributions for DP statistics, which can be used to construct hypothesis tests and confidence intervals, but which are only applicable to limited models. \citet{awan2020one} also provide asymptotic techniques that can be used to conduct approximate hypothesis tests, given DP summary statistics, but which have limited accuracy in finite samples.

Rather than the classical regime of fixing the type I error, and minimizing the type II error, there are several works on DP testing, where the goal is to optimize the sample complexity required to generate a test which places both the type I and type II errors below a certain threshold. \citet{Canonne2019} show that for simple hypothesis tests, a noisy clamped likelihood ratio test achieves optimal sample complexity. \citet{cai2017priv} and \citet{kakizaki2017differentially} both study the problem of $\ep$-DP discrete identity testing from  the sampling complexity perspective. \citet{aliakbarpour2018differentially} also studies $\ep$-DP identitiy testing as well as DP equivalence testing. \citet{acharya2018differentially} study identity and closeness testing of discrete distributions in the $(\ep,\de)$-DP framework. {\citet{bun2019private} derive sample complexity bounds for differentially privacy hypothesis selection, where the goal is to choose among a set of potential data generating distributions, which one has the smallest total variation distance to the true distribution. \citet{suresh2021robust} develop an alternative to the Neyman-Pearson Lemma for simple hypotheses, which is robust to misspecification of the hypotheses; due to the connection between robustness and differential privacy \citep{dwork2009differential}, this could be a promising tool for developing private tests.}

Outside the hypothesis testing setting, there is some additional work on optimal population inference under DP. \citet{Duchi2017} give general techniques to derive minimax rates under local DP, and in particular give minimax optimal point estimates for the mean, median, generalized linear models, and nonparametric density estimation.  \citet{Karwa2017} develop nearly optimal confidence intervals for normally distributed data with finite sample guarantees, which could potentially be inverted to give approximately UMP unbiased tests.

Notable works that develop optimal DP mechanisms for general loss functions are \citet{geng2015optimal} and \citet{ghosh2012universally}, which give mechanisms that optimize  symmetric convex loss functions, centered at a real-valued statistic.
Similarly, \citet{awan2021structure} derive optimal mechanisms among the class of $K$-Norm Mechanisms for a fixed statistic and sample size.

\section{Background}\label{s:background}
In this section, we review some basic notation as well as background on differential privacy. Notation and terminology regarding hypothesis testing is deferred to Appendix \ref{s:testing}.

we say that a real-valued function $f(x)$ is \emph{increasing} (\emph{decreasing}) if $a\leq b$ implies $f(a)\leq f(b)$ (resp. $f(a)\geq f(b)$). We say that $f$ is \emph{strictly increasing} (\emph{strictly decreasing}) if $a<b$ implies $f(a)<f(b)$ (resp. $f(a)>f(b)$). Given an increasing function $f$, we define its inverse to be $f^{-1}(y) = \inf\{x\in \RR\mid y\leq f(x)\}$. For a decreasing function $f$, the inverse is defined to be $f^{-1}(y) = \inf\{x\in \RR\mid y\geq f(x)\}$.

For a real-valued random variable $X$, its \emph{cumulative distribution function} (cdf) is defined as $F_X(t) = P(X\leq t)$, and its \emph{quantile function} is $F^{-1}_X$. A real-valued random variable is \emph{continuous} if its cdf $F_X(t)$ is continuous in $t$, and $X$ is \emph{symmetric about zero} if $F_X(t) = 1-F_X(-t)$. For a random variable $X\sim P$, with cdf $F$, we use $P$ and $F(\cdot)$ interchangeably to denote the distribution of $X$.  %Given a real-valued continuous random variable $X$ with support $[L,U]$, 

\subsection{Differential Privacy}
In this section, we review the definition of $f$-DP which is formulated in terms of constraints on hypothesis tests and relate it to other notions of DP in the literature. A \emph{mechanism} $M$ is a randomized algorithm that takes as input a database $D$, and outputs a (randomized) statistic $M(D)$ in an abstract space $\mscr Y$.  Given two databases $X$ and $X'$ which differ in one person's contribution, a mechanism $M$ satisfies differential privacy if given the output of $M$ it is difficult to determine whether the original database was $X$ or $X'$.

The notion ``differing in one person's contribution'' is often formalized in terms of a metric. In this paper, we  use the Hamming metric, which is defined as follows: For any set $\mscr X$, we write $\mscr X^n = \{(x_1,x_2,\ldots, x_n) \mid x_i \in \mscr X\text{ for all } 1\leq i\leq n\}$. The \emph{Hamming metric} on $\mscr X^n$ is defined by $H(X,X') =\#\{i\mid X_i\neq X'_i\}$. If $H(X,X')\leq 1$, we call $X$ and $X'$ \emph{adjacent databases}. Note that by using the Hamming metric, we assume that the sample size $n$ is a public value and does not require privacy protection. 

%%%   f-DP.

All of the major variants of DP state that given a randomized algorithm $M$, for any two adjacent databases $X$, $X'$, the distributions of $M(X)$ and $M(X')$ should be ``similar.'' While many DP variants measure similarity in terms of divergences, recently \citet{dong2022gaussian} proposed $f$-DP, which formalizes similarity in terms of constraints on hypothesis tests. 

For two probability distributions $P$ and $Q$, the \emph{tradeoff function} $T(P,Q):[0,1]\rightarrow [0,1]$ is defined as $T(P,Q)(\alpha) = \inf \{1-\EE_Q \phi \mid \EE_P(\phi)\leq \alpha\}$, where the infinimum is over all measurable tests $\phi$. The tradeoff function can be interpreted as follows: If $T(P,Q)(\alpha)=\beta$, then the most powerful test $\phi$ which is trying to distinguish between $H_0=\{P\}$ and $H_1:\{Q\}$ at type I error $\leq \alpha$ has type II error $\beta$. A larger tradeoff function means that it is harder to distinguish between $P$ and $Q$. 
Note that the tradeoff function is closely related to the receiver-operator curve (ROC), and captures the difficulty of distinguishing between $P$ and $Q$. A function $f:[0,1]\rightarrow [0,1]$ is a tradeoff function if and only if $f$ is convex, continuous, decreasing, and $f(x) \leq 1-x$ for all $x \in [0,1]$ \citep[Proposition 2.2]{dong2022gaussian}.  We say that a tradeoff function $f$ is \emph{nontrivial} if $f(\alpha)<1-\alpha$ for some $\alpha\in (0,1)$; that is if $f$ is not identically equal to $1-\alpha$. 

  \begin{figure}
      \centering
      \includegraphics[width=.5\linewidth]{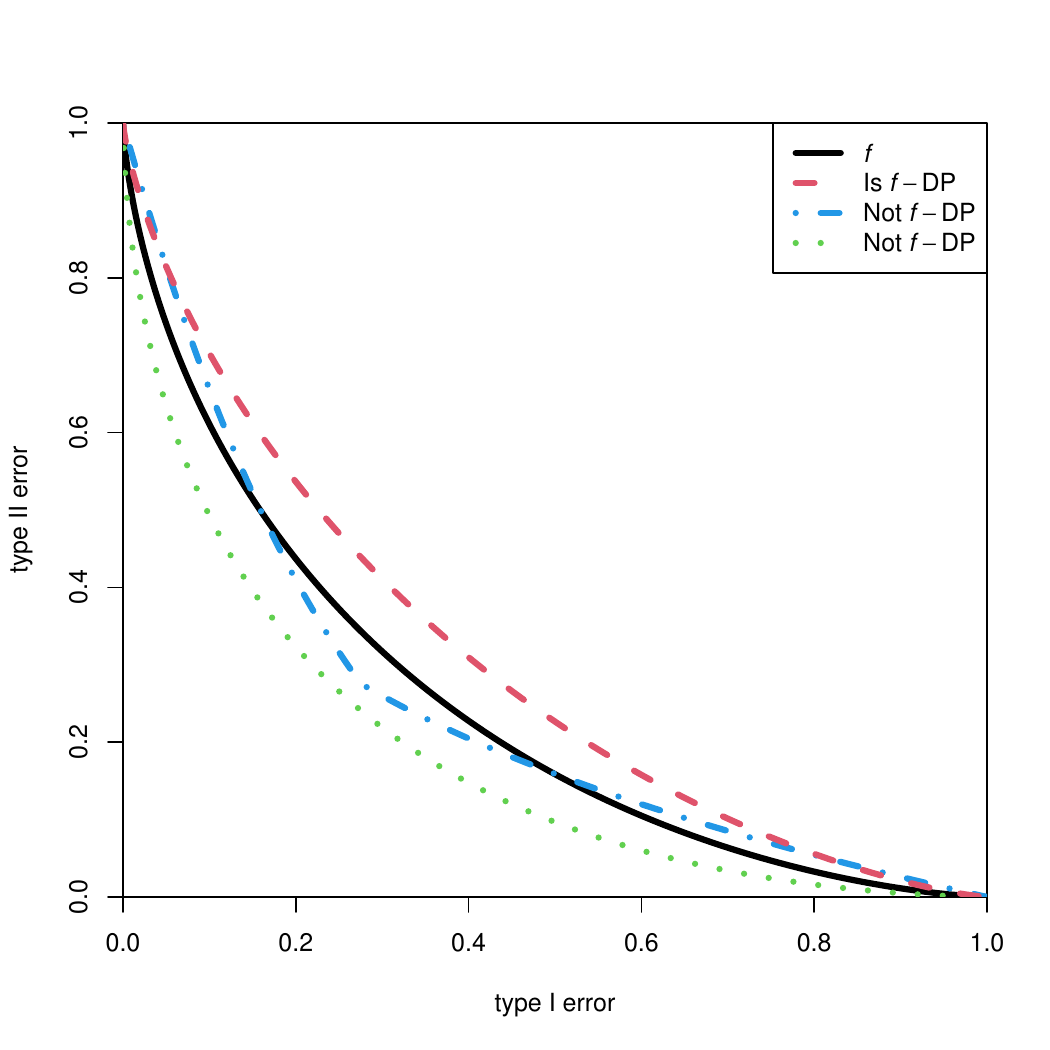}
      \caption{A plot of three examples of $T(M(D),M(D'))$. Only the red, dashed tradeoff curve satisfies $f$-DP.}
      \label{fig:tradeoff}
  \end{figure}

\begin{defn}[$f$-DP: \citealp{dong2022gaussian}]\label{def:fDP}
 Let $f$ be a tradeoff function. A mechanism $M$ satisfies $f$-DP if \[T(M(D),M(D'))\geq f\] for all $D,D'\in \mscr X^n$ such that $H(D,D')\leq 1$.
\end{defn}

See Figure \ref{fig:tradeoff} for examples of tradeoff functions which do and do not satisfy $f$-DP for a particular $f$. In the above definition, the inequality $T(M(D),M(D'))\geq f$ is shorthand for $T(M(D),M(D'))(\alpha)\geq f(\alpha)$ for all $\alpha\in [0,1]$. Without loss of generality we can assume that $f$ is symmetric: $f(\alpha) = f^{-1}(\alpha)$, where $f^{-1}(\alpha) = \inf \{t\in [0,1]\mid f(t)\leq \alpha\}$. This is due to the fact that adjacency of databases is a symmetric relation \citep[Proposition 2.4]{dong2022gaussian}. %\citet[Proposition 2.4]{dong2022gaussian}, which states that for a given $f$ and a  mechanism $M$ that is $f$-DP, there exists a symmetric $f^*\geq f$ such that $M$ is $f^*$-DP. 
For the remainder of the paper, we assume that $f$-DP also requires this symmetry.

\citet{Wasserman2010:StatisticalFDP} and \citet{Kairouz2017} both showed that $(\ep,\de)$-DP can be expressed in terms of hypothesis testing, and in fact \citet{dong2022gaussian} showed that $(\ep,\de)$-DP can be expressed as a special case of $f$-DP. 

\begin{defn}[$(\ep,\de)$-DP: \citealp{dwork2006calibrating}]\label{def:epDeDP}
Let $\ep>0$ and $\de\geq 0$, and define $f_{\ep,\de}(\alpha) = \max\{0,1-\de-\exp(\ep)\alpha,\exp(-\ep)(1-\de-\alpha)\}$. Then we say that a mechanism $M$ satisfies \emph{$(\ep,\de)$-DP} if it satisfies $f_{\ep,\de}$-DP.
\end{defn}

Another notable special case of $f$-DP is Gaussian DP ($\mu$-GDP). \citet{dong2022gaussian} showed that $\mu$-GDP is perhaps the most natural single parameter privacy definition, due to the central limit theorem for composition. %, as it is the only single parameter family that is loss-lessly closed under composition. 
Gaussian DP is closely related to zero-concentrated differential privacy (zCDP) \citep{bun2016concentrated}, a very popular relaxation of DP. GDP is slightly stronger than zCDP in that a mechanism satisfying GDP satisfies zCDP \citep[Corollary B.6]{dong2022gaussian}, but the converse is not true \citep[Proposition B.7]{dong2022gaussian}.

\begin{defn}[Gaussian differential privacy: \citealp{dong2022gaussian}]\label{def:GDP}
Let $\mu>0$ and define 
\[G_{\mu}(\alpha) = T(N(0,1),N(\mu,1))(\alpha) = \Phi(\Phi^{-1}(1-\alpha)-\mu),\] where $\Phi$ is the cdf of $N(0,1)$. We say that a mechanism $M$ satisfies $\mu$-\emph{Gaussian differential privacy} ($\mu$-GDP) if it is $G_\mu$-DP.
\end{defn}

%Gaussian DP (and the closely related zCDP) offer an important qualitative benefit over the original $(\ep,\de)$-DP, which is that the privacy guarantee offered by Gaussain noise cannot be properly accounted for in the $(\ep,\de)$-DP framework. Furthermore, noise distributions which are 
%the noise distributions typically used in $(\ep,\de)$-DP often result in statistically unnatural noise, or when Gaussian noise is used it can not be perfectly communicated with a single pair of $\ep$ and $\de$ values. However Guassian DP perfectly characterizes the privacy cost of adding Gaussian noise, and using Gaussian noise is much more convenient for statisticians. 

%%%%%%%%%%%%%%%%%%%%%%%%%%%%%%%%%%%%%%%%%%%%
%%%
%%%%%%%%%%%%%%%%%%%%%%%%%%%%%%%%%%%%%%%%%%%%
\section{Canonical noise distributions}\label{s:canonical}

One of the most basic techniques of designing a privacy mechanism is through adding data-independent noise. The earliest DP mechanisms add either Laplace or Gaussian noise, and there have since been several works developing optimal additive mechanisms including the geometric (discrete Laplace) \citep{ghosh2012universally}, truncated-uniform-Laplace (Tulap) \citep{awan2018differentially,awan2020differentially}, and staircase mechanisms \citep{geng2015optimal}. There have also been several works exploring multivariate and infinite-dimensional additive mechanisms such as $K$-norm \citep{hardt2010geometry,awan2021structure}, elliptical perturbations \citep{reimherr2019elliptical}, and Gaussian processes \citep{hall2013differential,mirshani2019formal}.

While there are many choices of additive mechanisms to achieve $f$-DP, we are interested in adding the least noise necessary in order to maximize the utility of the output. Rather than measuring the amount of noise by its variance or entropy, we focus on whether the privacy guarantee is tight. 

In this section, we introduce the concept \emph{canonical noise distribution} (CND), which captures whether a real-valued distribution is perfectly tailored to satisfy $f$-DP. We formalize this in Definition \ref{def:CND}. We then show that for any symmetric $f$, we can always construct a CND, where the construction is given in Definition \ref{def:CNDsynthetic} and proved to be a CND in Theorem \ref{thm:canonical}. We will see in Section \ref{s:HT} that CNDs are fundamental for understanding the nature of $f$-DP hypothesis tests, for constructing ``free'' DP $p$-values, and for the design of uniformly most powerful $f$-DP tests for binary data. We also see in Section \ref{s:application} that CNDs are central to our application of difference-of-proportions tests as well. 

Before we define canonical noise distribution, we must introduce the \emph{sensitivity} of a statistic, a central concept of DP \citep{dwork2006calibrating}. A statistic $T:\mscr X^n\rightarrow \RR$ has \emph{sensitivity} $\Delta>0$ if $|T(X)-T(X')|\leq \Delta$ for all $H(X,X')\leq 1$. As the sensitivity measures how much a statistic can change when one person's data is modified, additive noise must be scaled proportionally to the sensitivity in order to protect privacy. 

\begin{defn}\label{def:CND}
  Let $f$ be a symmetric nontrivial tradeoff function. A {continuous} distribution function $F$ is a \emph{canonical noise distribution} (CND) for $f$ if 
  \begin{enumerate}
      \item {for every statistic $S:\mscr X^n\rightarrow \RR$} with sensitivity $\Delta>0$, and $N\sim F(\cdot)$, the mechanism $S(X) + \Delta N$ satisfies $f$-DP. Equivalently, for every $m\in [0,1]$, $T(F(\cdot),F(\cdot-m))\geq f$,
      \item $f(\alpha)=T(F(\cdot),F(\cdot-1))(\alpha)$ for all $\alpha \in (0,1)$,
      \item $T(F(\cdot),F(\cdot-1))(\alpha) = F(F^{-1}(1-\alpha)-1)$ for all $\alpha \in (0,1)$,
      \item $F(x) = 1-F(-x)$ for all $x\in \RR$; that is, $F$ is the cdf of a random variable which is symmetric about zero.
  \end{enumerate}
\end{defn}

The most important conditions of Definition \ref{def:CND} are 1 and 2, which state that the distribution can be used to satisfy $f$-DP and that the privacy bound is tight. {For property 1, the value $m$ can be interpreted as the quantity $|S(X)-S(X')|/\Delta$; then by the symmetry of $F$, it can be seen that $T(S(X)+\Delta N,S(X')+\Delta N)=T(F(\cdot),F(\cdot-m))$.} Condition 3 of Definition \ref{def:CND} gives a closed form for the tradeoff function, and is equivalent to requiring that the optimal rejection set for discerning between $F(\cdot)$ and $F(\cdot-1)$ is of the form $(x,\infty)$ for some $x\in \RR$. The last condition of Definition \ref{def:CND} enforces symmetry of the distribution, which makes CNDs much easier to work with. 

Finally note that conditions 1 and 2 are not equivalent. Adding excessive noise would satisfy 1, but not 2, whereas a mechanism which fails $T(F(\cdot),F(\cdot-m))\geq T(F(\cdot),F(\cdot-1))$ for some $m\in (0,1)$ would not satisfy property 1. The following example illustrates both cases.

\begin{example}
Consider the discrete Laplace mechanism, which has cdf $F(t) = \frac{1-b}{1+b}b^{|t|}$ for $t\in \ZZ$ and $b\in (0,1)$. Then it can be verified that the discrete Laplace distribution with $b=\exp(-\ep)$ satisfies $T(F(\cdot),F(\cdot-1))=f_{\ep,0}$, but not part 1 of Definition \ref{def:CND}. For example, if $S(X)=0$ and $S(X')=.1$, adding discrete Laplace noise {$N\sim F$ results in distributions with disjoint support, since 
%allows perfect discernment of $X$ and $X'$, {since $S(X)+N$ and $S(X')+N$ have disjoint support: 
$S(X)+N$ takes values in $\ZZ$, whereas $S(X')+N$ takes values in $\ZZ+.1$. As the supports of the distributions are disjoint, we can have zero type I and type II error when testing between $X$ and $X'$, violating the $f_{\ep,0}$ bound.}

It is well known that the continuous Laplace mechanism with scale parameter $\Delta/\ep$ satisfies $\ep$-DP, when added to a $\Delta$-sensitivity statistic, and so satisfies property 1 of Definition \ref{def:CND} for $f_{\ep,0}$. However, as \citet{dong2022gaussian} noted, it can be verified that the Laplace distribution does not satisfy property 2 of Definition \ref{def:CND}, as there exists $\alpha\in (0,1)$ such that the tradeoff function is strictly greater than $f_{\ep,0}$ at $\alpha$.
\end{example}

{
\begin{remark}
Note that property 2 of Definition \ref{def:CND} captures the intuition that a privacy mechanism should match the tradeoff function in the privacy guarantee to avoid introducing excessive noise. While this is indeed an intuitive idea, this has never previously been formalized into a precise criterion for a privacy mechanism, as we do in Definition \ref{def:CND}. Furthermore, no prior work has attempted to build a mechanism that matches the tradeoff function for an arbitrary $f$-DP guarantee. In Theorem \ref{thm:canonical}, we not only prove that a CND exists, but give a construction to build a CND for every $f$. 
\end{remark}}

\begin{example}
[CND for GDP]
The distribution $N(0,1/\mu)$, which has cdf $\Phi(1/\mu)$ ($\Phi$ is the cdf of a standard normal) is a CND for $G_\mu$, defined in Definition \ref{def:GDP}. Property 1 is proved in \citep{dong2022gaussian}, properties 2 and 3 are easily verified, and the distribution is obviously symmetric. \citet{dong2022gaussian} state that ``GDP precisely characterizes the Gaussian mechanism.'' From the opposite perspective, we argue that this is because the normal distribution is a CND for $G_\mu$.
\end{example}

\begin{prop}\label{prop:cnd}
 Let $f$ be a symmetric nontrivial tradeoff function. Let $F$ be a CND for $f$, and $G$ be another cdf such that $T(G(\cdot),G(\cdot-1))\geq f$. Let $N\sim F$ and $M\sim G$. Then there exists a randomized function $\mathrm{Proc}:\RR\rightarrow \RR$ which satisfies $\mathrm{Proc}(N)\overset d=M$ and $\mathrm{Proc}(N+1)\overset d=M+1$, where ``$\overset d=$'' means \emph{equal in distribution}.
\end{prop}
Proposition \ref{prop:cnd} follows from property 2 in Definition \ref{def:CND} along with \citet[Theorem 2.10]{dong2022gaussian}, which is based on Blackwell's Theorem \citep{blackwell1950comparison}. Proposition \ref{prop:cnd} shows that if we add noise from a CND to a statistic $S(X)$ versus $S(X)+1$, we can post-process the result to obtain the same result as if we added noise from another distribution that achieves $f$-DP. This shows in some sense that a CND adds the least noise necessary to achieve $f$-DP. {Note that Proposition \ref{prop:cnd} does not imply that a CND is optimal in every sense: for example, \citet{geng2015optimal} derived the minimum variance additive $(\ep,0)$-DP mechanism, which is not a CND for $f_{\ep,0}$. We will see in Section \ref{s:HT} that the properties of Definition \ref{def:CND} do lead to optimal properties of DP hypothesis tests.}

In the remainder of this section, we show that given any tradeoff function $f$, we can always construct a canonical noise distribution (CND), but that a CND need not be unique.

\begin{restatable}{lem}{lemrecurrence}\label{lem:recurrence}
  Let $f$ be a symmetric nontrivial tradeoff function and let $F$ be a CND for $f$. Then $F(x)=1-f(F(x-1))$ when $F(x-1)>0$ and $F(x)=f(1-F(x+1))$ when $F(x+1)<1$. 
\end{restatable}
\begin{proof}[Proof sketch.]
The result follows from properties 2, 3, and 4 of Definition \ref{def:CND} along with some algebra of cdfs. 
\end{proof}

In the Lemma \ref{lem:recurrence}, we see that a CND satisfies an interesting recurrence relation. If we know the value $F(x)=c$ for some $x\in \RR$ and $c\in (0,1)$, then we know the value of $F(y)$ for all $y \in \ZZ+x$. This means that if we specify $F$ on an interval of length 1, such as $[-1/2,1/2]$, then $F$ is completely determined by the recurrence relation. While there are many choices to specify $F$ on $[-1/2,1/2]$, each of which may or may not lead to a CND. We show that using a particular linear function in $[-1/2,1/2]$ does indeed give a CND. The remainder of this section is devoted to this construction of a CND and the proof that it has the properties of Definition \ref{def:CND}.

\begin{defn}\label{def:CNDsynthetic}
  Let $f$ be a symmetric nontrivial tradeoff function, and let {$c\in [0,1/2)$} be the unique fixed point of $f$: $f(c)=c$. We define $F_f:\RR\rightarrow \RR$ as 
  \[ F_f(x) = \begin{cases}
  f(1-F_f(x+1))&x<-1/2\\
  c(1/2-x) + (1-c)(x+1/2)&-1/2\leq x\leq 1/2\\
  1-f(F_f(x-1))&x>1/2.\\
  \end{cases}\]
\end{defn}

In Definition \ref{def:CNDsynthetic}, the fact that there is a unique fixed point follows from the fact that $f$ is convex and decreasing, and so intersects the line $y=\alpha$ at a unique value. In Lemma \ref{lem:contraction} we establish that the fixed point $c$ lies in the interval $[0,1/2)$. Note that in Definition \ref{def:CNDsynthetic}, the cdf corresponds to a uniform random variable  on the interval $[-1/2,1/2]$, but due to the recursive nature of $F_f$ and the fact that $f$ is in general non-linear, the CND of Definition \ref{def:CNDsynthetic} need not be uniformly distributed on any other intervals. See Figure \ref{fig:cnd} for a plot of the pdf and cdf of the CND of Definition \ref{def:CNDsynthetic} corresponding to the tradeoff function $G_1$. 

The following proposition verifies that $F_f$ is a distribution function, as well as some other properties, such as continuity, symmetry, and concavity/convexity.

\begin{restatable}{prop}{proptradeoff}\label{prop:tradeoff}
  Let $f$ be a symmetric nontrivial tradeoff function, and let $F\defeq F_f$. Then 
  \begin{enumerate}
  \item $F(x)$ is a cdf for a symmetric, continuous, real-valued random variable,% That is,
  %\begin{enumerate}
  %    \item $F(x)$ takes values in $[0,1]$,
  %    \item $F(x) = 1-F(-x)$,
  %    \item $F(x)$ is continuous,
  %    \item $F'(x)$ is defined almost everywhere, and $F(x)$ is increasing,
  %    \item $\lim_{x\rightarrow \infty} F(x)=1$ $\lim_{x\rightarrow -\infty} F(x)=0$.
  %\end{enumerate}
  \item $F(x)$ satisfies $F(x) = 1-f(F(x-1))$ whenever $F(x-1)>0$ and $F(x)=f(1-F(x+1))$ whenever $F(x+1)<1$.
  \item $F'(x)$ is decreasing on $(-1/2,\infty)$ and increasing on $(-\infty,1/2)$,
  \item $F(x)$ is strictly increasing on $\{x\mid 0<F(x)<1\}$.
  \end{enumerate}
  \end{restatable}
   \begin{proof}[Proof sketch.]
   Most of the properties are proved by induction, checking that the properties hold on intervals of the type $[x-1/2,x+1/2]$ for $x\in \ZZ$ as well as at the break points at half-integer values.  The full proof is found in Appendix \ref{s:proofs}.
   \end{proof}

  Theorem \ref{thm:canonical} below states that for any nontrivial tradeoff function, this construction yields a canonical noise distribution, which can be constructed as in Definition \ref{def:CNDsynthetic}. This CND can be used to add perfectly calibrated noise to a statistic to achieve $f$-DP. As we will see later, the existence (and construction) of a CND will enable us to prove that any $f$-DP test can be post-processed from a test statistic, and this implies that we can always obtain hypothesis testing $p$-values at no additional privacy cost, a generalization of the result of \citet{awan2018differentially} which previously only held for $(\ep,\de)$-DP and for Bernoulli data.
  
  \begin{restatable}{thm}{thmcanonical}
  \label{thm:canonical}
    Let $f$ be a symmetric nontrivial tradeoff function and let $F_f$ be as in Definition \ref{def:CNDsynthetic}. Then $F_f$ is a canonical noise distribution for $f$. %That is,
    %\begin{enumerate}
    %\item  $f(\alpha) = T(F_f(\cdot),F_f(\cdot-1))=F_f(F_f^{-1}(1-\alpha)-1)$
    %\item Let $S(X)$ be a real-valued statistic with sensitivity $\Delta$, and let $N\sim F_f(\cdot)$. Then $S(X)+\Delta N$ satisfies $f$-DP. 
    %\end{enumerate}
  \end{restatable}
  \begin{proof}[Proof sketch.]
  $F_f$ was already shown to be symmetric in Proposition \ref{prop:tradeoff}. The two equalities, $f(\alpha)=T(F(\cdot),F(\cdot-1))(\alpha)=F(F^{-1}(1-\alpha)-1)$ can also be easily verified using the properties of Proposition \ref{prop:tradeoff}. The main challenge is to show that $T(F(\cdot),F(\cdot-m))\geq T(F(\cdot),F(\cdot-1))$ for $m\in (0,1)$. Lemma \ref{lem:symTrade} in the appendix gives an alternative technical condition which makes it easier to verify property 1 of Definition \ref{def:CND}.
  \end{proof}

  %The main properties of a canonical noise distribution is that it 1) satisfies the recurrence relation (lines 2 and 3 of Definition \ref{def:CNDsynthetic}), and that the trade-off $T(F(\cdot),F(\cdot-m))$ function is decreasing in $m\in [0,1]$. 
  It turns out that the requirements of Definition \ref{def:CND} do not uniquely determine a distribution. For instance, $\Phi$ the cdf of a standard normal is a CND for $1$-GDP, but $\Phi$ is different from the construction in Definition \ref{def:CNDsynthetic}. See Figure \ref{fig:cnd} for the cdf and pdf of these two CNDs. Note that the CND of Definition \ref{def:CNDsynthetic} is uniform in $[-1/2,1/2]$ and has ``kinks'' at each half-integer value. On the other hand, the standard normal is smooth.   This example shows that for certain tradeoff functions there may be a more natural CND than the one constructed in Definition \ref{def:CNDsynthetic}. 
  
  {While there may be more natural CNDs in some settings, we emphasize the generality of the construction in Definition \ref{def:CNDsynthetic}. In Proposition \ref{prop:samplingCND}, we present an exact method to sample from the CND of Definition \ref{def:CNDsynthetic} based on inverse transform sampling, 
  %an expression of the quantile function of the CND $F_f$ constructed in Definition \ref{def:CNDsynthetic}, 
  allowing for straightforward implementation and application of our CND results.}
  
  %\begin{remark}[Sampling the CND]
  %In Appendix \ref{s:proofs}, Proposition \ref{prop:samplingCND} specifies the quantile function for the CND constructed in Definition \ref{def:CNDsynthetic}, which can be efficiently computed in a finite number of recursive steps. Using this quantile function for inverse transform sampling, we can easily sample directly from the CND $F_f$.    
  %\end{remark}
  
  \begin{figure}
      \centering
          \begin{subfigure}[t]{0.48\linewidth}
       \includegraphics[width=\linewidth]{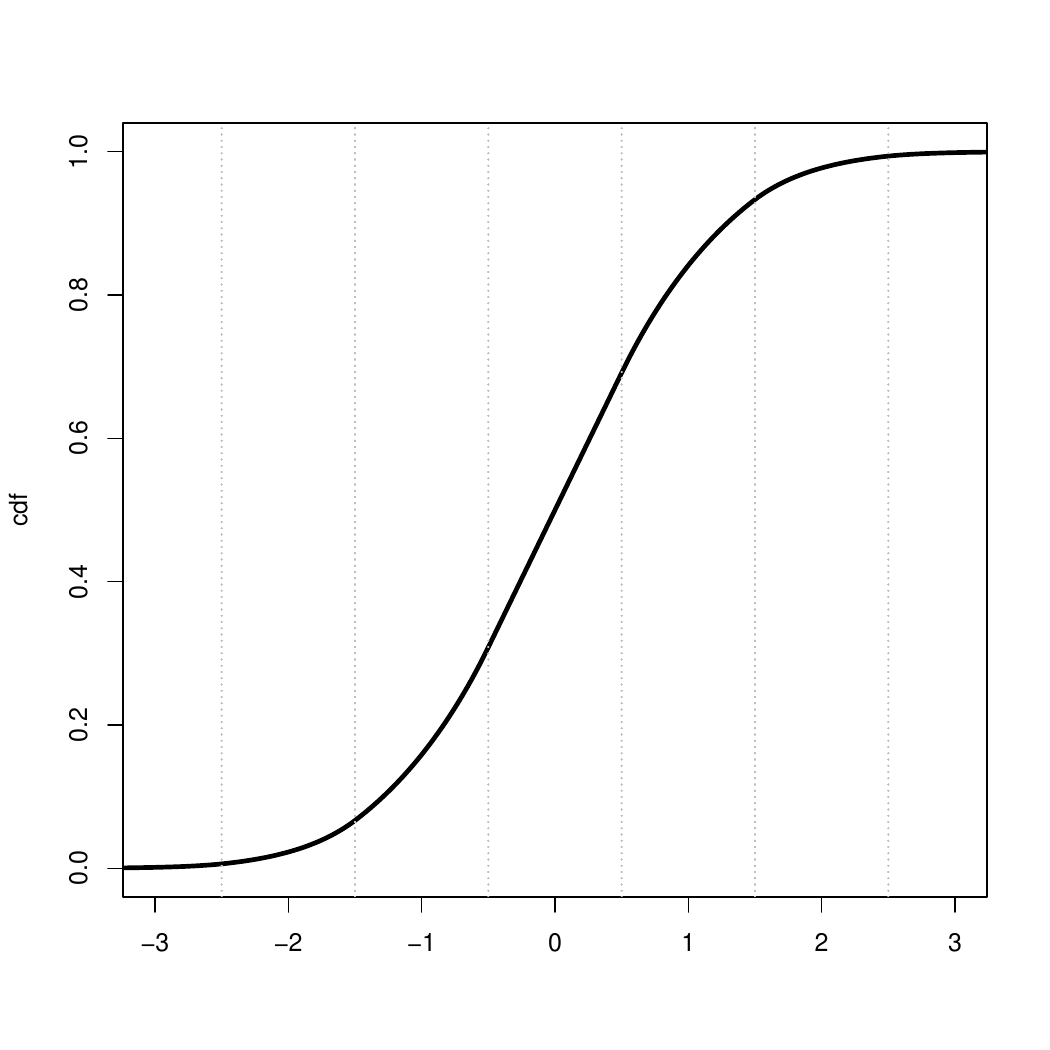}
      \caption{Plot of cdf of the CND of Definition \ref{def:CNDsynthetic} corresponding to $G_1$. The function is linear between -1/2 and 1/2.}
      \end{subfigure}\hspace{.02\linewidth}
		\begin{subfigure}[t]{0.48\linewidth}
      \includegraphics[width=\linewidth]{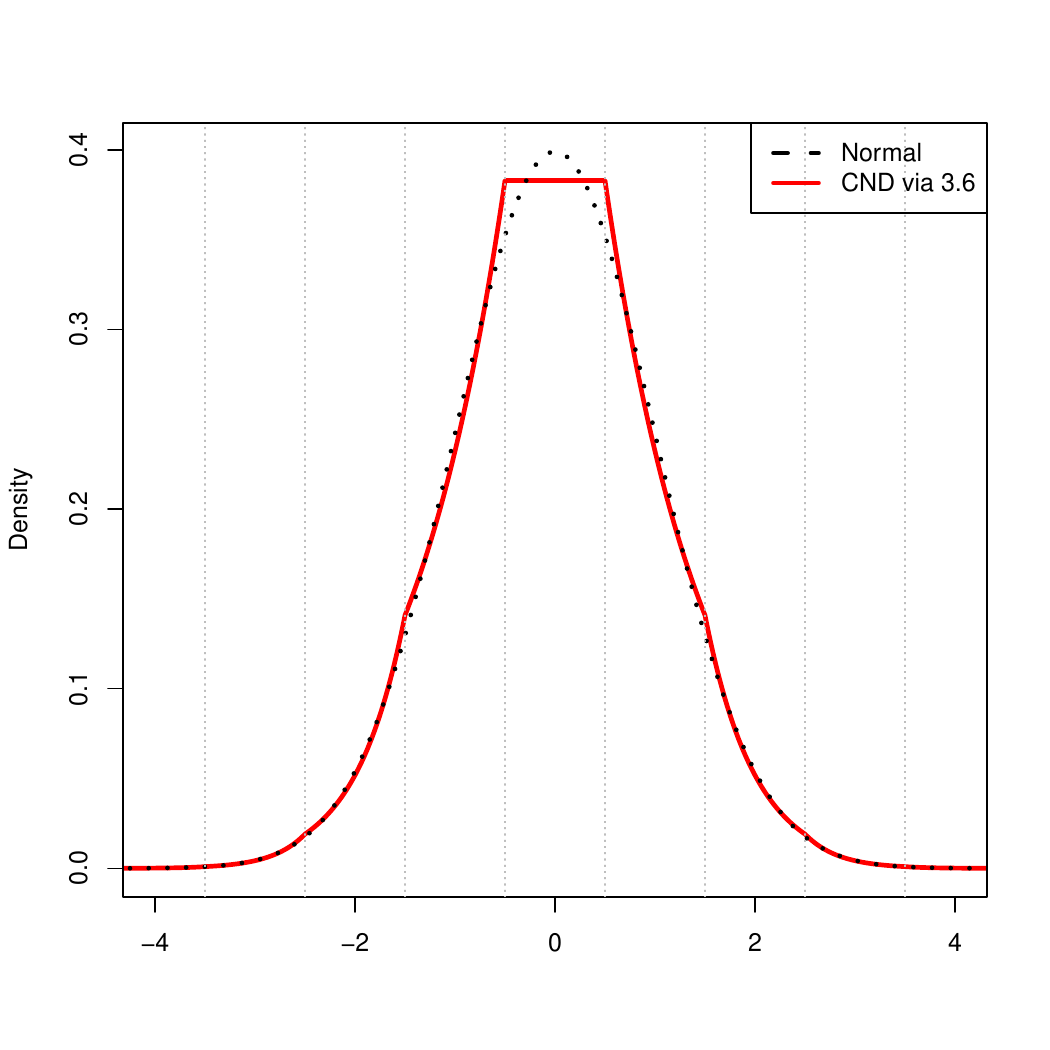}
      \caption{Density plots of $N(0,1)$ as well as the CND of Definition \ref{def:CNDsynthetic} for the tradeoff function $G_1$. }
      \end{subfigure}
\caption{Plots of CND construction of Definition \ref{def:CNDsynthetic}. The vertical lines are at half-integer values. }
      \label{fig:cnd}
  \end{figure}
  
  \subsection{Canonical noise for \texorpdfstring{$(\ep,\de)$}{TEXT}-DP}
  So far, we have developed a constructive and general method of generating canonical noise distributions for $f$-DP. In the special case of $(\ep,\de)$-DP, the CND $F_f$ is equal to the cdf of the Tulap distribution, proposed in \citet{awan2018differentially}, which is an extension of the Staircase mechanism \citep{geng2015optimal} from $(\ep,0)$-DP to $(\ep,\de)$-DP. %The Tulap and Staircase mechanisms are also related to the geometric mechanism (discrete Laplace), which was shown to be universally optimal for count statistics in $(\ep,0)$-DP \citep{ghosh2012universally}. 
  
  \begin{restatable}{cor}{cortulap}\label{cor:tulap}
  The distribution $\mathrm{Tulap}(0,b,q)$, where $b=\exp(-\ep)$ and $q = \frac{2\de b}{1-b+2\de b}$ is a CND for $f_{\ep,\de}$-DP, which agrees with the construction of Definition \ref{def:CNDsynthetic}. 
  \end{restatable}
  \begin{proof}[Proof sketch.]
  The cdf of $\mathrm{Tulap}(0,b,q)$ is defined in the full proof. From the definition, it is easy to verify that the cdf of a Tulap random variable agrees with $F_f$ on $[-1/2,1/2]$. By \citet[Lemma 2.8]{awan2020differentially}, the Tulap cdf also satisfies the recurrence relation of Definition \ref{def:CNDsynthetic}. 
  \end{proof}
  
  It was claimed in both \citet{awan2018differentially} and \citet{awan2020differentially} that adding Tulap noise satisfied $(\ep,\de)$-DP, but their proof is actually incorrect and only holds for integer valued statistics. The above Corollary along with Theorem \ref{thm:canonical} offers a complete and correct argument for \citet[Theorem 2.11]{awan2020differentially}. 
  
  In \citet{awan2018differentially} and \citet{awan2020differentially}, it was shown that the Tulap distribution could be used to design optimal hypothesis tests and confidence intervals for Bernoulli data. Our notion of a canonical noise distribution, and the fact that Tulap is a CND for $(\ep,\de)$-DP sheds some light on why it had such optimality properties (even further explored in Section \ref{s:HT}). The Tulap distribution is also closely related to discrete Laplace and the Staircase distributions, which were shown by \citet{ghosh2012universally} and \citet{geng2015optimal} respectively to be optimal in terms of maximizing various definitions of utility in $(\ep,0)$-DP. 

 While continuous Laplace noise is commonly used in $(\ep,0)$-DP, \citet{dong2022gaussian} pointed out that the tradeoff function for Laplace noise does not agree with $f_{\ep,\de}$ for any values of $\ep$ and $\de$. From this observation, we conclude from Definition \ref{def:CND} that Laplace is not a CND for $(\ep,\de)$-DP. From the perspective of CNDs, Tulap noise is preferable over the Laplace mechanism.

%%%%%%%%%%%%%%%%%%%%%%%%%%%%%%%%%%%%%%%%%%%%%%%%%%
%%%   Hypothesis Tests
%%%%%%%%%%%%%%%%%%%%%%%%%%%%%%%%%%%%%%%%%%%%%%%%%%
\section{The nature of \texorpdfstring{$f$}{TEXT}-DP tests}\label{s:HT}

Recall that a test is a function $\phi: \mscr X^n\rightarrow [0,1]$, where $\phi(x)$  represents the probability of rejecting the null hypothesis given that we observed $x$. However, the mechanism corresponding to this test releases a random value drawn as $\mathrm{Bern}(\phi(x))$, where 1 represents ``Reject'' and 0 represents ``Accept.'' we say that the test $\phi$ satisfies $f$-DP if the corresponding mechanism $\mathrm{Bern}(\phi(x))$ satisfies $f$-DP. Intuitively, Lemma \ref{lem:DPtest} shows that a test satisfies $f$-DP if for adjacent databases $x$ and $x'$, the values $\phi(x)$ and $\phi(x')$ are close in terms of an inequality based on $f$. 

\begin{restatable}{lem}{lemDPtest}\label{lem:DPtest}
  Let $f$ be a symmetric tradeoff function. A test $\phi: \mscr X^n \rightarrow [0,1]$ satisfies $f$-DP if and only if $\phi(x) \leq 1-f(\phi(x'))$ for all $x,x'\in \mscr X^n$ such that $H(x,x')\leq 1$.
\end{restatable}
\begin{proof}[Proof sketch.]
If we take the rejection region to be the set $\{1\}$ then $\phi(x)$ is the type I error and $1-\phi(x')$ is the type II error. The $f$-DP guarantee requires that $f(\phi(x))\leq 1-\phi(x')$, or equivalently, $\phi(x')\leq 1-f(\phi(x))$. Using the rejection region $\{0\}$ and some algebra, we get $\phi(x)\leq 1-f(\phi(x'))$. The full proof argues more precisely using the Neyman Pearson Lemma, considering also randomized tests. 
\end{proof}

Lemma \ref{lem:DPtest} greatly simplifies the search for $f$-DP hypothesis tests and generalizes the bounds on private tests established in \citet{awan2018differentially}.

\begin{example}[$(\ep,\de)$-DP tests]
When we apply Lemma \ref{lem:DPtest} to the setting of $(\ep,\de)$-DP, we have the two inequalities: $(1-\phi(x))\geq 1-\de-\exp(\ep) \phi(x')$ and $(1-\phi(x))\geq \exp(-\ep)(1-\de-\phi(x'))$. Some algebra gives 
\[\phi(x) \leq \begin{cases}
\de + \exp(\ep)\phi(x')\\
1-\exp(-\ep) (1-\de-\phi(x')),
\end{cases}\]
which agrees with the constraints derived in \citet{awan2018differentially}.
\end{example}

The result of Lemma \ref{lem:DPtest} can also be expressed in terms of canonical noise distributions in Corollary \ref{cor:CND_HT}, giving the elegant relation that $F^{-1}(\phi(x))$ and $F^{-1}(\phi(x'))$ differ by at most 1 when $x$ and $x'$ are adjacent.

\begin{restatable}[Canonical Noise Distributions]{cor}{corCNDHT}\label{cor:CND_HT}
Let $f$ be a symmetric nontrivial tradeoff function and let $F$ be a canonical noise distribution for $f$. Then a test $\phi$ satisfies $f$-DP if and only if  $F^{-1}(\phi(x)) \leq F^{-1}(\phi(x'))+1$ for all $x,x'\in \mscr X^n$ such that $H(x,x')\leq 1$. 
\end{restatable}
\begin{proof}[Proof sketch.]
The result follows from the fact that $f(\alpha)= F(F^{-1}(1-\alpha)-1)$, the symmetry of $F$, and some algebra of cdfs.
\end{proof}

Corollary \ref{cor:CND_HT} is also important for the construction of ``free'' DP $p$-values in Section \ref{s:free}.

\begin{comment}
\begin{example}[log-concave and symmetric $f$]
Let $P$ be a probability distribution on $\RR$ with cdf $F$, quantile function $F^{-1}$ and log-concave density $f$,
which is symmetric about zero. Let $\xi \sim P$. Define $f(\alpha) = T(\xi,\xi+t)(\alpha) = F(F^{-1}(1-\alpha)-t)$ (see \citet[citation]{dong2022gaussian} for proof).
Then the constraints on $f$-DP tests can be expressed as $F^{-1}(\phi(x))\leq F^{-1}(\phi(x')) + t$ for all $H(x,x')\leq 1$. 
\end{example}

\begin{example}[GDP tests]
As a special case of a log-concave and symmetric density, the constraints on a $\mu$-GDP test can be expressed as 
\begin{align*}
% \phi(x)&\leq1-\Phi(\Phi^{-1}(1-\phi(x'))-\mu)\\
%\phi(x)&\leq \Phi(\Phi^{-1}(\phi(x'))+\mu)\\
\Phi^{-1}(\phi(x))&\leq \Phi^{-1}(\phi(x'))+\mu,
\end{align*}
where the last expression shows that $\Phi^{-1}(\phi(x))$ changes by at most $\mu$ when one person's data is changed. 
\end{example}
\end{comment}

  \subsection{Free \texorpdfstring{$f$}{TEXT}-DP \texorpdfstring{$p$}{TEXT}-values}\label{s:free}
  In \citet{awan2018differentially}, it was shown that for Bernoulli data, the uniformly most powerful DP test could also be expressed as the post-processing of a privatized test statistic, offering $p$-values at no additional privacy cost. We generalize this result using the concept of canonical noise distributions and show that any $f$-DP test can be expressed as a post-processing threshold test based on a privatized test statistic, and that the test statistic can also be used to give private $p$-values. 
  
  Typically in statistics, it is preferred to report a $p$-value rather than an accept/reject decision at a single type I error. A $p$-value provides a continuous summary of how much evidence there is for the alternative hypothesis and allows for the reader to determine whether there is enough evidence to reject at the reader's personal type I error. Lower $p$-values give more evidence for the alternative hypothesis.
  
  However, with privacy, one may wonder whether releasing a $p$-value rather than just the accept/reject decision would result in an increased privacy cost, or conversely whether a $p$-value at the same privacy level would have lower power. In fact, this question is related to fundamental concepts in differential privacy such as post-processing, privacy amplification, and composition. In Lemma \ref{lem:post}, we recall the post-processing property of DP, which states that after a DP result is released, no post-processing can compromise the DP guarantee. 
  
  \begin{lem}
    [Post-processing: \citealp{dong2022gaussian}]\label{lem:post}
    Let $M$ be an $f$-DP mechanism taking values in $\mscr Y$. Let $\mathrm{Proc}$ be a mechanism from $\mscr Y$ to $\mscr Z$. Then $\mathrm{Proc}\circ M$ satisfies $f$-DP. 
  \end{lem}
  
  Theorem \ref{thm:pVal} is the main result of this section, demonstrating that given an arbitrary $f$-DP hypothesis test, we can construct a summary statistic and $p$-values, all with no additional privacy cost, using a CND.
  
  \begin{restatable}{thm}{thmpVal}\label{thm:pVal}
    Let $\phi:\mscr X^n \rightarrow [0,1]$ be an $f$-DP test. Let $F$ be a CND for $f$, and draw $N\sim F$. Then
    \begin{enumerate}
        \item releasing $T= F^{-1}(\phi(x))+N$ satisfies $f$-DP,
        \item  the variable $Z=I(T\geq 0)$, a post-processing of $T$, is distributed as $Z\mid X=x\sim \mathrm{Bern}(\phi(x))$,
        \item the value $p = \sup_{\theta_0\in H_0}\EE_{X\sim \theta_0} F(F^{-1}(\phi(X))-T)$ is also a post-processing of $T$ and is a $p$-value for $H_0$, %(This may not be as powerful as $\phi$). 
        \item if $H_0$ is a simple hypothesis and $\EE_{H_0}\phi=\alpha$, then at type I error $\alpha$, the $p$-value from part 3 is as powerful as $\phi$ at every alternative.
    \end{enumerate}
  \end{restatable}
  \begin{proof}[Proof sketch.]
  Property 1 follows from Corollary \ref{cor:CND_HT}, the observation that $F^{-1}(\phi(x))$ has sensitivity 1, and property 1 of Definition \ref{def:CND}. Property 2 can be verified using algebra of cdfs. Property 3 is a standard construction of a $p$-value \citep[Theorem 8.3.27]{casella2002statistical}. Property 4 is a special case of Lemma \ref{lem:threshold}, a general lemma about $p$-values.
  \end{proof}
  
  We see from Theorem \ref{thm:pVal} that given an $f$-DP test $\phi$, we can report both a summary statistic (namely, $T$) as well as a $p$-value (a post-processing of $T$) which contain strictly more information than only sampling $\mathrm{Bern}(\phi(x))$. This shows that for simple null hypotheses, there is no general privacy amplification when post-processing a $p$-value or test statistic to a binary accept/reject decision. 
  
  While in part 3 of Theorem \ref{thm:pVal} there are no assumptions on $H_0$, for some composite null hypotheses, the resulting $p$-value may have very low power. Part 4 states that if the null hypothesis is a singleton, then the power is perfectly preserved.

  We also remark that while the proof of Theorem \ref{thm:pVal} is not technical, it heavily relies on the properties of the CND, showing that the notion of CND has exactly the right properties for Theorem \ref{thm:pVal} to hold.
  
  Note that Theorem \ref{thm:pVal} starts with an $f$-DP test, and shows how to get a private summary statistic and $p$-values. However, constructing a private test $\phi$ is another matter. In Section \ref{s:exchange}, we show that for exchangeable binary data, we can construct a most powerful $f$-DP test in terms of a CND. 
  
  \begin{remark}
  While recently there has been controversy around the use of $p$-values in scientific research \citep{colquhoun2017reproducibility,wasserstein2016asa}, this is mostly due to the misuse or misinterpretation of a $p$-value. Many of the criticisms of $p$-values can be addressed by including additional statistical measures such as the effect size, confidence intervals, likeihood ratios, or Bayes factors. We view $p$-values as a valuable tool that is a component of a complete statistical analysis. Since the $p$-values of Theorem \ref{thm:pVal} are a post-processing of a private summary statistic, that statistic can also be potentially used for other statistical inference tasks, such as in \citet{awan2020differentially}.
  \end{remark}

\subsection{Most powerful tests for exchangeable binary data}\label{s:exchange}
In this section, we extend the main result of \citet{awan2018differentially}, that of constructing most powerful DP tests, to general $f$-DP as well as exchangeable distributions on $\{0,1\}^n$. In contrast, the hypothesis tests of \citet{awan2018differentially} were limited to $(\ep,\de)$-DP and i.i.d. Bernoulli data. A distribution $P$ on a set $\mscr X^n$ is \emph{exchangeable} if given $\ul X\sim P$ and a permutation $\pi$, $\ul X\overset d = \pi(\ul X)$. Note that i.i.d. data are always exchangeable, but there are exchangeable distributions that are not i.i.d. For example, sampling without replacement results in exchangeable but non-i.i.d. data. 

In the next result, we extend Theorem 3.2 of \citet{awan2018differentially} from $(\ep,\de)$-DP to the setting of general $f$-DP. The argument is essentially identical. We include the proof for completeness.

\begin{lemma}[Theorem 3.2 of \citet{awan2018differentially}]\label{lem:exchangeable}
Let $\mscr P$ be a set of exchangeable distributions on $\mscr X^n$. Let $\phi:\mscr X^n \rightarrow [0,1]$ be a test satisfying $f$-DP. Then there exists a test $\phi':\mscr X^n \rightarrow [0,1]$ such that for all $\ul x\in \mscr X^n$, $\phi'(\ul x)$ only depends on the empirical distribution of $\ul x$, and $\int \phi'(\ul x) \ dP = \int \phi(\ul x) \ dP$ for all $P \in \mscr P$.
\end{lemma}
\begin{proof}
Define $\phi'(\ul x) = \frac 1{n!} \sum_{\pi \in \sigma(n)} \phi(\pi(\ul x))$, where $\sigma(n)$ is the symmetric group on $n$ letters. Note that for any $\pi \in \sigma(n)$, $\phi(\pi(\cdot))$ satisfies $f$-DP (just rearranging the sample space). Furthermore, $\int \phi(\pi(\ul x)) \ dP = \int \phi(\ul x) \ dP$ by exchangeability. Finally, by the convexity of $f$, the set of tests $\phi$ which satisfy $\phi(x)\leq 1-f(\phi(x'))$ is a convex set, and so is closed under convex combinations. So, $\phi'$ defined above satisfies $f$-DP, and by the linearity of integrals, preserves the expectations. 
\end{proof}

We work with the sample space $\mscr X=\{0,1\}$. Note that by Lemma \ref{lem:exchangeable}, because we are dealing with exchangeable distributions, the test need only depend on $X=\sum_{i=1}^n X_i$, so we define $\phi(x)$ for $x=0,1,2,\ldots, n$. Since changing one $X_i$ only changes $X$ by $\pm1$, we need only relate $\phi(x)$ and $\phi(x-1)$.

The main result of this section, Theorem \ref{thm:binary} constructs {not only the first private hypothesis test in the general $f$-DP framework, but derives} a most powerful $f$-DP test as well as a corresponding $p$-value in terms of the canonical noise distribution. The proof of Theorem \ref{thm:binary} is similar to the proof of {\citet[Theorem 4.5]{awan2018differentially}, further demonstrating} that the canonical noise distribution is the appropriate concept needed to extend their result from $(\ep,\de)$-DP to arbitrary $f$-DP. {Just like in \citet{awan2018differentially}, we have the surprising result that the UMP DP test in this case only depends on the summary statistic $x+N$, where $N$ is a CND.} The extension from Bernoulli distributions to arbitrary exchangeable binary variables is simply an observation that the argument only depends on the likelihood ratio. However, the extension to exchangeable distributions will allow us to apply {Theorem \ref{thm:binary}} to the difference-of-proportions problem in Section \ref{s:application}.

\begin{restatable}{thm}{thmbinary}\label{thm:binary}
  Let $f$ be a symmetric nontrivial tradeoff function and let $F$ be a CND of $f$. Let $\mscr X=\{0,1\}$. Let $P$ and $Q$ be two exchangeable distributions on $\mscr X^n$ with pmfs $p$ and $q$ such that $\frac{q}{p}$ is an increasing function of $x=\sum_{i=1}^n x_i$. Let $\alpha\in(0,1)$.  Then a most powerful $f$-DP test $\phi$ with level $\alpha$ for $H_0: X\sim P$ versus $H_1: X\sim Q$ can be expressed in any of the following forms:
  \begin{enumerate}
      \item There exists $y\in \{0,1,2,\ldots, n\}$ and $c\in (0,1)$ such that for all $x\in \{0,1,2,\ldots, n\}$,
      \[\phi(x) = \begin{cases}
    0&x<y,\\
  c& x=y,\\
  1-f(\phi(x-1))& x>y,
  \end{cases}\]
  where if $y>0$ then $c$ satisfies $c\leq 1-f(0)$, and $c$ and $y$ are chosen such that $\EE_{P}\phi(x)=\alpha$.  If $f(0)=1$, then $y=0$.
  \item $\phi(x) = F(x-m)$, where $m\in \RR$ is chosen such that $\EE_P \phi(x)=\alpha$. 
 \item Let $N\sim F$. The variable $T=X+N$ satisfies $f$-DP. Then $p = \EE_{X\sim P} F(X-T)$ is a $p$-value and $I(p\leq \alpha)\mid X=I(T\geq m)\mid X\sim \mathrm{Bern}(\phi(X))$, where $\phi(x)$ agrees with 1 and 2 above.
  \end{enumerate}
\end{restatable}
\begin{proof}[Proof sketch.]
Similar to the proof of \citet[Theorem 4.5]{awan2018differentially}, we begin by establishing the equivalence of forms 1 and 2, and arguing that there exists a test of the form 2 by the Intermediate Value Theorem. Using \citet[Lemma 4.4]{awan2018differentially}, a variation of the Neyman Pearson Lemma, we argue that the proposed $\phi$ is most powerful. Statement 3 uses the expressions from Theorem \ref{thm:pVal} as well as some distributional algebra of CNDs to get the more explicit formula. 
\end{proof}
While Theorem \ref{thm:pVal} took an $f$-DP test and produced ``free'' private $p$-values, Theorem \ref{thm:binary} constructs an optimal test from scratch beginning only with a CND.

\begin{example}
Let us consider what distributions fit within the framework of Theorem \ref{thm:binary}. If the variables $X_i$ are i.i.d., then they are distributed as Bernoulli. However, it is possible for the variables to be exchangeable and not independent. For example, the sum $X = \sum_{i=1}^n X_i$ could be distributed as a hypergeometric or Fisher's noncentral hypergeometric, which arises in two sample tests of proportions, see Section \ref{s:application}. For other exchangeable binary distributions, see \citet{dang2009unified}.
\end{example}

  \begin{remark}\label{rem:Awan}
  Theorem \ref{thm:binary} and Corollary \ref{cor:CND_HT} show that the results of \citet{awan2020differentially} extend to arbitrary $f$-DP. By simply modifying the Tulap distribution to a CND, all of the other results of \citet{awan2020differentially} carry over as well. In particular, for Bernoulli data, there exists a UMP one-sided test, a UMP unbiased two-sided test, UMA one sided confidence interval and UMA unbiased two-sided confidence interval. All of these quantities are a post-processing of the summary value $X+N$, where the noise $N$ is drawn from a CND $F$ of $f$. 
  \end{remark}
  
  %%%%%%%%%%%%%%%%%%%%%%%%%%%%%%%%%%%%%%%%%%%%%%%%%%%%%%%%%%%%
  %%%   APPLICATION
  %%%%%%%%%%%%%%%%%%%%%%%%%%%%%%%%%%%%%%%%%%%%%%%%%%%%%%%%%%%%%
  \section{Extension to semi-private difference-of-proportions tests}\label{s:application}
  Testing two population proportions is a very common hypothesis testing problem, which arises in clinical trials with control and test groups, A/B testing, and observation studies comparing two groups (such as men and women, students from two universities, or aspects of two different countries). As such, the techniques for testing such hypotheses are very standardized and taught in many introductory statistics textbooks. However, there are limited techniques to test these hypotheses under $f$-DP. 

%Suppose that we have collected binary responses from two disjoint populations $\mscr X$ and $\mscr Y$, where we have $n$ responses from population $\mscr X$ and $m$ responses from population $\mscr Y$. We model the responses as random variables $X_1,\ldots, X_n \iid \mathrm{Bern}(\theta_X)$ and $Y_1,\ldots, Y_m \iid \mathrm{Bern}(\theta_Y)$, where we also assume that $(X_1,\ldots,X_n) \indep (Y_1,\ldots,Y_n)$. For privacy, we consider two datasets \emph{adjacent} if either one of the $X_i$ is changed or one of the $Y_i$ is changed (but only one total value). We consider $m$ and $n$ to be publicly known values. 

%We are interested in testing the following hypothesis:
%\[H_0: \theta_X\geq \theta_Y \quad \text{ versus } \quad H_1: \theta_X<\theta_Y,\]
%subject to the constraint of $f$-DP. Such one-sided tests can also be converted to two-sided tests using a Bonferroni correction, as discussed in Remark \ref{rem:Bonferroni}, at the end of Section \ref{s:inversion}.%$(\ep,\de)$-DP. 

%First, recall that by Lemma \ref{lem:exchangeable}, any DP hypothesis test need only depend on the empirical distribution, which is equivalent to $(X,Y)=(\sum_{i=1}^n X_i, \sum_{j=1}^m Y_j)$. We denote by $\phi(x,y)$ the test function, which depends on the observed values of $X=x$ and $Y=y$. 

In Appendix \ref{s:noUMPU} we show that subject to differential privacy, there does not exist a UMP (unbiased) $f$-DP test. Nevertheless, we use the techniques developed earlier in this paper to derive a ``semi-private'' UMP unbiased test, which gives an upper bound on the power of any $f$-DP UMP unbiased test. The novel concept of ``semi-privacy'' enforces some of the DP constraints but not others, and this framework may be of independent interest when analyzing a combination of private and non-private releases (see Remark \ref{rem:semi} for more details). We then construct an $f$-DP test which allows for optimal inference for the two population parameters, and which we show through simulations to have comparable power to the semi-private UMP unbiased test. In the case of $\ep$-DP, we show through simulations that the proposed DP test is similar to the semi-private UMP unbiased test with privacy parameter $(\ep/\sqrt 2)$. We also demonstrate that the proposed test has more accurate $p$-values and type I error than commonly used Normal approximation tests.

\subsection{Semi-private UMP unbiased test}\label{s:semiprivate}
In this section, we simplify the search for an $f$-DP test for the difference of proportions, establishing a condition for the test to be \emph{unbiased}. However, as demonstrated through an example in Appendix \ref{s:noUMPU}, there does not in general exist a UMP unbiased (UMPU) $f$-DP test. By weakening the privacy guarantee, we develop a ``semi-private'' UMPU test which can be efficiently implemented. While the ``semi-private'' test does not satisfy $f$-DP, it gives an upper bound on the power of any other unbiased $f$-DP test, and serves as a useful baseline in Section \ref{s:simulations}.

We observe independent $X_i \iid \mathrm{Bern}(\theta_X)$ for $i=1,\ldots, n$ and $Y_j \iid \mathrm{Bern}(\theta_Y)$ for $j=1,\ldots, m$. For privacy, we consider two datasets \emph{adjacent} if either one of the $X_i$ is changed or one of the $Y_i$ is changed (but only one total value). We consider $m$ and $n$ to be publicly known values. We wish to test $H_0: \theta_X\geq\theta_Y$ versus $H_1: \theta_X<\theta_Y$, subject to the constraint of differential privacy. Such one-sided tests can also be converted to two-sided tests using a Bonferroni correction, as discussed in Remark \ref{rem:Bonferroni}, at the end of Section \ref{s:inversion}.%$(\ep,\de)$-DP. 

By a similar argument as in Lemma \ref{lem:exchangeable}, it is sufficient to consider tests which are functions of the empirical distributions of $\ul X$ and $\ul Y$. Equivalently, we may restrict to tests which are functions of $X = \sum_{i=1}^n X_i$ and $Y = \sum_{j=1}^m Y_j$. We consider two databases adjacent if either $X$ changes by 1 or if $Y$ changes by 1 (but not both). By Lemma \ref{lem:DPtest}, a test $\phi(x,y)$ satisfies $f$-DP if the following set of inequalities hold
\begin{equation}\label{eq:DPIneq}
  \begin{array}{cc}
    \phi(x,y)\leq 1-f(\phi(x+1,y))\\%&\quad(1-\phi(x,y))\leq e^\ep(1-\phi(x+1,y))+\de\\
    \phi(x,y)\leq 1-f(\phi(x-1,y))\\%&\quad(1-\phi(x,y))\leq e^\ep(1-\phi(x-1,y))+\de\\
    \phi(x,y)\leq 1-f( \phi(x,y+1))\\%&\quad(1-\phi(x,y))\leq e^\ep (1-\phi(x,y+1))+\de\\
    \phi(x,y)\leq 1-f(\phi(x,y-1)),%&\quad(1-\phi(x,y))\leq e^\ep(1-\phi(x,y-1))+\de,
  \end{array}\end{equation}
for all pairs of $(x,y)$. 

Classically, it is known that even without privacy there is no uniformly most powerful test for this problem. Traditionally, attention is restricted to unbiased tests. Recall that a test is unbiased if for all $\theta_1\in \Theta_1$ and $\theta_0\in \Theta_0$, the power at $\theta_1$ is higher than at $\theta_0$ (here, $\theta$ represents the pair $(\theta_X,\theta_Y)$). Because the variables $(X,Y)$ have distribution in the exponential family, the search for a UMP unbiased test can be restricted to tests which satisfy $\EE_{\theta_X=\theta_Y} (\phi(X,Y)\mid X+Y=z)=\alpha$ \citep[Proof of Theorem 4.124]{schervish2012theory}, since $X+Y$ is a complete sufficient statistic under $H_0$. When $\theta_X=\theta_Y=\theta_0$, $X+Y \sim \mathrm{Binom}(m+n,\theta_0)$, and $Y\mid (X+Y=z) \sim \mathrm{Hyper}(m,n,z)$, where $\mathrm{Hyper}(m,n,z)$ is the hypergeometric distribution, where we draw $m$ balls out of a total of $m+n$ balls, and where $z$ balls are white, and the random variable counts the number of drawn white balls. This is equivalent to a permutation test where we shuffle the labels of the observations. Lemma \ref{lem:hyper} summarizes these observations. 

%By Proposition \ref{prop:unbiased} and Lemma \ref{lem:Neyman}, we know that an unbiased test must have Neyman structure relative to $X+Y$, since this is a boundedly complete sufficient statistic under $\ol{\Theta_0}\cap \ol{\Theta_1}$, where the overline indicates the closure of the set. So the test must satisfy $\EE_\ta (\phi(X,Y)\mid X+Y=z)=\al$ for all $z$ and all $\theta\in \ol{\Theta_0}\cap \ol{\Theta_1}$. 

\begin{lem}\label{lem:hyper}
Let $X\sim \mathrm{Binom}(n,\theta_X)$ and $Y\sim \mathrm{Binom}(m,\theta_Y)$ be independent. Consider the test $H_0: \theta_X\geq \theta_Y$ and $H_1: \theta_X<\theta_Y$. Let $\Phi$ be a set of tests. If there exists a UMP test $\phi\in\Phi$ among those which satisfy 
\begin{equation}\label{eq:hyper}
    \EE_{H\sim\mathrm{Hyper}(m,n,z)}\phi(z-H,H)=\alpha,
\end{equation} 
for all $\alpha$, then $\phi$ is UMP unbiased size $\alpha$ among $\Phi$. 
\end{lem}
\begin{proof}
It is easy to verify that the power function is continuous, and that $X+Y$ is a boundedly complete sufficient statistic under $H_0$. By \citet[Proposition 4.92]{schervish2012theory} and \citet[Lemma 4.122]{schervish2012theory}, the set of unbiased tests for this problem is a subset of the tests which satisfy Equation \eqref{eq:hyper}. It is also clear that Equation \eqref{eq:hyper} implies that the test is size $\alpha$. It follows that if a test is UMP among the tests in $\Phi$ satisfying Equation \eqref{eq:hyper} then it is UMP unbiased size $\alpha$ among $\Phi$. 
\end{proof}

However, as demonstrated by an example given later in Appendix \ref{s:noUMPU}, in general there is no UMP test for the hypothesis $H_0: \ta_X\geq\ta_Y$ versus $H_1: \ta_X<\ta_Y$ among the set
\begin{equation}\label{eq:Phi}
    \Phi_{f} = \l\{\phi(x,y) \mid \phi \text{ satisfies inequalities  \eqref{eq:DPIneq} and Equation \eqref{eq:hyper}}\r\}.
\end{equation}
The reason for this is that Lemma \ref{lem:hyper} suggests that a UMP unbiased test relies on being able to construct a UMP test, given $X+Y=z$. However, the inequalities \eqref{eq:DPIneq} put constraints, relating $\phi(x,y)$ for different values of $z$. 

Instead of requiring that all of the inequalities \eqref{eq:DPIneq} hold, we weaken the requirement of differential privacy, to only include the constraints relating $(x,y)$ with the same sum $x+y=z$. We call the following the set of ``semi-private'' tests:
\[\Phi^{\text{semi}}_{f}= \l\{\phi(x,y) \middle|
  \begin{array}{c}
    \text{for each } z\in \{0,1,\ldots, m+n\},\\
    \text{ there exists } \psi \in \Phi_{f},\\
    \text{ s.t. } \phi(x,y) = \psi(x,y) \text{ for all } x+y=z
  \end{array}\r\}.\]
   Intuitively, $\Phi^{\text{semi}}_{f}$ is the set of tests, which satisfy the set of implied constraints of \eqref{eq:DPIneq}, which only relate $(x,y)$ and $(x+1,y-1)$. So, the summary $z=X+Y$ is not protected at all, but for any $X+Y=z$, $(X,Y)$ must satisfy $f$-DP. While these semi-private tests are not necessarily intended for the purpose of privacy protection, by weakening the privacy requirement, they offer an upper bound on the performance of any DP test, as stated in Corollary \ref{cor:semi}.

\begin{restatable}
  [Semi-Private UMPU]{thm}{thmsemi}\label{thm:semi}
  Let $f$ be a symmetric nontrivial tradeoff function and let $F$ be a CND for $f$. Let $X\sim \mathrm{Binom}(n,\theta_X)$ and $Y\sim \mathrm{Binom}(m,\ta_Y)$ be independent. Let $\al\in (0,1)$ be given. For the hypothesis $H_0: \ta_X\geq\ta_Y$ versus $H_1: \ta_X<\ta_Y$,
  \begin{enumerate}
  \item $\phi^*(x,y) = F(y-x-c(x+y))$ is the UMPU test of size $\al$ among $\Phi^{\text{semi}}_{f}$, where $c(x+y)$ is chosen such that $\EE_{H\sim \mathrm{Hyper}(m,n,x+y)} \phi^*((x+y)-H,H) = \al$.
  \item Set $T = Y-X+N$, where $N\sim F$, and set $Z=X+Y$. Then 
  \[p = \EE_{H\sim \mathrm{Hyper}(m,n,Z)}F(2H-Z-T)\] is the exact $p$-value corresponding to $\phi^*$.
  \end{enumerate}
\end{restatable}
\begin{proof}[Proof sketch.]
Lemma \ref{lem:hyper} reduced the problem to determining whether the test is UMP among those which satisfy Equation \eqref{eq:hyper}. The technical lemmas \ref{lem:convert} and \ref{lem:twice}, given in Appendix \ref{s:proofs}, quantify the privacy of the semi-private tests when viewed as a function of $y$ (where $z$ is fixed), and determine the CND of the derived tradeoff function. Conditional on $z$, the distribution of $Y$ is a Fisher noncentral hypergeometric distribution \citep{harkness1965properties,fog2008sampling}. By Theorem \ref{thm:binary} we can construct the most powerful DP test based on the CND. Finally, we verify a monotone likelihood ratio property of the noncentral hypergeometrics to argue that the test is in fact uniformly most powerful.
\end{proof}

Corollary \ref{cor:semi} shows that while the semiprivate UMPU test does not satisfy $f$-DP, we can use it as a benchmark to compare other tests, as it gives an upper bound on the highest possible power of any unbiased $f$-DP level $\alpha$ test. 

\begin{cor}\label{cor:semi}
  Let $\phi^*(x,y)$ be the UMPU size $\al$ test among $\Phi^{\text{semi}}_f$, and let $\phi(x,y)$ be any unbiased, level $\al$ test in $\Phi_{f}$. Then
  \[\EE_{\substack{X\sim \ta_X\\Y\sim \ta_Y}} \phi^*(X,Y) \geq \EE_{\substack{X\sim \ta_X\\Y\sim \ta_Y}} \phi(X,Y),\]
  for any values of $\theta_X\leq \theta_Y$.
\end{cor}

{\begin{remark}\label{rem:semi}
    The semi-private framework could potentially be of independent interest, as it is an example of a setting where some statistics are preserved exactly, whereas others are protected with privacy noise. For example, this is similar to the framework used for the 2020 Decennial Census, where certain counts are preserved without any privacy noise, and the other counts are sanitized by an additive noise mechanism. While they phrase their privacy guarantee in terms of post-processing, one could also view it as a ``semi-private'' procedure, where their privacy guarantee only holds for the databases which agree with the preserved counts. This is an alternative perspective to \emph{subspace differential privacy} \citep{gao2022subspace}, which restricts the output of a mechanism rather than the input database.
\end{remark}}

%%%%%%%%%%%%%%%%%%%%%%%%%%%%%%%%%%%%%%%%%%%%%%%%%%%%%%%%%%%%%%%%%%%%%%%%%%%%%
%%%    INVERSION   TEST
%%%%%%%%%%%%%%%%%%%%%%%%%%%%%%%%%%%%%%%%%%%%%%%%%%%%%%%%%%%%%%%%%%%%%%%%%%%%%%%
\subsection{Designing an  \texorpdfstring{$f$}{TEXT}-DP test for difference-of-proportions}\label{s:inversion}
Based on the negative result of Appendix \ref{s:noUMPU}, we consider a different approach to building a well-performing DP test. A very common non-private test used to test $H_0: \thx\geq\thy$ versus $H_1: \thx<\thy$ for $X\sim \mathrm{Binom}(n,\thx)$ and $Y\sim \mathrm{Binom}(m,\thy)$ is based on the test statistic 
\[Y/m-X/n,\]
which is intuitive as this quantity captures the sample evidence for the difference between $\thx$ and $\thy$. In fact this statistic has the important property that its expectation under the null does not depend on the parameter $\theta_X=\theta_Y$. If this were not the case, then tests based on this statistic would have limited power \citep{robins2000asymptotic}. However, the sampling distribution of this quantity depends on the parameter $\theta = \thx=\thy$ under the null (e.g., for $\theta=1/2$, the variance of $Y/m-X/n$ is higher than when $\theta$ is larger or smaller). Typically, the central limit theorem is used to justify that
\[\frac{Y/m-X/n}{\sqrt{(1/m+1/n)\hat \theta_{0}(1-\hat \theta_{0}})} \approx N(0,1),\]
where $\hat \theta_0 = \frac{X+Y}{m+n}$ is the maximum likelihood estimator for $\theta$ under the null. The central limit approximation works well in large samples, but for small samples this  approximation can be inadequate as demonstrated in the simulations of Section \ref{s:simulations}.

\subsubsection{Inversion-based parametric bootstrap \texorpdfstring{$f$}{TEXT}-DP test}
In this section, we consider tests based on the following privatized summary quantities $X+N_1$ and $Y+N_2$, where $N_1,N_2\iid F$ where $F$ is a CND of $f$. The vector $(X+N_1,Y+N_2)$ satisfies $f$-DP, since only one of $X$ and $Y$ changes by at most 1, between adjacent databases. 

\begin{remark}
Basing our test on these two noisy statistics has a few important benefits. As noted in Remark \ref{rem:Awan}, given $X+N_1$ and $Y+N_2$ we can perform optimal hypothesis tests and confidence intervals for $\theta_X$ and $\theta_Y$ combining Theorem \ref{thm:binary}, Corollary \ref{cor:CND_HT} and the other results of \citet{awan2020differentially}. In general this is not the case for an arbitrary $f$-DP test of $H_0: \thx\geq\thy$ versus $H_1: \thx<\thy$. While \ref{thm:pVal} says that we can always get a summary statistic and $p$-value out of an arbitrary $f$-DP test, these values may not contain enough information to do inference (let alone optimal inference) for $\theta_X$ and $\theta_Y$ separately.
\end{remark}

Then we consider the quantity $\displaystyle T = m^{-1}(Y+N_2) - n^{-1}(X+N_1)$. 
Asymptotics tells us that under the null hypothesis, $T/\sqrt{(1/m+1/n)\theta(1-\theta)}\overset d \rightarrow N(0,1)$, which is the same sampling distribution as without privacy. However, as many other researchers have noted, while these approximations are serviceable in classical settings, the approximations are too poor when privacy noise is introduced \citep{wang2018statistical}. One reason for this is that the noise introduced to achieve privacy, such as Laplace or Tulap,  often has heavier tails than the limit distribution, which is often Gaussian.

We notice that $T$ is a linear combination of independent random variables. So, we can use characteristic functions to derive the sampling distribution of $T$ under a specific null parameter $\theta$. 

we use $\psi_X(\cdot)$ to denote the characteristic function of a random variable $X$: $\psi_X(t) \defeq \EE_{X} e^{itX}$. Recall that for independent random variables $X_1,\ldots, X_n$ and real values $a_1,\ldots, a_n$, if $X=\sum_{i=1}^n a_i X_i$, then    $\psi_{X}(t) = \prod_{i=1}^n \psi_{X_i}(a_it)$. 

Then the characteristic function of our test statistic $T$ is given by
\[\psi_{T\sim \theta}(t) = \psi_{Y\sim \theta}(t/m)\psi_{N_2}(t/m)\psi_{X\sim \theta}(-t/n)\psi_{N_1}(-t/n).\]
We know the characteristic function for a binomial random variable, and for many common DP distributions $N$, we have formulas for $\psi_N$ as well. 

We can use the following inversion formula to evaluate the cdf of $T$.

\begin{lemma}[Inversion Formula: Gil-Pelaez]\label{lem:inversion}
Let $X$ be a real-valued continuous random variable, with characteristic function $\psi_X(t)$. Then the cdf of $X$ can be evaluated as
\[F_{X}(x) = \int_0^\infty \frac{\mathrm{Im}(e^{-itx} \psi_{X}(t))}{t} \ dt,\]
  where $\mathrm{Im}(\cdot)$ returns the imaginary component of a complex number: $\mathrm{Im}(z) = (z-z^*)/(2i)$, where $z^*$ is the complex conjugate of $z$.
\end{lemma}

Lemma \ref{lem:inversion} gives a computationally tractable method of evaluating the exact sampling distribution of $T$ at a given null parameter. Since larger values of $T$ give more evidence of the alternative hypothesis, $p(T)=1-F_{T\sim \theta_0}(T)$ is a $p$-value for the null hypothesis $H_0: \theta_X=\theta_Y=\theta_0$ \citep[Theorem 8.3.27]{casella2002statistical}. However, this $p$-value depends on the null parameter $\theta_0$, which we likely do not know. A solution is to substitute an estimator for $\theta_0$ under the null hypothesis that $\thx=\thy$, based on the privatized statistics $X+N_1$ and $Y + N_2$.  A natural estimator is $\hat \theta_{0} = \min\{\max\{\frac{X+N_1+Y+N_2}{m+n},0\},1\}$. Plugging this estimate in for $\theta_0$ gives the approximate $p$-value:

%We also do not want to use the worst case value of $\theta_0 = 1/2$, as this will result in a dramatic loss of power for extreme values of $\thx$ and $\thy$. 
\begin{comment}
\begin{prop}
  Let $X\sim \mathrm{Binom}(n,\thx)$ and $Y\sim \mathrm{Binom}(m,\thy)$ be independent, and let $N_1\sim f_{N_1}$ and $N_2\sim f_{N_2}$ be independent continuous random variables, which are symmetric about zero. Set 
  \[T = \frac{Y+N_2}{m} - \frac{X+N_1}{n},\]
  and note that the characteristic function of $T$ is given by:
  \[\psi_{T\sim \theta}(t) = \psi_Y(t/m)\psi_{N_2}(t/m)\psi_{X}(-t/n)\psi_{N_1}(-t/n).\]
Using the inversion formula, we have that 
\[F_{T}(x) = \int_0^\infty \frac{\mathrm{Im}(e^{-itx} \psi_{T}(t))}{t} \ dt.\]
Then $p(T) = 1-F_{T\sim \theta_0}(T)$ is a $p$-value for the null hypothesis $H_0: \thx=\thy=\theta_0$. 
\end{prop}
\end{comment}

\[\twid p(T,\hat\theta_0) = 1-F_{T_0\sim \hat \theta_{0}}(T).\]
This approximate $p$-value is our recommended $f$-DP test for the difference-of-proportions testing problem, and the procedure is summarized in Algorithm \ref{alg:inversion} for the cases of $(\ep,0)$-DP and $\mu$-GDP. While $p$-value is not exact, and is thus not guaranteed to have the intended type I error, the results of \citet{robins2000asymptotic} imply that this $p$-value is asymptotically uniform under the null, implying that the test is asymptotically unbiased, with asymptotically accurate type I errors. Furthermore, as we demonstrate in Section \ref{s:simulations}, for even sample sizes as small as $n,m\geq 30$, the approximation is incredibly accurate, offering accuracy even higher than the classic normal approximation test, which is widely used and accepted. We also show in Section \ref{s:simulations} that the power of the test is comparable to the semi-private test of Section \ref{s:semiprivate} indicating that it is near optimal. 

\begin{algorithm}[t]
\SetAlgoLined
%\KwResult{$\ep$-DP or $\rho$-CDP test}
 Let $X$, $Y$, $m$, and $n$ be given. Let either $\ep$ or $\mu$ be given.\;
 \If{$\ep$-DP}{
 Draw $N_1,N_2 \iid \mathrm{Tulap}(0,\exp(-\ep),0)$\;
    Set $\psi_{N}(t)  = \frac{\l[1-\exp(-\ep)\r]^2\l[\exp(-it/2)-\exp(it/2)\r]}{it\l[1-\exp(it-\ep)\r]\l[1-\exp(-it-\ep)\r]}$\;
 }\If{$\mu$-GDP}{
 Draw $N_1,N_2 \iid N(0,1/\mu^2)$\;
 Set $\psi_{N}(t) = \exp(-t^2/(2\mu^2))$\;
 }
 Set $\psi_{Y\sim \theta}(t) = ((1-\theta) + \theta \exp(it))^m$ and $\psi_{X\sim \theta}(t) = ((1-\theta) + \theta \exp(it))^n$\;
  Set $\hat X = X+N_1$ and $\hat Y = Y+N_2$\;
  Set $T = \hat Y/m-\hat X/n$\;
  Set $\hat \theta = \min\l\{\max\l\{\frac{\hat X+\hat Y}{m+n},0\r\},1\r\}$\;
  Set $\psi_{T\sim \theta}(t) =  \psi_{Y\sim \theta}(t/m)\psi_{X\sim \theta}(-t/n) \psi_N(t/m)\psi_N(-t/n)$\;
  Output $p$-value and summary values: $p= 1-\int_0^\infty \frac{\mathrm{Im}(\exp(itT) \psi_{T\sim \hat \theta}(t))}{t} \ dt$, $X+N_1$, $Y+N_2$
 \caption{$\ep$-DP or $\mu$-GDP approximate $p$-value, based on inversion.
 }\label{alg:inversion}
\end{algorithm}

\begin{remark}
While the p-value generated from Algorithm \ref{alg:inversion} may seem complex, it is relatively easy to implement. For instance in R, the command \texttt{integrate} can perform an accurate numerical integral. Another strength of Algorithm \ref{alg:inversion} is that the running time does not depend on the sample size $m$ or $n$, whereas the semi-private test runs in $O(m)$ time. 
\end{remark}

\begin{remark}
Algorithm \ref{alg:inversion} can be viewed as an exact evaluation of a parametric bootstrap, where we by-pass the need for sampling by numerically computing the cdf. As such,  we avoid the additional error and running time produced by the Monte Carlo sampling.
\end{remark}

\begin{remark}\label{rem:Bonferroni}
    While we focus on the one-sided hypothesis $H_0: \theta_X\geq \theta_Y$ versus $H_1: \theta_X<\theta_Y$, the test of Algorithm \ref{alg:inversion} can be easily modified to produce a ``two-sided'' test for $H_0: \theta_X=\theta_Y$ versus $H_1: \theta_X\neq \theta_Y$. Call $p$ the one-sided $p$-value from Algorithm \ref{alg:inversion}. Then $p_2 = 2\min\{p,1-p\}$ is a $p$-value for the two-sided test. This method of combining multiple tests called a \emph{Bonferroni correction} or an \emph{intersection-union test} \citep[Section 8.2.3]{casella2002statistical}.
\end{remark}

%%%%%%%%%%%%%%%%%%%%%%%%%%%%%%%%%%%%%%%%%%%%%%%%%%%%%
%%%   SIMULATIONS
%%%%%%%%%%%%%%%%%%%%%%%%%%%%%%%%%%%%%%%%%%%%%%%%%%%%%%
\subsection{Simulations}\label{s:simulations}
In this section, we perform several simulations to compare the performance of our proposed DP test to other competing DP tests, the semi-private UMPU test, as well as popularly used non-private tests. While our results can be applied to arbitrary $f$-DP, we only run our simulations for $(\ep,0)$-DP as this privacy definition is commonly used and introduces noise that is difficult to incorporate.

In Section \ref{s:power}, we consider the empirical power of the tests, and show that the inversion DP test out-performs other DP tests, and by comparing against the semi-private test with privacy budget $\ep/\sqrt 2$, show that it is observed to be more powerful than any $(\ep/\sqrt 2)$-DP test {(see Remark \ref{rem:sqrt2} for the intuition behind the factor of $1/\sqrt 2$)}. In Section \ref{s:typeI}, we consider the type I error of the various tests, and show that the observed type I error of the inversion test is more accurate than the commonly used non-private normal approximation test. We also show that naive DP normal approximation tests have unacceptably inaccurate empirical type I errors. In Section \ref{s:pValues}, we plot the empirical cumulative distribution functions (cdf) of the $p$-values from the various tests demonstrating from another perspective that the proposed test has accurate type I error. % In Section \ref{s:roc}, we consider ROC curves of each of our tests, which allow us to understand the power of each test according to the empirical type I error. This simulation reinforces the utility of the inversion test, and also demonstrates that in the two-sided case, the inversion test has similar performance to the likelihood ratio test, but with much more reliable type I error rates. 

\subsubsection{Type I Error}\label{s:typeI}

The first simulation that we will consider, and one of the most important, demonstrates the reliability of the type I error guarantees of our proposed test against alternative tests. Recall that in the best practices of scientific research, many approximate statistical tests are widely used and accepted. For instance, most hypothesis testing tools are based on asymptotic theory which approximates the sampling distribution, such as the central limit theorem. As such, many widely used tests do not have exact type I error guarantees, but the error of these tests has been determined to be small enough for practical purposes. In Section \ref{s:inversion}, our proposed inversion-based test also involves an approximation to the sampling distribution. We demonstrate in the following simulation that the type I errors of this proposed test are more accurate than the widely accepted normal-approximation test.

For the simulation, we measure the empirical Type I error as the null $\theta_0$ takes values in $\{.05,.1,\ldots,.95\}$ and sample sizes are set to $m=n=30$, based on 20,000 replicates for each $\theta_0$ value. We consider two values for the nominal type I error: in the left plot of Figure \ref{fig:typeI} we set $\alpha=.01$ and in the right plot of Figure \ref{fig:typeI} we set $\alpha =.05$. The dotted horizontal lines represent a 95\% Monte Carlo confidence interval assuming that the true type I error is equal to the nominal level. As there are 19 unique theta values, if a curve crosses these thresholds more than once, this is evidence that the type I error is not appropriately calibrated. For this simulation, we only consider approximate tests as the non-private UMPU test and the semi-private test have perfectly calibrated type I errors. 

In red is the classic normal approximation test, described in Section \ref{s:inversion}. Such approximations are often considered accurate enough when the sample sizes $n$ and $m$ are greater than 30. Some rules of thumb for this problem require that there are at least 8 successes and failures in each group for the approximation to be accurate enough \citep[p. 321]{akritas2015probability}. We see in the left plot of Figure \ref{fig:typeI} that while this test has reasonable empirical type I error for moderate values of $\theta_0$, the test is overly conservative for extreme values of $\theta_0$. In the right plot of Figure \ref{fig:typeI}, we see that the normal approximation test is much less reliable in this setting, with seven of the nineteen values outside of the 95\% confidence region. We see that at extreme values of $\theta_0$, the actual type I error rates are much higher than the nominal level, resulting in excessive false positives. It is interesting that the type I errors are over-conservative when $\alpha=.01$ and inflated when $\alpha=.05$. In general, it is hard to predict whether in a particular setting the type I errors will be too high or too low. 

In green is an $\ep$-DP normal approximation test, proposed by \citet{karwa2018correspondence} which is analogous to the one-sample test of \citet{vu2009differential}. See Appendix \ref{s:otherTests} for a description of the method. While the empirical type I errors of this test are acceptable when $\alpha=.05$, we see that for $\alpha=.01$, the empirical type I error is approximately .016 and is entirely outside the confidence region. We conclude that the type I errors for this normal approximation test are unreliable for these settings.

In light blue is an $\ep$-DP, which splits the budget between privatizing $T=Y-X$ and $Z=X+Y$, and plugs in the results into the semi-private test of Theorem \ref{thm:semi}. The test is described in Algorithm \ref{alg:plugin}, which appears in Appendix \ref{s:otherTests}. The empirical type I errors for the plugin test are slightly higher than expected, crossing the confidence band three times in the left plot and once in the right plot, but are much more reliable than either of the normal approximation tests discussed above.

Finally, in magenta is the inversion-based test of Algorithm \ref{alg:inversion}.  The empirical type I errors of the inversion-based test lie entirely within the confidence bands for both settings of $\alpha$. This indicates that for the settings of these simulations, the type I errors of the inversion test are indistinguishable from the nominal level, and are much more accurate than the classic normal approximation test or a DP normal approximation test, such as in \citet{vu2009differential}. 

\begin{figure}%[ht]
    \centering
    \begin{subfigure}[t]{.48\linewidth}
    \includegraphics[width=\linewidth]{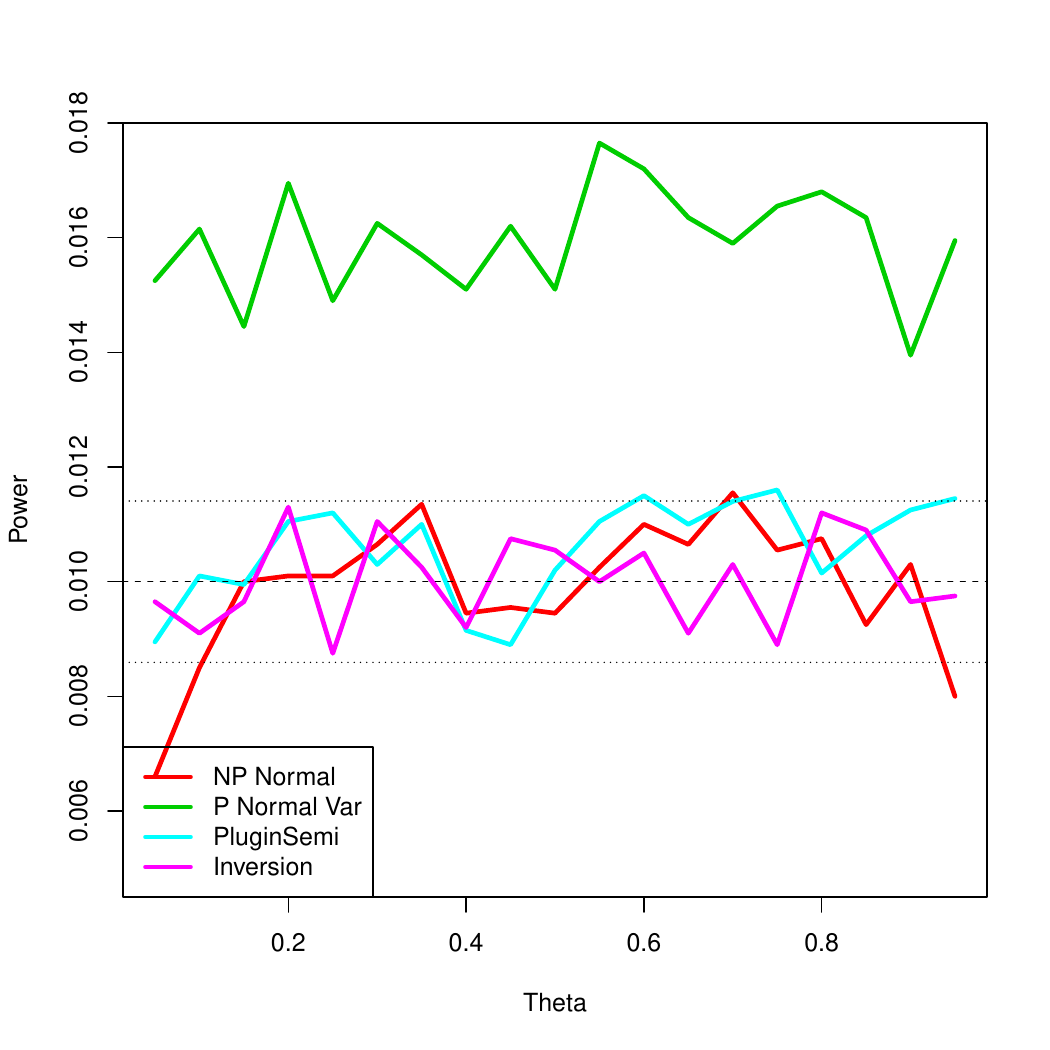}
    \end{subfigure}
		\hspace{.02\linewidth}
		\begin{subfigure}[t]{0.48\linewidth}
    \includegraphics[width=\linewidth]{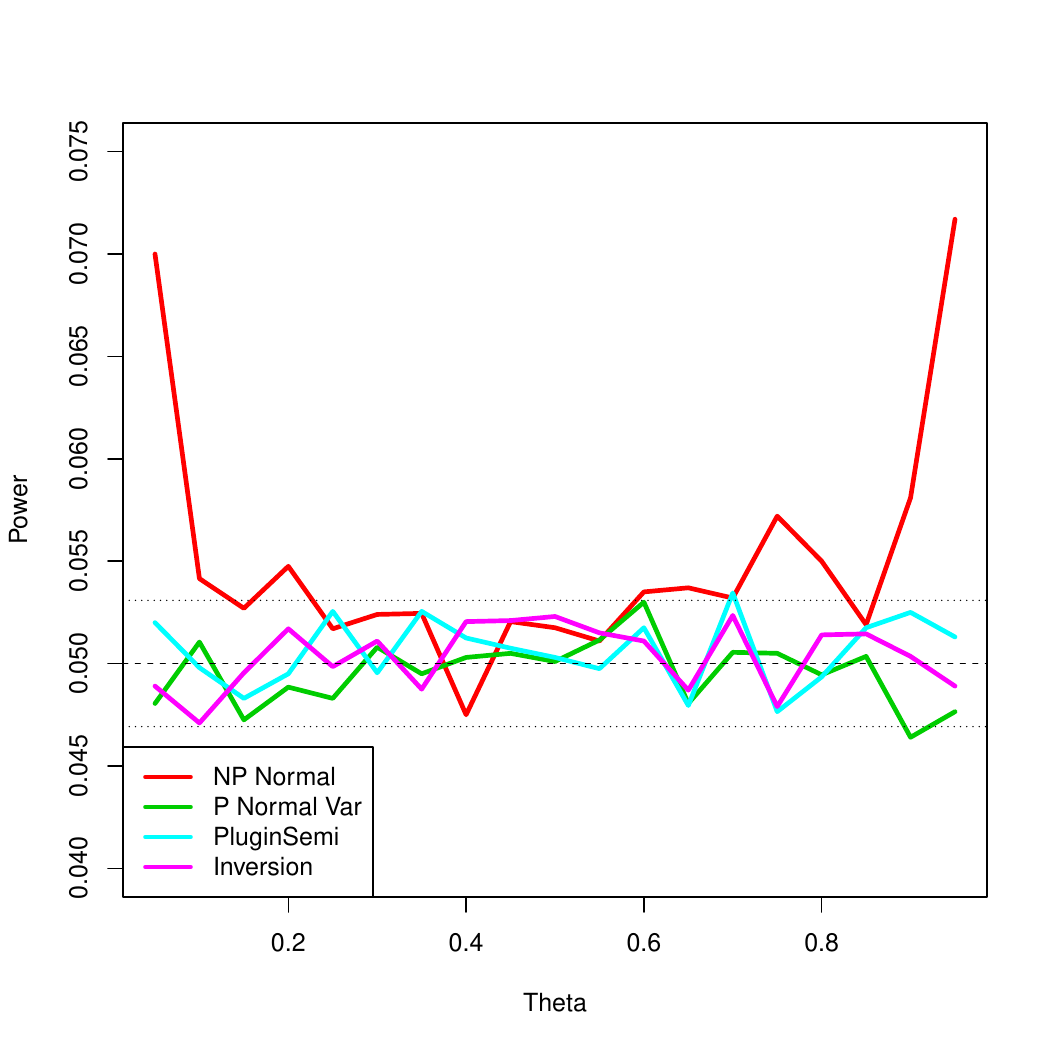}
\end{subfigure}
    \caption{Empirical Type I error as $\theta_0$ varies in $\{.05,.1,\ldots,.95\}$. The nominal $\alpha$ level is $.01$ (left) and $.05$ (right). $m=n=30$, $\ep=.1$, and results are over 20,000 replicates for each $\theta_0$ value.} %The nominal type I error is .01. The sample size is $m=n=30$. Based on 10000 replicates for each theta value. The dotted horizontal lines represent a 95\% confidence interval if the true power is $.01$. As there are 19 unique theta values, if a curve crosses these thresholds more than once, this is evidence that the type I error may not be calibrated. We see that the Private Normal test has greatly inflated type I error, and that the non-private Normal test is inaccurate for extreme values of $\theta_0$. We see that the plugin semi-private test and the inversion test both have accurate type I error.}
    \label{fig:typeI}
\end{figure}

\subsubsection{P-values}\label{s:pValues}
 In this section, we consider the empirical cumulative distribution function (cdf) of the $p$-values, while holding $\theta_0$ fixed. This can be interpreted as varying the nominal $\alpha$ value on the $x$-axis, with the empirical type I error on the $y$-axis. This differs from the previous simulation, where we varied the null value of $\theta$ along the $x$-axis, but left the nominal value of $\alpha$ fixed. Combined with the previous results, this simulation gives a more complete picture of how accurate the type I errors are, for a spectrum of nominal $\alpha$ values. 

For the simulation, we set $\theta_0=.95$, $n=30$, $m=40$, and $\ep=.1$. We chose to investigate $\theta_0=.95$ since the type I errors in Section \ref{s:typeI} were found to be more inaccurate for extreme values of $\theta_0$. The results are based on 100,000 replicates with these settings. The simulation includes the same tests as in Section \ref{s:typeI}, marked with the same color scheme, as well as a test based on the simulation-based method of \citet{awan2020one}.  Included is a dotted black line of intercept 0 and slope 1, which represents perfectly calibrated type I error rates. 

We see that for these simulation settings, the non-private normal approximation test has inflated type I errors for nominal $\alpha$ values between .02 and .2. The DP normal approximation test has inflated type I error rates for nominal alpha values below .05, and deflated type I error rates for larger values of $\alpha$. The plugin test also has inflated type I errors in this setting, while not as extreme as the normal approximate test. Finally, the curve for the inversion test is visually indistinguishable from the dotted black line, indicating that this tests has well-calibrated type I errors for this simulation setting, much improved over the other approximate tests considered here.

\citet{awan2020one} tackled the same DP testing problem, and also based their test on adding Tulap noise to both $X$ and $Y$. They implement their test using the OASIS algorithm, which they argue gives asymptotically accurate type I errors. We include their test in this section for comparison, and while \citet{awan2020one} advocated this approach in large samples, we see in the left plot of Figure \ref{fig:power} that it has greatly inflated type I errors for the smaller sample sizes considered in this simulation.

\begin{figure}%[ht]
    \centering
    \begin{subfigure}[t]{0.48\linewidth}
	\includegraphics[width=\linewidth]{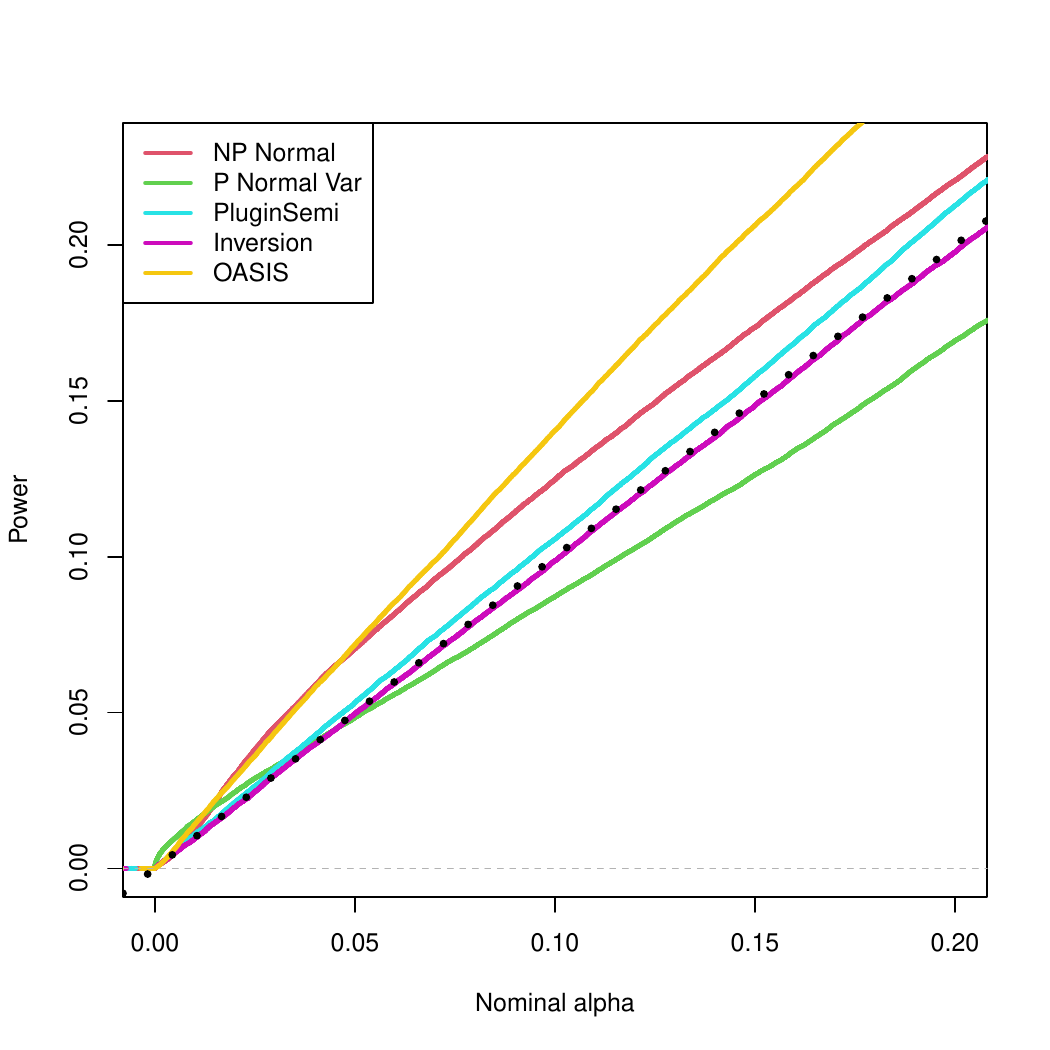}
	\caption{Empirical cdf of the $p$-values. $n=30$, $m=40$, $\ep=.1$, $\theta_0=.95$, and results are based on 100,000 replicates.}
	\end{subfigure}
	\hspace{.02\linewidth}
    \begin{subfigure}[t]{.48\linewidth}
    \includegraphics[width=\linewidth]{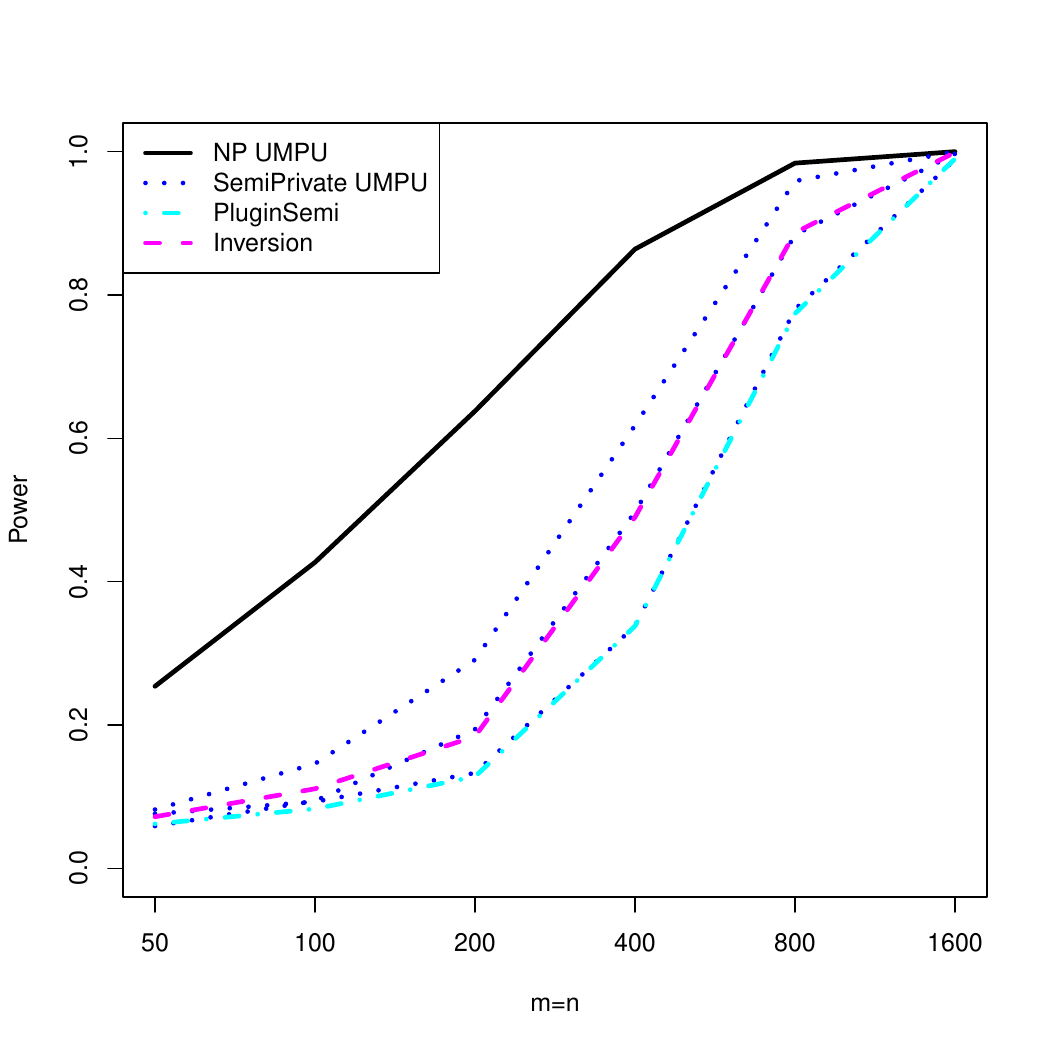}
    \caption{Empirical power at $\theta_X=.5$ and $\theta_Y=.6$, while $m=n$ varies on the $x$-axis. The privacy parameter is $\ep=.1$, and the results are averaged over 1000 replicates for each sample size.}
    \end{subfigure}
    \caption{Simulation results comparing the $p$-values and power of various tests.}
    %\caption{Comparison of power versus sample size, where $m=n$. The parameters for the simulations were $\theta_X=.5$, $\theta_Y=.65$, $\epsilon=.1$, and $\alpha=.05$. The black curve is the non-private UMPU test. The upper light blue dotted curve is the semi-private UMPU from Section (?).The lower light blue dashed curve is the semi-private UMPU using $\ep/2$. The dark blue curve is the semi-private UMPU using $Y-X+N(\ep/2)$ for the test statistic and $X+Y+N(\ep/2)$ as a plug-in estimate for $X+Y$. The magenta dashed curve is our proposed test, based on the inversion formula. The dashed horizontal line is at $\alpha=.05$. Note that the x-axis is on the log-scale. The power is estimated over 1000 samples at each sample size.}
    \label{fig:power}
\end{figure}
\subsubsection{Power}\label{s:power}
Finally, we compare the power of our candidate tests. We use the semi-private UMPU test as a baseline for comparison: recall from Theorem \ref{thm:semi} that the semi-private test has perfectly calibrated type I errors, and is uniformly more powerful than any DP unbiased test. As such, it serves as an upper bound on the power of the other candidate tests. We will see that the inversion test (with $\ep=.1$) has power similar to the semi-private UMPU with $\ep=(.1/\sqrt{2})$, indicating that its power cannot be beaten by the most powerful $(\ep/\sqrt{2})$-DP unbiased test. 

For the simulation, we vary the sample size $n=m$ along the $x$-axis and measure the empirical power on the $y$-axis, at a nominal $\alpha$ level of $.05$. The privacy parameter is set to $\ep=.1$ and the results are based on 1000 replicates for each sample size. In black is the non-private UMPU test, described Appendix \ref{s:umpu}, which is guaranteed to be more powerful than any of the private tests considered in this paper. The dotted dark blue curve is the semi-private UMPU test of Section \ref{s:semiprivate}. Since the semi-private UMPU has a weaker privacy guarantee than DP, this test should also give an upper bound on the power of any DP test. We also include the semi-private test implemented with $\ep=.1/\sqrt 2$ and $\ep=.1/2$, with the same color and line scheme. We see that the plugin test, appearing in light blue, has similar power as the semi-private test with $\ep=.1/ 2$, indicating that this test is more powerful than any $\ep/2$ test. In magenta, we have the inversion-based test, which we see has similar power as the semi-private test with $\ep=.1/\sqrt 2$, indicating that it is more powerful than any $\ep/\sqrt 2$ test.

\begin{remark}\label{rem:sqrt2}
    That the inversion test has comparable power to the semi-private test with $\ep/\sqrt 2$ can be understood as follows: the semi-private test is based on the test statistic $S=Y-X+N$, where $N$ is a Tulap random variable. On the other hand, the inversion test is based on $\twid X=X+N_1$ and $\twid Y=Y+N_2$. If we tried to approximate the test statistic $S$ using  $\twid X$ and $\twid Y$, we end up with $\twid S = Y-X+(N_1-N_2)$. If the same privacy parameters are used for $N$ and $N_1$, $N_2$, then $\var(N_1-N_2)=2\var(N)$. By decreasing the privacy parameter of $N$ to $\ep/\sqrt 2$, we obtain equality of the variances. 
\end{remark}
    \section{Discussion}
    In this paper we proposed the new concept \emph{canonical noise distribution}, which expanded upon previous notions of an optimal noise adding mechanism for privacy. We showed that a CND is a fundamental concept in $f$-DP, connecting it to optimality properties of private hypothesis testing. Using CNDs and the theoretical results on $f$-DP hypothesis tests, we also developed a novel DP test for the difference-of-proportions, which was shown to have accurate type I errors and near optimal power. The introduction of CNDs also raises several questions:
    
     It was noted in Section \ref{s:canonical} that the CND is in general not unique for a given tradeoff function. While the construction in Definition \ref{def:CNDsynthetic} always results in a CND, {and has a simple sampling procedure,} it may not be the most natural CND. For example, when applied to the tradeoff function $G_1$, we see in Figure \ref{fig:cnd} that the CND constructed by Definition \ref{def:CNDsynthetic} has a non-differentiable pdf. On the other hand, $N(0,1)$ is also a CND for $G_1$ which has a smooth pdf. {One may wonder if} there a more natural construction of a CND which recovers $N(0,1)$ in the case of $G_1$, {and similarly, if} there ia a CND for $f_{\ep,\de}$ which has a continuous or smooth pdf. {A recent paper that builds upon the present work, \citet{awan2022logconcave}, partially answers these questions, showing that in some cases it is possible to construct a \emph{log-concave} CND, which recovers $N(0,1)$ in the case of $G_1$; surprisingly, \citet{awan2022logconcave} also show that the Tulap distribution is the \emph{unique} CND for $f_{\ep,0}$, ruling out the possibility of a smooth CND for $f_{\ep,0}$.}
     
     {Another question is whether there is a
     natural and meaningful extension of CNDs to vector-valued statistics. %In the case of GDP, the multivariate normal distribution has analogous properties as in Definition 3.1. Does there always exist such multivariate extensions of the CND? 
     The follow-up paper, \citet{awan2022logconcave}, partially answers this question, giving a definition of a multivariate CND and general constructions under various assumptions. While they show that there exists multivariate CNDs for many general classes of tradeoff functions, including GDP, Laplace-DP, and $(\ep,\de)$-DP, they also prove that there is \emph{no} multivariate CND for $f_{\ep,0}$. 
     
     While this paper focused on the connection between CNDs and private hypothesis tests, it is an open question whether there are other fundamental optimality properties of CNDs. It was also noted in the introduction that additive noise mechanisms often appear as a component of more complex DP mechanisms, and it is worth investigating whether CNDs can be used to optimize these other mechanisms for a particular $f$-DP guarantee.}
     
     %e showed that CNDs are closely connected to private $p$-values and optimal DP hypothesis tests. Are there other fundamental optimality properties of CNDs?
       
    The applications to DP hypothesis tests also raise many interesting questions. In general, there always exists a most powerful DP test for any composite null and simple alternative, as shown in Proposition \ref{prop:power}, which can be expressed as the solution to a convex optimization problem. However, solving the optimization problem is computationally burdensome for all but the simplest of problems. In Theorem \ref{thm:binary}, we were able to derive closed-form expressions for the most powerful DP tests. Do there exist closed-form expressions for other UMP DP tests to avoid computational optimization?
        
         We also introduced the semi-private framework which allowed us to derive an upper bound on the power of any unbiased $f$-DP test. Can this framework be applied to other DP testing problems to derive similar bounds? We also remarked that the semi-private framework may be useful to better understand the privacy guarantee of mechansisms where certain statistics are privatized, whereas others are reported exactly, such as by in the 2020 Decennial US Census -- it remains to be seen whether the semi-private framework can give new results or new understanding in these settings. 
        
        % In Section \ref{s:inversion}, we were able to use the Gil-Pelaez inversion formula to evaluate the cdf of the private test statistic. This avoided Monte Carlo sampling, improving both the running time and accuracy of the test. However, this approach relied on the fact that the test statistic was a linear combination of known distributions. Is this approach useful for other DP testing problems? For problems with more complex test statistics, are there similar methods of exactly evaluating the sampling distribution? 

\section*{Acknowledgements}
{This work was supported in part by Cooperative Agreement CB16ADR0160001 from the U.S. Census Bureau. The first author was also supported in part by NSF Award Numbers SES-1534433, SES-1853209, and SES-2150615, and is very grateful for the hospitality of the Center for Research on Computation and Society at Harvard University, where part of this work was completed. The second author was also supported in part by a Simons Investigator Award.}

%% if your bibliography is in bibtex format, uncomment commands:
\bibliographystyle{imsart-nameyear} % Style BST file (imsart-number.bst or imsart-nameyear.bst)
\bibliography{bibliography.bib}       % Bibliography file (usually '*.bib')

\pagebreak
%\maketitle
\setcounter{page}{1}
\thispagestyle{empty}
%%%%%%%%%%%%%%%%%%%%%%%%%%%%%%%%%%%%%%%%%%%%%%
%% Single Appendix:                         %%
%%%%%%%%%%%%%%%%%%%%%%%%%%%%%%%%%%%%%%%%%%%%%%
\begin{center}
    \huge Canonical Noise and Private Hypothesis Tests\\
    Supplementary Materials\\
    \large 
    Jordan Awan and Salil Vadhan
\end{center}
\begin{appendix}
%\section*{???}%% if no title is needed, leave empty \section*{}.
\section{Background on Hypothesis Testing}\label{s:testing}

In this section, we review the definitions of randomized hypothesis tests and $p$-values. %In this paper, we consider hypothesis tests from two perspectives. From one, we will be constructing privacy-preserving tests with desirable properties. On the other hand, the notion of privacy that we work with is also formulated in terms of constraints on hypothesis tests. 

\begin{defn}
  [Hypothesis Test]\label{HT} 
  Let $X\in \mscr X$ be distributed $X\sim P_\ta$, where $\ta\in \Ta$. Let $\Ta_0, \Ta_1$ be a partition of $\Ta$. A \emph{(randomized) test} of $H_0: \ta \in \Ta_0$ versus $H_1: \ta\in \Ta_1$ is a measurable function $\phi: \mscr X \rightarrow [0,1]$. We call $H_0: \theta\in \Theta_0$ the \emph{null hypothesis} and $H_1: \theta\in \Theta_1$ the \emph{alternative hypothesis}. We interpret the test $\phi(x)$ as the probability of rejecting the null hypothesis after observing $x\in \mscr X$. We say a test $\phi$ is at \emph{level} $\al$ if $\sup_{\ta\in \Ta_0} \EE_{P_\ta} \phi \leq \al$, and at \emph{size} $\al$ if $\sup_{\ta\in \Ta_0} \EE_{P_\ta} \phi = \al$. The size is also called the \emph{type I error} and represents the probability of mistakenly rejecting the null hypothesis. The \emph{power} of $\phi$ at $\ta$ is denoted $\beta_\phi(\ta) = \EE_{P_\ta}\phi$, which is the probability of rejecting when the true parameter is $\theta$. A test $\phi$ is \emph{unbiased} if $\beta_\phi(\ta_0)\leq \beta_\phi(\ta_1)$ for all $\ta_0\in \Ta_0$ and $\ta_1\in \Ta_1$; that is, the power is always higher at any alternative than at any null value.

Let $\Phi$ be a set of tests for $H_0: \ta \in \Ta_0$ versus $H_1:\ta\in \Ta_1$. We say that $\phi^*\in \Phi$ is the \emph{uniformly most powerful} (UMP) test among $\Phi$ at level $\al$ if it is level $\alpha$ and for any other level $\alpha$ test $\phi \in \Phi$, we have 
$\beta_{\phi^*}(\ta) \geq \beta_{\phi}(\ta)$, for all $\ta\in \Ta_1$. If  $\Theta_1$ has cardinality one, we simply say that $\phi$ is the \emph{most powerful test}.
\end{defn}

Classically, randomized tests appear in the Neyman-Pearson Lemma and their role in that setting is to allow a test to achieve a specified size. However, for privacy, we require that all of our tests are randomized, and use the randomness to achieve differential privacy. 

Usually, rather than a binary accept/reject decision from a randomized test, it is preferable to report a $p$-value, which gives a continuous measure of how much evidence there is for the alternative hypothesis over the null. Smaller values of $p$ give more evidence for the alternative. 

\begin{defn}[$p$-Value]
   Let $X\in \mscr X$ be distributed $X \sim P_\ta$, where $\ta\in \Ta$. Let $\Ta_0, \Ta_1$ be a partition of $\Ta$.
   Let $p$ be a random variable, taking values in $[0,1]$. Define $p(X)\defeq p|X$ to be the random variable $p$ conditioned on $X$. We say that $p$ is a \emph{$p$-value} for the test $H_0: \ta\in \Ta_0$ versus $H_1: \ta\in \Ta_1$ if
   \[\sup_{\ta_0\in \Ta_0} P_{\ta_0} (p(X)\leq \al) \leq \al,\]
   where the probability is over both $p$ and $X$. In other words, for every $\ta\in \Ta_0$, the distribution of $p(X)$ stochastically dominates $U(0,1)$.
 \end{defn}
 A $p$-value represents the probability of observing data as extreme or more extreme as the present sample, when the null hypothesis is true. Often the measure of ``extreme'' is based on a specific test statistic. A small $p$-value offers evidence that the present sample is unlikely to have been generated by the null model. 
 
Given a $p$-value $p(X)$, $\phi(X) = P(p(X)<\al\mid X)$ is a test for the same hypothesis, at level $\al$. For each $\al$, let $\phi_\al: \mscr X\rightarrow [0,1]$ be a test at level $\al$.  Let $U \sim U[0,1]$. Then $p(X) = \inf \{\al \mid \phi_\al\geq U\}$ is a $p$-value for the same test. See \citet{Geyer2005} for a deeper understanding of randomized tests, $p$-values, and confidence sets in terms of fuzzy set theory.

\section{Most powerful \texorpdfstring{$f$}{TEXT}-DP test as convex optimization}\label{s:convex}
In this section, we show that for an arbitrary null hypothesis, and a simple alternative hypothesis, there exists a most powerful $\alpha$-level $f$-DP test, which can be expressed as the solution to a convex optimization problem. This result is an extension of \citet[Remark 3.1]{awan2018differentially}, which showed that in the case of $(\ep,\de)$-DP the most powerful test is the solution to a linear program.

\begin{prop}\label{prop:power}
Let $\Theta$ be a set of parameters, and $\{P_\theta \mid \theta\in \Theta\}$ be a set of distributions on $\mscr X^n$. Let $\Theta_0\subset \Theta$ and $\theta_1\in \Theta\setminus \Theta_0$. Then a most powerful $\alpha$-level $f$-DP for $H_0: \theta\in \Theta_0$ versus $H_1: \theta=\theta_1$ is the solution to a convex optimization problem.
\end{prop}
\begin{proof}
First, note that the $f$-DP constraint on tests: $0\geq \phi(x')-1+f(\phi(x))$ is a convex constraint, since $f$ is convex. Furthermore, the type I error constraints $\EE_{P_{\theta_0}} \phi(x) \leq \alpha$ are linear and hence convex. The intersection of the privacy constraints and the type I error constraints is thus a convex set. This set is non-empty as the constant test $\phi(x)=c$ lies inside the set for all $c\in [0,\alpha]$. Finally, the power $\EE_{P_{\theta_1}}\phi(x)$ is a linear objective.
\end{proof}

\section{difference-of-proportions Non-private UMPU}\label{s:umpu}

Suppose we observe $X_i \iid \mathrm{Bern}(\theta_X)$ for $i=1,\ldots, n$ and $Y_j \iid \mathrm{Bern}(\theta_Y)$ for $j=1,\ldots, m$, and we wish to test $H_0: \theta_X\leq\theta_Y$ versus $H_1: \theta_X<\theta_Y$. Denote $X = \sum_{i=1}^n X_i$ and $Y = \sum_{j=1}^m Y_j$.
The joint distribution of $(X_i,Y_j)_{i,j}$ is
\begin{align*}
  f_{\ta_X,\ta_Y}(\ul x,\ul y)&= \prod_{i=1}^n \ta_X^{x_i}(1-\ta_X)^{1-x_i}\prod_{j=1}^m \ta_Y^{y_j}(1-\ta_Y)^{1-y_j}\\
                              &=\ta_X^{\sum_{i=1}^n x}(1-\ta_X)^{n-\sum_{i=1}^n x} \ta_Y^{\sum_{j=1}^m y} (1-\ta_Y)^{m-\sum_{j=1}^m y}\\
  &= (1-\ta_X)^n(1-\ta_Y)^{m} \exp\l(\sum_{i=1}^n x_i \log\l(\frac{\ta_X}{1-\ta_X}\r) + \sum_{j=1}^m y_j \log \l( \frac{\ta_Y}{1-\ta_Y}\r)\r)
\end{align*}

By relabeling $\eta_0=\log\left(\frac{\ta_X}{1-\ta_X}\right)$ and $\eta_1$ such that $\eta_1=\log\left(\frac{\ta_Y}{1-\ta_Y}\right)-\eta_0$, and setting $x = \sum_{i=1}^n x_i$ and $y = \sum_{j=1}^m y_i$, we can write 
\[f_{\ta_x,\ta_y}(\ul x, \ul y) = (1-\ta_X)^n(1-\ta_Y)^{m} \exp\l((x+y)\eta_0+y\eta_1\r),
\]
and from this expression we see that $\eta_0$ and $\eta_1$ are natural exponential family parameters for the sufficient statistics $\{(X+Y),Y\}$. We can also re-express our test $H_0: \theta_X\geq \theta_Y$ versus $H_1: \theta_X<\theta_Y$ as $H_0: \eta_1\leq 0$ versus $H_1: \eta_1>0$. This now fits the assumptions of \citet[Theorem 4.124]{schervish2012theory}. Since $Y\mid X+Y=z$ has a monotone likelihood ratio in $\eta_1$, we know that the UMP unbiased test for the above hypothesis is of the form
\[\phi(X,Y) = \begin{cases}
    0&Y<c\\
    a&Y=c\\
    1&Y>c
  \end{cases}\]
where $c$ and $a$ depend on the value of $X+Y=z$, and are chosen such that $\EE_{\ta_X=\ta_Y}(\phi\mid X+Y=z)=\al$.

In fact there is a more convenient formulation of this test, which gives exact $p$-values. First note that $\phi(X,Y)$ can be written in the form $\phi(X,Y) =F_U(Y-c')$, where $F_U(\cdot)$ is the cdf of $U\sim \mathrm{Unif}(-1/2,1/2)$, and $c'$ is a real number, which depends on $z=X+Y$. Then we can write
\begin{align*}
  \phi(X,Y)&= F_U(Y-c')\\
           &=P(U\leq Y-c'(z)\mid X,Y)\\
  &= P(c'\leq Y+U\mid X,Y)
\end{align*}
From the last equality, we see that the UMPU test depends on the (random) test statistic $T=Y+U$, and on the value $z=X+Y$. The $p$-value corresponding to $T$ is
\begin{align*}
  p&=P(Y+U\geq T\mid T,X+Y=z)\\
   &=P(Y-T\geq U\mid T,X+Y=z)\\
   &= \EE[F_U(Y-T)\mid T,X+Y=z]
\end{align*}
This $p$-value can be computed fairly efficiently, since the expected value is over the $m$ hypergeometric values of $Y$ given $z$. Lastly, to check that this $p$-value agrees with the UMPU, we want to show that $P(p(T,z)\leq \al\mid T,X+Y=z)=\phi(X,Y)$. To this end,
\begin{align*}
  P(p(T,z)\leq \al\mid T,X+Y=z)&= P(1-F_{Y+U\mid z}(T)\leq \al\mid T,z)\\
                               &=P(1-\al \leq F_{T\mid z}(T)\mid T,z)\\
                               &=P(F^{-1}_{T\mid z}(1-\al)\leq T\mid T,z)\\
                               &=P(F^{-1}_{T\mid z}(1-\al) \leq Y+U\mid Y,z)\\
                               &=P(c\leq Y+U\mid Y,z)\\
                               &= \phi(X,Y)
\end{align*}
where $c=F^{-1}_{T\mid z}(1-\al)$ is a constant, which only depends on $z$.

\section{Non-Existence of UMPU in difference-of-proportions}\label{s:noUMPU}

In this section, we give a simple example demonstrating that there is no UMP unbiased $f$-DP test for the problem of Section \ref{s:application}. In particular, we work with $(\ep,0)$-DP.
\begin{comment}{\begin{thm}
  [Non-existence of DP UMPU]
  Let $X_i \iid \mathrm{Bern}(\ta_X)$ and $Y_j \iid \mathrm{Bern}(\ta_Y)$ for $i=1,\ldots, n$ and $j=1,\ldots, m$. There is no UMPU test $\phi(x,y)$, among those satisfying inequalities \eqref{eq:DPIneq}, for the hypothesis $H_0: \ta_X=\ta_Y$ versus $H_1: \ta_X<\ta_Y$.
\end{thm}
\begin{proof}
  Suppose that there exists such a test, which we call $\phi^*(x,y)$. Then $\phi^*$ is the solution to the linear program
  \begin{itemize}
  \item $ \arg\max_{\phi(x,y)} \EE_{\substack{X\sim \mathrm{Binom}(n,\ta_X)\\Y\sim \mathrm{Binom}(m,\ta_Y)}} \phi(X,Y)$
    
  \item $  \text{s.t.  inequalities \eqref{eq:DPIneq} and } \EE_{Y\sim \mathrm{Hyper}(m,n,z)} (\phi(z-Y,Y))=\al,\text{ for all } z=0,\ldots, m+n$,
  \end{itemize}
  for all values $\ta_X<\ta_Y$. However, different values of $\ta_X$ and $\ta_Y$ lead to different most powerful unbiased tests. See Example \ref{ex:noUMPU} for an example.
\end{proof}}\end{comment}

%\begin{example}[No DP-UMPU]\label{ex:noUMPU}
Suppose that $m=1$ and $n=2$. Then $Y\sim \mathrm{Binom}(1,\thy)$ and $X\sim \mathrm{Binom}(2,\thx)$. we consider unbiased tests which satisfy $(1,0)$-DP, at level $.05$. Equation \eqref{eq:hyper} imposes the following constraints on a test $\phi(x,y)$:
\begin{equation}\label{eq:neyman}
\begin{split}
    \phi(2,1)&=.05\\
    \phi(0,0)&=.05\\
    (1/3)\phi(0,1) + (2/3)\phi(1,0)&=.05\\
    (2/3)\phi(1,1) + (1/3)\phi(2,0)&=.05
    \end{split}
\end{equation}

1) Suppose that $\thx=0$ and $\thy=1$. The following test maximizes the power in this case
\[\left(\begin{array}{c|ccc}
y=0&.05&.05e^{-1}&.05\\
y=1&.15 - .1e^{-1}&.05&.05\\\hline
x:&0&1&2
\end{array}\right)
\approx \left(\begin{array}{c|ccc}
y=0&.05&.0184&.05\\
y=1&.1132&.05&.05\\\hline
x:&0&1&2
\end{array}\right)\]
We can see that maximizing the power is equivalent to maximizing the value of $\phi(0,1)$, as $P(X=0,Y=1)=1$. Increasing $\phi(0,1)$ any further, would require decreasing $\phi(1,0)$. But for privacy we require $\phi(1,0)\geq \exp(-1) \phi(0,0)$, which is tight. The other privacy constraints can be easily verified. So, the above test is the most powerful unbiased test for $\thx=0$, $\thy=1$.

2) Suppose that $\thx=1/2$ and $\thy=1$, and consider the following test: 

\[\left(\begin{array}{c|ccc}
y=0&.05&.075e^{-1}-.025e^{-2}&e^{-1}.05\\
y=1&.15(1-e^{-1})+.05e^{-2} &.075-.025e^{-1}&.05\\\hline
x:&0&1&2
\end{array}\right)\]
\[\approx 
\left(\begin{array}{c|ccc}
y=0&.05&.0242&.0184\\
y=1&.1016&.0658&.05\\\hline
x:&0&1&2
\end{array}\right)\]
It can be verified that this test satisfies the constraints of Equation \eqref{eq:neyman} as well as the $\ep$-DP constraints. Note that the power formula for $\thx=1/2$ and $\thy=1$ is
\[.5^2[\phi(0,1) + 2\phi(1,1) + \phi(2,1)],\]
and we see that the test above has higher power compared to the test from part 1). Since the test in part 1) was most powerful unbiased 1-DP test for $\thx=0$ and $\thy=1$, but it is not most powerful unbiased 1-DP test for $\thx=1/2$ and $\thy=1$, we conclude that there is no uniformly most powerful unbiased 1-DP test in this setting. 
%\end{example}

% \lstinputlisting{./code/funRegrExamples.R}

% \bibliographystyle{plain} 
% \bibliography{biblio-Apoly} 
\section{Alternative DP tests}\label{s:otherTests}
In this section, we describe the other DP tests that appear in the simulations of Section \ref{s:simulations}. 

In Algorithm \ref{alg:normalDP}, we describe an $\ep$-DP normal approximation test, proposed by \citet{karwa2018correspondence}. This test is analogous to the test of a single population proportion described in \citet{vu2009differential}, where a normal approximation with inflated variance is used to approximate the sampling distribution. This test adds independent Laplace noise to $X$ and $Y$, bases the test statistic on the difference of the estimated proportions, estimates the variance of the test statistic using a plug-in estimate, and then approximates the sampling distribution of the test statistic as normal. In Algorithm \ref{alg:normalDP}, $\Phi$ denotes the cdf of $N(0,1)$.

\begin{algorithm}[H]
\SetAlgoLined
 \KwData{Let $X$, $Y$, $m$, and $n$ be given. Let $\ep>0$ be given.}
%\KwResult{$\ep$-DP $p$-value and private summary values}
 Draw $L_1,L_2 \iid \mathrm{Laplace}(0,1/\ep)$\;
 Set $\twid X=X+L_1$ and $\twid Y = Y+L_2$\;
 Set $\twid \theta = \min\{\max\{\frac{\twid X + \twid Y}{m+n},0\},1\}$\;
 Set $T = \twid Y/m-\twid X/n$\;
 Set $\mathrm{var} = \twid \theta(1-\twid \theta) + \frac{2}{(m\ep)^2} + \frac{2}{(n\ep)^2}$\;
  \KwResult{ $p$-value: $p= 1-\Phi(T/\sqrt{\mathrm{var}})$, $\twid X$, $\twid Y$.}
 \caption{$\ep$-DP Normal approximation $p$-value.}
 \label{alg:normalDP}
\end{algorithm}

Another DP test would be to take the semiprivate UMPU test of Section \ref{s:semiprivate}, and using \emph{composition}, produce both a privatized test statistic as well as a private estimate of the value $Z=X+Y$. Then plugging in the estimate of $Z$ gives a fully $\ep$-DP version of the semiprivate test. The full procedure is described in Algorithm \ref{alg:plugin}. In Algorithm \ref{alg:plugin}, $F$ represents the cdf of the variables $\mathrm{Tulap}(0,\exp(-\ep/2),0)$. 

\begin{algorithm}[H]
\SetAlgoLined
 \KwData{Let $X$, $Y$, $m$, and $n$ be given. Let $\ep>0$ be given.}
%\KwResult{$\ep$-DP $p$-value and private summary values}
 Draw $L_1,L_2 \iid \mathrm{Tulap}(0,\exp(-\ep/2),0)$\;
 Set $\twid T=(Y-X)+L_1$ and $\twid Z=(X+Y)+L_2$\;
 Set $T = \twid Y/m-\twid X/n$\;
  \KwResult{ $p$-value: $p= \EE_{H\sim \mathrm{Hyper}(m,n,\twid Z)}F(2H-\twid Z)$}
 \caption{$\ep$-DP plug-in $p$-value.}
 \label{alg:plugin}
\end{algorithm}

  %%%%%%%%%%%%%%%%%%%%%%%%%%%%%%%%%%%%%%%%%%%%%%%%%%%%%%%%%%%%%%%%%%%%%%%
  %%%   PROOFS
  %%%%%%%%%%%%%%%%%%%%%%%%%%%%%%%%%%%%%%%%%%%%%%%%%%%%%%%%%%%%%%%%%%%%%%%
  \section{Proofs}\label{s:proofs}
  {
  The Galois inequalities are a well-known property of cdfs and their quantile functions. We include a short proof for completeness.
  \begin{lemma}[Galois Inequalities]\label{lem:galois}
  Let $F$ be a cdf and $F^{-1}(p) \defeq \inf \{x\mid p\leq F(x)\}$ be its quantile function. Then $F^{-1}(p)\leq x$ if and only if $p\leq F(x)$. 
  \end{lemma}
  \begin{proof}
  Suppose that $p\leq F(x)$. This holds if and only if $x\in \{t\mid p\leq F(t)\}$. Since $F$ is monotone increasing, $\{t\mid p\leq F(t)\}$ is of the form $[F^{-1}(p),\infty)$. Therefore, $x\in \{t\mid p\leq F(t)\}$ holds if and only if $x\geq \inf \{t\mid p\leq F(t)\}=F^{-1}(p)$.
  \end{proof}
  
  Let $F$ be a cdf. We say that $F$ is \emph{invertible} at $t\in \RR$ if $F^{-1} \circ F(t)=t$. Similarly, we say that a symmetric tradeoff function $f$ is \emph{invertible} at $\alpha\in [0,1]$ if $f^{-1}\circ f(\alpha) = \alpha$ (or equivalently, $f\circ f(\alpha)=\alpha$).

  %a monotone function. We call $F$ \emph{invertible} at $t\in \RR$ if $F(x)=F(t)$ implies that $x=t$. We will use this terminology for both tradeoff functions and cdfs. Notice that for either tradeoff functions or cdfs, if $F$ is invertible at $t$, then $F^{-1}\circ F(t)=t$. 
  
  \begin{lemma}\label{lem:invertible}
  Let $f$ be a symmetric tradeoff function and let $F$ be a cdf. Then 
  \begin{enumerate}
      \item $f$ is invertible for all $\alpha\in [0,f(0)]$ and $f\circ f(\alpha)\leq \alpha$ otherwise.
      \item if $F$ is continuous, then $F\circ F^{-1}(p)=p$ for all $p\in [0,1]$,
      \item $F^{-1}\circ F(t)\leq t$. 
  \end{enumerate}
  \end{lemma}
  \begin{proof}
  \begin{enumerate}
      \item Since, $f^{-1}(\alpha) = \inf \{t \mid f(t)\leq \alpha\}$, we have that $f^{-1}(f(\alpha)) = \inf \{t\mid f(t)\leq f(\alpha)\}$. We notice that $\alpha\in \{t\mid f(t)\leq f(\alpha)\}$, and so we have that $f^{-1}(f(\alpha)) = \inf \{t\mid f(t)\leq f(\alpha)\}\leq \alpha$. 
      
      Next, notice that because $f$ is a tradeoff function, it is convex, decreasing, and $f(1)=0$. This implies that the only possibility for $f(a)=f(b)$ is either $a=b$ or $f(a)=f(b)=0$. Hence, $f$ is invertible on $[0,f^{-1}(0)]=[0,f(0)]$. 
      
      \item Since $F$ is continuous, by the Intermediate Value Theorem, for any $p\in (0,1)$, there exists $x\in \RR$ such that $p=F(x)$. So, we can write $F^{-1}(p) = \inf\{t\mid p=F(t)\}$. Since $F$ is continuous,  $\{t\mid p=F(t)\}$ is a closed set, and so we have that $F^{-1}(p)\in \{t\mid p=F(t)\}$. The result follows:  $F(F^{-1}(p))=p$. If $F^{-1}(0)=-\infty$ or $F^{-1}(1)=\infty$, we can allow $F$ to take as input $-\infty$ and $+\infty$, with $F(-\infty)=0$ and $F(\infty)=1$. Then we get $F(F^{-1}(p))=p$ when $p\in \{0,1\}$ as well. %\todo{change definition? Or restrict $p$?} 
      \item Note that $F^{-1}\circ F(t) = \inf \{x\mid F(t) \leq F(x)\}$. We see that $t\in \{x\mid F(t)\leq F(x)\}$, so $F^{-1}\circ F(t) = \inf \{x\mid F(t) \leq F(x)\}\leq t$. 
  \end{enumerate}
  \end{proof}
  
  Lemma \ref{lem:support} is a technical lemma establishing that the patterns of invertibility/non-invertibility of a CND $F$ satisfy a recurrence.
  
  \begin{lemma}\label{lem:support}
  Let $F$ be a CND for some tradeoff function $f$, and call $(-U,U) = [F^{-1}(0),-F^{-1}(0)]$ (which may be $(-\infty,\infty)$). Suppose for contradiction that there exists $-U<a<b<U$ such that $F(a)=F(b)$. Then $F(a+1)=F(b+1)$ and $F(a-1)=F(b-1)$.
  \end{lemma}
  \begin{proof}
  we use a different parametrization for the proof. Suppose that there exists $-U<a-1<b-1<U$ such that $F(a-1)=F(b-1)$ but that $F(a)<F(b)$. Then there exists $[c,d]\subset [a,b]$ such that $F$ is invertible for all $t\in [c,d]$, since $F$ is a continuous cdf. To see this, let $A = \{t\mid \exists x\neq t \text{ s.t. } F(x)=F(t)\}$, which we identify as the union of disjoint closed sets. Then $\RR\setminus A$ is an open set. Because $F$ is continuous, $\RR\setminus A$ is non-empty, and so there exists the desired interval $[c,d]\subset \RR\setminus A$.

  Call $\alpha_c=1-F(c)$ and $\alpha_d=1-F(d)$. Then $\alpha_c>\alpha_d$. Note that 
  \begin{align*}
      F^{-1}(1-\alpha_c)&= F^{-1}(F(c))=c,\\
      F^{-1}(1-\alpha_d)&= F^{-1}(F(d))=d,
  \end{align*}
  since $F$ is invertible on $[c,d]$. Then $f(\alpha_c)= F(F^{-1}(1-\alpha_c)-1)=F(c-1)$ and $f(\alpha_d)=F(F^{-1}(1-\alpha_d)-1)=F(d-1)$ because $F$ is a CND for $f$, and using the above identities. However, since $F(a-1)=F(b-1)$, and $F$ is monotone, we have that $F(c-1)=F(d-1)$, which implies that $f(\alpha_c)=f(\alpha_d)$. Earlier we noted that $\alpha_c>\alpha_d$; since tradeoff functions are decreasing and convex, the only possibility for $f(\alpha_c)=f(\alpha_d)$ is for $f(\alpha_c)=f(\alpha_d)=0$. Now let $t\in [c,d]$. Just as above, we denote $\alpha_t = 1-F(t)$, which satisfies $t=F^{-1}(1-\alpha_t)$ and $f(\alpha_t)=0$ just like above. %Then we have that $0=f(\alpha_t)=F(F^{-1}(1-\alpha_t)-1)$, which implies that $F^{-1}(1-\alpha_t)-1\leq -U$ or equivalently $F^{-1}(1-\alpha_t)\leq 1-U$. By Lemma \ref{lem:galois}, this holds if and only if $(1-\alpha_t)\leq F(1-U)$. 
Consider the following inequalities, where each line is implied by the line above:

\begin{align}
    F(F^{-1}(1-\alpha_t)-1)&=f(\alpha_t)=0\label{eq:support1}\\
    F^{-1}(1-\alpha_t)-1&\leq -U\label{eq:support2}\\
    F^{-1}(1-\alpha_t)&\leq 1-U\label{eq:support3}\\
    (1-\alpha_t)&\leq F(1-U)\label{eq:support4}\\
    1-F(1-U)&\leq \alpha_t=1-F(t)\label{eq:support5}\\
    F(t)&\leq F(1-U)\label{eq:support6}\\
    t&\leq F^{-1} (F(1-U))\leq 1-U,\label{eq:support7}
\end{align}
where \eqref{eq:support1} used our observation that $f(\alpha_c)=f(\alpha_d)=0$, the monotonicity of $f$, and the fact that $F$ is a CND for $f$, \eqref{eq:support2} used the fact that $-U=F^{-1}(0)$ and part 3 of Lemma \ref{lem:invertible}, Lemma \ref{lem:galois} gives \eqref{eq:support4}, we used the fact that $\alpha_t = 1-F(t)$ for \eqref{eq:support5}, and part 3 of Lemma \ref{lem:invertible} for \eqref{eq:support7}. We then have that $t-1\leq -U$, which implies that $a-1\leq c-1\leq t-1\leq -U$, which contradicts the assumption that $-U<a-1$. We conclude that $F(a)=F(b)$. Using the parametrization in the Lemma statement, we have that if there exists $-U<a<b<U$ such that $F(a)=F(b)$, then $F(a+1)=F(b+1)$.

Now suppose that $-U<a<b<U$ such that $F(a)=F(b)$, and we will show that $F(a-1)=F(b-1)$. By symmetry, we have that $F(-b)=F(-a)$ and $-U<-b<-a<U$. By the above work, we have that $F(-b+1)=F(-a+1)$. Applying symmetry again, we have $F(b-1)=F(a-1)$, which establishes the result. 
  \end{proof}
  }
   %%%%
   \lemrecurrence*
   \begin{proof}%[Proof of Lemma \ref{lem:recurrence}.]
   {
  we prove the first recurrence in detail and remark that the second recurrence is obtained by a similar argument. We know that $f(\alpha) = F(F^{-1}(1-\alpha)-1)$ for all $\alpha \in (0,1)$. Assume that $F(x-1)\in (0,1)$. If $F$ is invertible at $1-x$, then plugging in $\alpha=F(x-1)$ gives
  \begin{align*}
      f(F(x-1))&= F(F^{-1}(1-F(x-1))-1)\\
      &=F(F^{-1}(F(1-x))-1)\\
      &=F(-x)\\
      &=1-F(x),
  \end{align*}
  where we used the fact that $F$ is symmetric and that $F$ is invertible at $1-x$. 
  
  If $1-x$ is not an invertible point of $F$, then it lies in the interval $[a,b]\defeq [\inf \{t\mid F(1-x)= F(t)\},\sup\{t\mid F(1-x)=F(t)\}]$. Note that $F$ is constant on this interval, and $F$ is invertible at $a$. By Lemma \ref{lem:support}, we have that $F$ is also constant on the interval $[a-1,b-1]$, which contains $-x$. Then 
  
 \begin{align}
     f(F(x-1))&=f(1-F(1-x))\\
     &=f(1-F(a))\label{eq:recurrence1}\\
     &=F(F^{-1}(F(a))-1)\label{eq:recurrence1.5}\\
     &=F(a-1)\label{eq:recurrence2}\\
     &=F(-x)\label{eq:recurrence3}\\
     &=1-F(x),
 \end{align}
 where for \eqref{eq:recurrence1} we use the fact that $F(1-x)=F(a)$, for \eqref{eq:recurrence1.5} we use the fact that $F$ is a CND for $f$, for \eqref{eq:recurrence2} we use the invertibility of $F$ at $a$, and for \eqref{eq:recurrence3} we use the facts that $-x\in [a-1,b-1]$ and $F$ is constant on $[a-1,b-1]$. 
 
 Finally, for the case that $F(x-1)=1$, we have that $F(x)=1$ since $F$ is increasing, and $1-f(F(x-1))=1-f(1)=1-0=1$ and we see that the recurrence holds for this case as well.
%The other recurrence is obtained by a similar argument.
  }
\end{proof}

\begin{lemma}\label{lem:contraction}
Let $f$ be a nontrivial symmetric tradeoff function, and let $c$ be the fixed point of $f$. Then 
\begin{enumerate}
    \item $c\in [0,1/2)$,
    \item iteratively applying $1-f(\cdot)$ to any point in $(c,1]$ approaches 1 in the limit. 
\end{enumerate}
\end{lemma}
\begin{proof}
{Since $f(\alpha)<1-\alpha$ for some $\alpha$, it follows that $c<1/2$: Suppose to the contrary that $f(1/2)=1/2$: by symmetry and convexity, $f(\alpha)=1-\alpha$ for all $\alpha$.
          
          Then the function $1-f(\cdot)$ is concave, increasing, and has slope $< 1$ on the set $(c,1)$. By the mean value inequality \citep[Proposition 1.3, Chapter 6.1]{shifrin2005multivariable}, we have that $1-f(\cdot)$ is a contraction map on $(c,1]$. By the contraction mapping theorem \citep[Theorem 1.2, Chapter 6.1]{shifrin2005multivariable}, $1-f(\cdot)$ has a unique fixed point on $(c,1]$.  By definition of $f$ as a tradeoff function, we know that $1-f(1)=1$, so the value $1$ must be the unique fixed point. The contraction mapping theorem also tells us that iteratively applying $1-f(\cdot)$ to any point in $(c,1]$ approaches the fixed point 1 in the limit.}
\end{proof}

\proptradeoff*
\begin{proof}%[Proof of Proposition \ref{prop:tradeoff}]%%% proof
  \begin{enumerate}
      \item First we will show that $F$ is a continuous cdf, which represents a symmetric random variable. We need to verify the following properties:  
      \begin{enumerate}
      %AAA
      \item $F(x)$ takes values in $[0,1]$: 
      
      First note that $c\in [0,1]$, since $f:[0,1]\rightarrow [0,1]$. Then $F(x) \in [0,1]$ for $x\in [-1/2,1/2]$. Finally, as $f$ takes values in $[0,1]$, by the recurrence relation {of Definition \ref{def:CNDsynthetic}}, we have that $F(x)\in [0,1]$ for all $x\in \RR$.
      %BBB
      \item $F(x) = 1-F(-x)$: 
      
      For values $x\in [-1/2,1/2]$, it is easy to verify that $F(x)=1-F(-x)$. Now assume that the relation $F(x) = 1-F(-x)$ holds on an interval $[-a,a]$ for some $a\geq 1/2$. Let $x\in [a,a+1]$. Then $F(x) = 1-f(F(x-1))=1-f(1-F(-x+1))=1-F(-x)$. By a symmetric argument, we have that the relation now holds on $[-a-1,a+1]$. By induction, we conclude that $F(x)=1-F(-x)$ on $\RR$. 
      %CCC
      \item $F(x)$ is continuous: 
      
      First note that $F(x)$ is continuous on $(-1/2,1/2)$. Next, as $f$ is convex, it is continuous on $(0,1)$. So, we have that $F(x)$ is continuous everywhere except potentially at half-integer values. We can verify that $F$ is continuous at 1/2: $\lim_{x\uparrow 1/2} F(x) = F(1/2)=1-c$, where as $\lim_{x\downarrow 1/2} F(x) = \lim_{x\downarrow 1/2} 1-f(F(x-1))=1-f(c)=1-c$, where we used the fact that $f$ is continuous on $(0,1)$. Now, assume that $F$ is continuous at a half integer value $x\geq 1/2$. Then $\lim_{y\rightarrow x+1} F(y) = \lim_{y\rightarrow x+1} 1-f(F(y-1))=1-f(F(x))=F(x+1)$. By induction, continuity holds on $[0,\infty)$. By symmetry (b), we have continuity on $\RR$. 
      %DDD
      \item $F'(x)$ is defined almost everywhere, and $F(x)$ is increasing: 
      
      Note that $F(x)$ is differentiable and $F'(x)>0$ on $(-1/2,1/2)$. As $f$ is convex it is  differentiable almost everywhere on $(0,1)$. Applying the recurrence relation, we have that $F(x)$ is differentiable a.e. on  $(n-1/2,n+1/2)$ for all $n\in \ZZ$. We conclude that $F(x)$ is differentiable a.e., as a countable union of measure zero sets has measure zero. As $F(x)$ is continuous, it suffices to verify that $F'(x)\geq 0$ almost everywhere. Let $x\geq 1/2$ such that $x\in \RR\setminus(\ZZ+1/2)$ and both $F'(x-1)$ and $f'(F(x-1))$ are defined. Then $\frac d{dx} F(x) = \frac d{dx} (1-f(F(x-1)))=-f'(F(x-1)) F'(x-1)$. For induction, we assume that $F'(x-1)\geq 0$. As $f'(y)\leq 0$ for all $y \in (0,1)$ where $f'$ is defined, we have that $F'(x)\geq 0$. By symmetry and induction, we have that $F'(x)\geq 0$ almost everywhere. Thus, $F$ is increasing. 
      %EEE
      \item $\lim_{x\rightarrow \infty} F(x)=1$ and $\lim_{x\rightarrow -\infty} F(x)=0$: 

      By symmetry, it suffices to show that $\lim_{x\rightarrow \infty} F(x)=1$. {By Lemma \ref{lem:contraction}, we have that $[1-c,1]\subset (c,1]$, and that iteratively applying $1-f(\cdot)$ to any point in $[1-c,1]$ approaches the fixed point 1 in the limit. 
      }

      %Since $f(\alpha)<1-\alpha$ for some $\alpha$, it follows that $c<1/2$. (Suppose to the contrary that $f(1/2)=1/2$: by symmetry and convexity, $f(\alpha)=1-\alpha$ for all $\alpha$.)
          
       %   Then the function $1-f(\cdot)$ is concave, increasing, and has slope $< 1$ on the set $(1-c,1)$. By the mean value inequality \citep[Proposition 1.3, Chapter 6.1]{shifrin2005multivariable}, we have that $1-f(\cdot)$ is a contraction map on $[1-c,1]$. By the contraction mapping theorem \citep[Theorem 1.2, Chapter 6.1]{shifrin2005multivariable}, $1-f(\cdot)$ has a unique fixed point on $[1-c,1]$.  By definition of $f$ as a tradeoff function, we know that $1-f(1)=1$, so the value $1$ must be the unique fixed point. The contraction mapping theorem also tells us that iteratively applying $1-f(\cdot)$ to any point in $[1-c,1]$ approaches the fixed point 1 in the limit. 
          
          Consider the sequence $Y_n$, where $Y_1=1-c$, and $Y_n = 1-f(Y_{n-1})$. Note that $Y_n = F(n-1/2)$ for $n \in \ZZ^+$. Then as $F$ is bounded and increasing, $\lim_{x\rightarrow \infty} F(x) = \lim_{n\rightarrow \infty} Y_n=1$, where in the last equality, we used the fact that $Y_n$ is constructed by iteratively applying $1-f(\cdot)$ and applying {Lemma \ref{lem:contraction}}. 
  \end{enumerate}
      \item First we will check that $1-f(F(-1/2)) = 1-f(c)=1-c=F(1/2)$. Let $x\in \RR$ such that $F(x-1)>0$. If $x\geq 1/2$ then we have $F(x)=1-f(F(x-1))$ by construction. If $x<1/2$, then $x-1<-1/2$. So, $F(x-1)=f(1-F(x))$ {by construction}. {We will apply $f$ to both sides of this last equation. By part 1 of Lemma \ref{lem:invertible}, to justify that $f\circ f(1-F(x))=1-F(x)$, we need to show that $1-F(x)\leq f(0)$. However, if $1-F(x)>f(0)$, then $F(x-1)=f(1-F(x))=0$, which contradicts our earlier assumption. So, applying $f$, we obtain $f(F(x-1))=1-F(x)$.}
      %Since $x<1/2$ and $F$ is increasing, we have that $0<F(x-1)\leq F(x)< F(1/2)=1-c\leq 1$ so, $1-F(x)\in (0,1)$. 
      %This justifies that $f(F(x-1))=1-F(x)$ and so $F(x) = 1-f(F(x-1))$. 
      The other recurrence holds by a similar argument due to the symmetries of $F$ and $f$. 
      %First we will check that $1-f(F(-1/2)) = 1-f(c)=1-c=F(1/2)$. We now have that for every point $x\in \RR$, either $F(x) = 1-f(F(x-1))$ or $F(x) = f(1-F(x+1))$. In the first case, assume that $F(x) = 1-f(F(x-1))$ and that $F(x-1)>0$. Then we have that $1-F(x)=f(F(x-1))$ implying $f^{-1}(1-F(x))=f^{-1}\circ f(F(x-1))$. By symmetry of $f$, we have that $f^{-1}(\alpha)=f(\alpha)$ for all $\alpha \in [0,1]$. Thus, $F(x-1)=f(1-F(x))$. The argument for the other case is analogous. 
      \item By symmetry, it suffices to check only that $F'(x)$ is decreasing on $(-1/2,\infty)$. Note that it holds trivially on $(-1/2,1/2)$, since $F$ is linear on this interval. For $x\geq 1/2$ a half integer, $F'$ is decreasing on $(x,x+1)$, since $1-f(\cdot)$ is a concave function.
      
      At $1/2$, we check the two limits $\lim_{y\uparrow 1/2} F'(y) = -c+1-c=1-2c$ and $\lim_{y\downarrow 1/2} F'(y) = \left[\lim_{y\downarrow 1/2}-f'(F(y-1))\left]\right[\lim_{y\downarrow 1/2} F'(y-1)\right]=\lim_{z\downarrow c} -f'(z)(1-2c)\leq 1-2c$, where we use the fact that $-f'(y)\leq 1$ for all $y\geq c$, since $c$ is the point of symmetry. 
      
      Now for induction, suppose that for some half integer $x\geq 1/2$, we have $\lim_{y\downarrow x} F'(y)\leq \lim_{y\uparrow x} F'(y)$. Then
      \begin{align*}
          \lim_{y\uparrow x+1} F'(y)&=\lim_{y\uparrow x} F'(y+1)\\
          &= \lim_{y\uparrow x} -f'(F(y))F'(y)\\
          &=\lim_{y\uparrow x} -f'(F(y)) \lim_{y\uparrow x} F'(y)\\
          &\geq \lim_{y\downarrow x} -f'(F(y)) \lim_{y\downarrow x} F'(y)\\
          &=\lim_{y\downarrow x} -f'(F(y))F'(y)\\
          &=\lim_{y\downarrow x} F'(y+1)\\
          &=\lim_{y\downarrow x+1} F'(y),
      \end{align*}
      where the inequality used the inductive hypothesis as well as the fact that $-f'(F(y))$ is positive and decreasing in $y$. 
      \item We will establish that $F$ is strictly increasing within its support $\{x\mid F(x)\in (0,1)\}$. It is non-decreasing and continuous by property 1. Suppose that $F$ is constant on an interval $(a,b)\subset\RR$. Then $F'(x)=0$ on $(a,b)$. By construction, we know that $F$ is strictly increasing on $(-1/2,1/2)$. By symmetry of $F$, we may assume that $(a,b)\subset (1/2,\infty)$.  However, property 3 states that $F'(x)$ is weakly decreasing on $(1/2,\infty)$. But then $F$ must be constant on $(a,\infty)$. This implies that $(a,b)\not\subset \{x\mid F(x)\in (0,1)\}$.    \qedhere
        \end{enumerate}
  \end{proof}

    Lemma \ref{lem:symTrade} is a technical lemma that is important for the proof of Theorem \ref{thm:canonical}.
  
  \begin{lem}\label{lem:symTrade}
    Let $X\sim F$ be a real-valued continuous random variable which is symmetric about zero. Let $m>0$ be given.     Let $f$ be an arbitrary symmetric tradeoff function. Then %\todo{fix}
    \begin{enumerate}
        \item $T(F(\cdot),F(\cdot -m)) = T(F(\cdot -m),F(\cdot))$ or equivalently $T(X,X+m)=T(X+m,X)$,
        \item to verify $f\leq T(F(\cdot),F(\cdot-m))$, it suffices to check that $f(\EE_F\phi(X))\leq 1-\EE_{F(x-m)} \phi(X)$, where $\phi(x)$ is either of the form $I(\frac{F'(x-m)}{F'(x)}>k)$ for $k\geq 1 $ or $I(\frac{F'(x-m)}{F'(x)}\geq k)$ for $k\geq 1$. 
    \end{enumerate}
  \end{lem}
   \begin{proof}%[Proof of Lemma \ref{lem:symTrade}.]
  \begin{enumerate}
    \item Note that the mapping $g(t) = -t+m$ is a bijection, hence applying it to both entries preserves the tradeoff function: 
    $T(X,X+m)=T(-X+m,-X)=T(X+m,X),$ where the last equality uses the fact that $-X\overset d =X$, as $X$ is symmetric. 
    \item  Since $F$ is a continuous cdf, the derivative $F'$ is defined almost everywhere, and $F'$ is a pdf for $F$. We will denote $F_m(x) = F(x-m)$ and $F_0(x) = F(x)$. 
    Then, the Neyman-Pearson Lemma tells us that the optimal test is $\phi^*(x) = {(1-\alpha)} \phi_{>k}(x)+{\alpha} \phi_{\geq k}(x)$, for some $\alpha\in [0,1]$ and $k\in \RR^{\geq 0}$, where $\phi_{>k}(x) = I\l(\frac{F'_m(x)}{F'_0(x)}>k\r)$ and $\phi_{\geq k}(x)=I\l(\frac{F'_m(x)}{F'_0(x)}\geq k\r)$. The points of the tradeoff function consist of the type I and type II error of the test $\phi^*$.
    {Since both type I and type II error of $\phi^*$ are linear in $\alpha$}, we have that the tradeoff function is linear {on the interval $[\EE_F\phi_{>k},\EE_F\phi_{\geq k}]$}. 
    %between the points corresponding to $\phi_{>k}$ and $\phi_{\geq k}$. 
    Since {the tradeoff function $f$ is }a convex function, {it} is upper bounded by secant lines; {so,} it suffices to check that $f(\EE_F\phi(X))\leq 1-\EE_{F(\cdot-m)} \phi(X)$ for $\phi=\phi_{>k}$ or $\phi=\phi_{\geq k}$. 
    
    Next, we argue that we need only consider $\phi_{>k}$ for $k\geq 1$ and $\phi_{\geq k}$ for $k>1$. We know from part 1 that $T(F(x),F(x-m))$ is symmetric. So, we need to verify that these tests fully specify the tradeoff function up to the point of symmetry, or equivalently up until the fixed point of $T(F(x),F(x-m))$. {Call
    \begin{align*}
        p &= P_{F_0}\left(\frac{F'_m(x)}{F'_0(x)}=1\right) \\
        &=\int I\left(\frac{F'_m(x)}{F'_0(x)}=1\right)F'_0(x) \ dx\\
        &=\int I\left(\frac{F'_m(x)}{F'_0(x)}=1\right) F'_m(x)\ dx\\
        &= P_{F_m}\left(\frac{F'_m(x)}{F'_0(x)}=1\right).
    \end{align*} 
    If $p>0$, then It suffices to verify that between the points corresponding to the tests $\phi_{>1}$ and $\phi_{\geq 1}$, the tradeoff function has slope -1. This is sufficient since the point of symmetry of a symmetric tradeoff function has $-1$ as a sub-derivative, and by concavity the derivative is increasing. Then 
    \[\EE_{F_0} \phi_{\geq 1} = \EE_{F_0} \phi_{>1} + P_{F_0}\left(\frac{F'_m(x)}{F'_0(x)}=1\right) =\EE_{F_0}\phi_{>1}+p,\]
    \[1-\EE_{F_m}\phi_{\geq 1} = 1-\EE_{F_m} \phi_{>1}-P_{F_m}\left(\frac{F'_m(x)}{F'_0(x)}=1\right)=1-\EE_{F_m}\phi_{>1}-p.\]
    We see that the slope is $-p/p=-1$.
    
    If $p=0$, then $\EE_F \phi_{>1}=\EE_F \phi_{\geq k}$ and $\EE_{F_m}\phi_{>1}=\EE_{F_m}\phi_{\geq 1}$. We will show that the type I and type II errors are equal for the test $\phi_{\geq 1}$, and hence $\EE_{F} \phi_{\geq 1}$ is the fixed point of $T(F_0,F_m)$. For an indeterminate $x$, call $y=m-x$. Then,
    \begin{align*}
\EE_{F_0}\phi_{\geq 1}&=\int I\left( \frac{F'(x-m)}{F'(x)} \geq 1\right) F'(x)\ dx\\
    &=\int I\left(\frac{F'(-y)}{F'(m-y)}\geq 1\right) F'(m-y)\ dx\\
    &=\int I\left(\frac{F'(y)}{F'(y-m)}\geq 1\right) F'(y-m)\ dx\\
    &=1-\int I\left(\frac{F'(y-m)}{F'(y)}>1\right)\\
    &=1-\EE_{F_m} \phi_{>1}\\
    &=1-\EE_{F_m}\phi_{\geq 1}.
    \end{align*}
    We see that $\EE_{F_0}\phi_{\geq 1}$ is the fixed point of $T(F_0,F_m)$. 
    }
    \end{enumerate}
  \end{proof}

\thmcanonical*
\begin{proof}%[Proof of Theorem \ref{thm:canonical}.]
  For simplicity of notation, we denote $F\defeq F_f$. We verify the four points of Definition \ref{def:CND}. It is easiest to prove the points in reverse order.
    \begin{enumerate}
    \item[4.] Symmetry was already shown in Proposition \ref{prop:tradeoff}. 
        \item[3.] First, we will show that $\frac{\frac{d}{dx} F(x-1)}{\frac{d}{dx} F(x)}$ is increasing in $x$. Let $x\in \RR$ be a point of differentiability of $F(x)$ and $F(x-1)$ and  be such that $0<F(x)<1$. Then $F'(x)>0$ by property 4 of Proposition \ref{prop:tradeoff}. So,
        \begin{align*}
            \frac{\frac{d}{dx} F(x-1)}{\frac{d}{dx} F(x)} &= \frac{\frac{d}{dx} f(1-F(x))}{\frac{d}{dx}F(x)}\\
            &=f'(1-F(x)) \frac{-F'(x)}{F'(x)}\\
            &=-f'(1-F(x)),
        \end{align*}
        where in the end, we note that $1-F(x)$ is decreasing, and $f'(\alpha)$ is increasing. We see that this quantity is increasing. {By property 4 of Proposition \ref{prop:cnd}, $F'(x)=0$ only if $F(x)=0$ or $F(x)=1$. When either $X\sim F$ or $X\sim F(\cdot-1)$, the probability that either $F(X)=F(X-1)=0$ or $F(X)=F(X-1)=1$ is zero, so we can disregard the case that  $F'(x-1)/F'(x)$ has the form $0/0$. If $x$ satisfies $F(x-1)=0$ and $F(x)\in (0,1)$, then the ratio $F'(x-1)/F'(x)=0$, whereas if $F(x)=1$ and $F(x-1)\in (0,1)$, then the ratio $F'(x-1)/F'(x)=+\infty$. These special cases preserve the increasing nature of the ratio.}

        %Note that when $X\sim F$ or $X\sim F(\cdot-1)$,  $F(X)=F(X-1)\in \{0,1\}$ with probability zero, so this case can be disregarded. The other possible case is if $F(x-1)\in (0,1)$, but $F(x)=1$, where the ratio  $F'(x-1)/F'(x)=+\infty$. This still preserves the increasing nature of the ratio.  %Since, $F(x)$ is differentiable almost everywhere and $f$ is differentiable almost everywhere, we have that the monotone likelihood ratio property holds almost everywhere. 
        
        Thus, we know that the optimal rejection set is of the form $(a,\infty)$. The type I error of this test is $\alpha=1-F(a)$, whereas the type II is $F(a-1)$. Then we have that the tradeoff function is $T(F(\cdot),F(\cdot-1))(\alpha) = F(F^{-1}(1-\alpha)-1)$. 
        
        \item [2.] Let $\alpha \in (0,1)$. Then 
        {
        \begin{align*}
            T(F(\cdot),F(\cdot-1))(\alpha) &= F(F^{-1}(1-\alpha)-1)\\
            &=f(1-F(F^{-1}(1-\alpha)))\\
            &=f(1-(1-\alpha))\\
            &=f(\alpha),
        \end{align*} 
        }
        where we used the identity $F(x-1) = f(1-F(x))$, provided that {$F(x)>0$} (property 2 of Proposition \ref{prop:tradeoff}), {as well as property 2 of Lemma \ref{lem:invertible}.}
        
        \item[1.] {We need to show that 
        %Note that statement 1 in Definition \ref{def:CND} is equivalent to requiring that
        $T(F(\cdot),F(\cdot-m))\geq T(F(\cdot),F(\cdot-1))$ for all $m\in [0,1]$}. We will denote $F_0=F$, $F_m=F(\cdot-m)$ and $F_1=F(\cdot-1)$. By Lemma \ref{lem:symTrade}, when testing $F_0$ versus $F_m$, it suffices to check rejection regions of the form $S_k = \{x\mid \frac{F'_m(x)}{F'_0(x)}>k\}$ for $k\geq 1$ or $S'_k = \{x\mid \frac{F'_m(x)}{F'_0(x)}\geq k\}$ for $k>1$. First we will check $S_k$ for $k\geq 1$. Since $F'$ is decreasing on $(-1/2,\infty)$ and increasing on $(-\infty,1/2)$, we have that $S_k\subset S_1=\{x\mid F'(x-m)>F'(x)\}\subset (1/2,\infty)$. Furthermore, on $(1/2,\infty)$, we have that $\frac{F'_1(x)}{F_m'(x)}\geq 1$, since $F'$ is decreasing.
        
        Then $\EE_{X\sim F_m} I(X\in S_k)=\int_{S_k} F'_m(x) \ dx\leq \int_{S_k} \frac{F'_1(x)}{F'_m(x)} F'_m(x) \ dx=\int_{S_k} F'_1(x){\ dx}=\EE_{X\sim F_1}I(X\in S_k)$, so that the power under $F_1$ is greater than the power under $F_m$. We can repeat the argument for rejection sets of the form $S_k'$ for $k>1$. By part 2 of Lemma \ref{lem:symTrade}, we have that $T(F(\cdot),F(\cdot-m))\geq T(F(\cdot),F(\cdot-1))$, {establishing} part 1 of Definition \ref{def:CND}.
    \end{enumerate}
  \end{proof}
  
  {In Definition \ref{def:CNDsynthetic} we defined the cdf of the constructed CND. While this expression is very useful for deriving properties of this distribution, the quantile function is important for sampling. In Proposition \ref{prop:samplingCND}, we give a recursive expression for the quantile function of the CND constructed in Definition \ref{def:CNDsynthetic} and show that it can be evaluated in a finite number of steps.}
    
  \begin{restatable}{prop}{propsamplingCND}\label{prop:samplingCND}
  Let $f$ be a symmetric nontrivial tradeoff function and let $F_f$ be as in Definition \ref{def:CNDsynthetic}. Then the quantile function $F_f^{-1}:(0,1)\rightarrow \RR$ for $F_f$ can be expressed as
  \[F_f^{-1}(u) = \begin{cases}
  F_f^{-1}(1-f(u))-1&u<c\\
  \frac{u-1/2}{1-2c}&c\leq u\leq 1-c\\
  F_f^{-1}(f(1-u))+1&u>1-c,
  \end{cases}\]
  where $c$ is the unique fixed point of $f$. {Furthermore, for any $u\in (0,1)$, the expression $Q_f(u)$ takes a finite number of recursive steps to evaluate. Thus,} if $U\sim U(0,1)$, then $F_f^{-1}(U) \sim F_f$. 
  \end{restatable}
  %\begin{proof}[Proof sketch.]
  %The result follows by inverting the formula for $F_f$ given in Definition \ref{def:CNDsynthetic}. Details are contained in Section \ref{s:proofs}.
  %\end{proof}
  %\propsamplingCND*
  \begin{proof}
 {For ease of notation, we will drop the subscripts of $F_f$ and $F_f^{-1}$.} By Proposition \ref{prop:tradeoff}, we have established that $F$ is strictly increasing on $\{x\mid F(x)\in (0,1)\}$.  {So, $F$ is invertible on $\{x\mid F(x)\in (0,1)\}$. }  %So, we know that the quantile function $Q_f=F^{-1}_f$ on $(0,1)$. 
   In the case that $u\in [c,1-c]$, it is easy to verify the expression of $F^{-1}$ by the construction of $F$ in Definition \ref{def:CNDsynthetic}.
    
  Suppose that $u\in (1-c,1)$. Then $F^{-1}(u)>1/2$, so by Definition \ref{def:CNDsynthetic}, we know that $F^{-1}(u)$ satisfies $u=1-f(F(F^{-1}(u)-1))$. {Since $u\in (0,1)$, this implies that $F(F^{-1}(u)-1)=f(1-u)< f(0)$, by parts 2 and 3 of Theorem \ref{thm:canonical}, so the previous equation implies that $f(1-u)=F(F^{-1}(u)-1)$. Next we will apply $F^{-1}$ to both sides, but we need to verify that $F$ is invertible at $F^{-1}(u)-1$. Call $M=\inf\{x\mid F(x)>0\}$; by the construction of Definition \ref{def:CNDsynthetic}, note that $M\leq -1/2$. It suffices to show that $F^{-1}(u)-1\geq M$, since $F$ is invertible whenever $F\in (0,1)$. Note that 
  \[u>1-c=F(1/2)\geq F(1+M),\]
  where the inequality uses the fact that  $M\leq -1/2$. So, we have that $u\geq F(1+M)$, which implies that $F^{-1}(u)-1\geq M$, by Lemma \ref{lem:galois}. Now that we know $F$ is invertible at $F^{-1}(u)-1$, we have $F^{-1}(f(1-u))=F^{-1}(u)-1$, which is equivalent to $F^{-1}(f(1-u))+1=F^{-1}(u)$, as claimed in the prescription of $F^{-1}$.   
  
  %Solving this equation for $Q_f(u)$ gives $Q_f(u) = Q_f(f(1-u))$, where we use the fact that $f\circ f(\alpha)=\alpha$ for all $\alpha\in (0,1)$. 
  
  For the case where $u<c$, note that since $F$ corresponds to a symmetric random variable, we have that $F^{-1}(u) = -F^{-1}(1-u)$. Applying the recursive formula for $u>1-c$ gives the result 
  \[F^{-1}(u) = -F^{-1}(1-u)
  =-[F^{-1}(f(u))+1]
  =-F^{-1}(f(u))-1
  =F^{-1}(1-f(u))-1,\] for $u<c$. 
  
  To see that the recursion only requires a finite number of iterations, note that if $u\in (1-c,1)$, then by Lemma \ref{lem:contraction} there exists $k\in \ZZ^{+}$ such that $u\in \left( (1-f)^{\circ k}(c),(1-f)^{\circ(k+1)}(c)\right]$ and necessarily $(1-f)^{\circ k}(c)<1$, where the notation means $(1-f)^{\circ k}(c) \defeq (1-f)\circ (1-f)\circ \cdots \circ (1-f)(c)$, where $(1-f)$ is composed $k$ times. Then,
  \begin{align}
      f(1-u)&\in \left( f(1-(1-f)^{\circ k}(c)),f(1-(1-f)^{\circ(k+1)}(c))\right]\\
      &=\left(f\circ f\circ (1-f)^{\circ (k-1)}(c),f\circ f\circ (1-f)^{\circ k}(c)\right]\\
      &=\left( (1-f)^{\circ (k-1)}(c),f\circ f\circ (1-f)^{\circ k}(c)\right]\label{eq:sampling1}\\
      &\subset \left( (1-f)^{\circ (k-1)}(c), (1-f)^{\circ k}(c)\right],\label{eq:sampling2}
  \end{align}
  where for \eqref{eq:sampling1}, we use the equation $f\circ f \circ (1-f)^{\circ (k-1)}(c)=(1-f)^{\circ (k-1)}$, which is justified as follows: notice that $(1-f)^{\circ (k-1)}(c)\leq f(0)$, as if $(1-f)^{\circ (k-1)}(c)> f(0)$ then $f\circ (1-f)^{\circ (k-1)}(c)=0$, which implies that $(1-f)^{\circ k}(c)=1$, contradicting our assumption about $(1-f)^{\circ k}(c)$. For \eqref{eq:sampling2}, we use the fact that $f\circ f(\alpha)\leq \alpha$. 
  %where we use the fact that $(1-f)^{\circ k}(c)<1$ and $f\circ f(\alpha) =\alpha$ for $\alpha\in (0,1)$. 
  We see that after $k$ iterations of the recursive formula, the evaluation reduces to $F^{-1}(u^*)$ for some $u^*\in [c,1-c]$. By symmetry, we have that when $u\in (0,c)$ the recursion finishes in a finite number of steps as well.

  By inverse transform sampling, when $U\sim U(0,1)$ we have that $F^{-1}(U)\sim F$. 
  }
  \end{proof}
  
 % \begin{remark}
 % While inverse transform sampling based on the quantile function described in Proposition \ref{prop:samplingCND} results in exact samples from the constructed CND, when using the recursive formula, the running time depends on the value of the uniform variable $U$, and so is related to the final value $Q_f(U)$. Due to this, the implementation may be vulnerable to a timing attack such as described in \citet{haeberlen2011differential}. We leave it to future work to develop alternative sampling methods that avoid this side channel.  
 % \end{remark}

  \cortulap*
  \begin{proof}%[Proof of Corollary \ref{cor:tulap}.]
    Recall that the cdf of $\mathrm{Tulap}(0,b,0)$, defined in \citet{awan2018differentially}, is 
    \[F_{N_0}(x) = \begin{cases}
    \frac{b^{-[x]}}{1+b}(b+\{x-[x]+1/2\}(1-b))& x\leq 0\\
    1- \frac{b^{[x]}}{1+b}(b+\{[x]-x+1/2\}(1-b))&x>0,
    \end{cases}\]
    where $[x]$ is the nearest integer function. The cdf of $\mathrm{Tulap}(0,b,q)$ is
    \[F_N(x) = \begin{cases}
    0&F_{N_0}(x)<q/2\\
    \frac{F_{N_0}(x)-q/2}{1-q}& q/2\leq F_{N_0}(x)\leq 1-q/2\\
    1&F_{N_0}(x)>1-q/2.
    \end{cases}\]
    By inspection, the fixed point of $f_{\ep,\de}$ is $c=\frac{1-\de}{1+e^\ep}$. It is easy to verify that $F_N(x) = c(1/2-x) + (1-c)(x+1/2)$ for $x\in (-1/2,1/2)$. By \citet[Lemma 2.8]{awan2020differentially}, we have that $F_N$ satisfies the recurrence relation in Definition \ref{def:CNDsynthetic}. We conclude that $F_N = F_f$.     
  \end{proof}
  
  \lemDPtest*
   \begin{proof}%[Proof of Lemma \ref{lem:DPtest}.]
The proof is based on {applying the Neyman Pearson Lemma to the testing of two Bernoulli random variables}. Let $x,x'\in \mscr X^n$ such that $H(x,x')\leq 1$ be given. We need to show that $T(\mathrm{Bern}(\phi(x)),\mathrm{Bern}(\phi(x')))\geq f$. Call $p\defeq \phi(x)$ and $q\defeq \phi(x')$.

{Note that if $p=q$, The result is trivial since the tradeoff function $T(\mathrm{Bern}(p),\mathrm{Bern}(q))=\mathrm{Id}\geq f$, and both $q=p\leq 1-f(p)=1-f(q)$, since $1-f(p)\geq p$. Next we will assume that $p<q$.}

By the Neyman Pearson Lemma, recall that the most powerful test $\psi$ to distinguish $H_0: \mathrm{Bern}(p)$ versus $H_1:\mathrm{Bern}(q)$ is of one of the two following forms:
\[\psi_1(x) = \begin{cases}
1&x=1\\
c&x=0
\end{cases},\qquad 
\psi_2(x) = \begin{cases}
c&x=1\\
0&x=0,
\end{cases}\]
where the case and the value $c$ are chosen such that the size is $\alpha$.

In the  case where $p>\alpha$, setting $c=\alpha/p$ allows for $\EE_{\mathrm{Bern}(p)} \psi_2(x) = cp=\alpha$. The type II error in this case is then $1-\EE_{\mathrm{Bern}(q)} \psi_2(x) = 1-cq=1-\frac{\alpha}{p}q$.

On the other hand if $p\leq \alpha$, then for $c = \frac{\alpha-p}{1-p} $ we have that $\EE_{\mathrm{Bern}(p)} \psi_1(x) = 1p+c(1-p)=\alpha$. Then the type II error is $1-\EE_{\mathrm{Bern}(q)} \psi_2(x) = 1-\left(1q+\frac{\alpha-p}{1-p} (1-q)\right)=\frac{1-\alpha}{1-p}(1-q)$.

Combining these two cases, we see that the tradeoff function is
\[T(\mathrm{Bern}(p),\mathrm{Bern}(q))=\begin{cases}
\frac{1-\alpha}{1-p}(1-q)&p\leq \alpha\\
1-\frac{\alpha}{p}(q)&p>\alpha,
\end{cases}\]
which is a piece-wise linear function with break points $(0,1)$, $(p,1-q)$ and $(1,0)$. Because $f$ is a convex function which satisfies $f(0)\leq 1$ and $f(1)=0$ (implied by $f(x)\leq 1-x$), we have that $T(\mathrm{Bern}(p),\mathrm{Bern}(q))\geq f$ if and only if $1-q\geq f(p)$ or equivalently $q\leq 1-f(p)$. 

{Now suppose that $p>q$. Note that by symmetry, establishing $T(\mathrm{Bern}(p),\mathrm{Bern}(q)\geq f$ is equivalent to establishing $T(\mathrm{Bern}(q),\mathrm{Bern}(p))\geq f^{-1}=f$. By our earlier work, swapping the roles of $p$ and $q$, we have that this inequality holds if and only if $p\leq 1-f(q)$. }
%Next, we have to address the case that $p>q$. Note that in this case, we have that $(1-p)<(1-q)$, yet $T(\mathrm{Bern}(p),\mathrm{Bern}(q))=T(\mathrm{Bern}(1-p),\mathrm{Bern}(1-q))$, since $1-x$ is a bijection. By our above derivation, we have that $T(\mathrm{Bern}(p),\mathrm{Bern}(q))=T(\mathrm{Bern}(1-p),\mathrm{Bern}(1-q))$ is a piece-wise linear function with break points $(0,1)$, $(1-p,q)$, and $(1,0)$. So, we have that $T(\mathrm{Bern}(p),\mathrm{Bern}(q))=T(\mathrm{Bern}(1-p),\mathrm{Bern}(1-q))\geq f$ if and only if 
%$1-p\geq f(q)$ or equivalently $p\leq 1-f(q)$. 
\end{proof}

\corCNDHT*
\begin{proof}%[Proof of Corollary \ref{cor:CND_HT}.]
%For $L\leq U$, we define the clamp function $[\cdot]_L^U:\RR\rightarrow [L,U]$ by $[x]_L^U\defeq \max\{\min\{x,U\},L\}$. Since $N\sim F$ is a continuous random variable with some range $[-U,U]$, we have that on $\RR$, $F^{-1} \circ F(x) =[x]_{-U}^U$, and  on $[0,1]$, $F\circ F^{-1}=\textrm{Id}$. 

Let $x,x'\in \mscr X^n$ such that $H(x,x')\leq 1$. For the reverse direction of the statement, 
{
suppose that $F^{-1}(\phi(x))\leq F^{-1}(\phi(x'))+1$. Applying $F$ preserves this inequality since $F$ is increasing, and $F\circ F^{-1}(\phi(x)) = \phi(x)$ by Lemma \ref{lem:inversion}. So, 

\begin{align*}
    \phi(x)&\leq F(F^{-1}(\phi(x'))+1)\\
    &=1-F(-F^{-1}(\phi(x'))-1)\\
    &=1-F(F^{-1}(1-\phi(x'))-1)\\
    &=1-f(\phi(x')),
\end{align*}
where we used the symmetry of $F$, and the fact that $F$ is a CND for $f$. By Lemma \ref{lem:DPtest} we conclude that $\phi$ satisfies $f$-DP. 

For the forward direction, suppose that $\phi(x)\leq 1-f(\phi(x'))$, or equivalently, $F(F^{-1}(1-\phi(x'))-1)\leq 1-\phi(x)$. Then each of the following inequalities follows from the one above:
\begin{align}
    F(F^{-1}(1-\phi(x'))&\leq 1-\phi(x)\\
    1-F(F^{-1}(1-\phi(x'))-1)&\geq \phi(x)\\
    F(1-F^{-1}(1-\phi(x')))&\geq \phi(x)\\
    1-F^{-1}(1-\phi(x'))&\geq F^{-1}(\phi(x))\label{eq:corCNDht1}\\
    F^{-1}(1-\phi(x'))-1&\leq -F^{-1}(\phi(x))\\
    F^{-1}(1-\phi(x'))-1&\leq F^{-1}(1-\phi(x))\\
    -F^{-1}(1-\phi(x))&\leq -F^{-1}(1-\phi(x'))+1\\
    F^{-1}(\phi(x))&\leq F^{-1}(\phi(x'))+1,
    \end{align}
}
where \eqref{eq:corCNDht1} used the Galois inequalities of Lemma \ref{lem:galois}, and the other steps used the symmetry of $F$ and basic algebraic manipulations. 
\end{proof}

  \begin{lemma}[Theorem 8.3.27 of \citealp{casella2002statistical}]\label{lem:pval}
Let $H_0: \theta=\theta_0$ versus $H_1: \theta\in \Theta_1$ be a hypothesis test with a simple null hypothesis. Let $T\sim \theta$ be a real-valued (continuous) test statistic. Assuming that large values of $T$ give evidence for $H_1$, a $p$-value for the hypothesis is 
\[p(T) = 1-F_{T\sim \theta_0}(T) = P_{T_0\sim\theta_0}(T_0>T).\]
\end{lemma}
  
  \begin{lem}\label{lem:threshold}
    Let $I$ be an arbitrary index set. Let $X$ be a random variable, and consider the following simple hypothesis test $H_0: X\sim P$ versus $H_1: X\sim Q_i$ for some $i\in I$. Let $T(X)$ be a {continuous }real-valued statistic, and consider the threshold test which rejects for large values of $T$. Let $t\in \RR$ and call $\alpha = P_{X\sim P}(T\geq t)$ and $\beta_i = P_{X\sim Q_i}(T\geq t)$. Then the $p$-value $p(T) = P_{X\sim P}(T(X)\geq T)$ satisfies $P_{P}(p(T)\leq \alpha) = \alpha$ and $P_{Q_i}(p(T)\leq \alpha)=\beta_i$ for all $i\in I$.
  \end{lem}
  \begin{proof}
  Call $F_{P}$ and $F_{Q_i}$ the cdf of the random variable $T(X)$, when $X\sim P$ and $X\sim Q_i$, respectively. Then $p(T) = 1-F_P(T)$. We have that  $P_P(p(T)\leq\alpha) = P_P(1-F_P(T)\leq \alpha)=P_P(F^{-1}_P(1-\alpha)\leq T) = 1-F_P(F_P^{-1}(1-\alpha))=\alpha$, where we use the fact that $F_P^{-1}\circ F_P(T)=T$ with probability one, and that $F_P\circ F_P^{-1}(\alpha)=\alpha$ by part 2 of Lemma \ref{lem:inversion}. Next consider 
  \begin{align}
      P_{Q_i}(p(T)\leq \alpha)&=P_{Q_i}(1-F_P(T)\leq \alpha)\\
      &=P_{Q_i}(F^{-1}_P(1-\alpha)\leq T)\label{eq:threshold1}\\
      &=1-F_{Q_i}(F^{-1}_P(1-\alpha))\\
      &=1-F_{Q_i}(t)\\
      &=P_{Q_i}(T\geq t)\\
      &=\beta_i,
  \end{align}
  {where \eqref{eq:threshold1} used the Galois inequalities (Lemma \ref{lem:galois}),} and we used the fact that $F^{-1}_P(1-\alpha)=t$ or equivalently, that $\alpha = 1-F_P(t)=P_P(T\geq t)$. 
  \end{proof}
  
  \thmpVal*
  \begin{proof}%[Proof of Theorem \ref{thm:pVal}.]
  \begin{enumerate}
      \item  In Corollary \ref{cor:CND_HT}, we saw that for any $x$ and $x'$ adjacent, $F^{-1}(\phi(x))\leq F^{-1}(\phi(x'))+1$. This implies that $F^{-1}(\phi(x))$ has sensitivity 1. The result follows from Definition \ref{def:CND}.
      \item It suffices to show that $P(T\geq 0)=\phi(x)$: 
      \begin{align*}
          P(T\geq 0) &= P(F^{-1}(\phi(x))+N\geq 0)\\
          &=P(N\leq F^{-1}(\phi(x)))\\
          &=F(F^{-1}(\phi(x)))\\
          &=\phi(x),
      \end{align*}
      %$P(T\geq 0) = P(F^{-1}(\phi(x))+N\geq 0)=P(N\leq F^{-1}(\phi(x)))=F(F^{-1}(\phi(x)))=\phi(x)$, 
      where $F\circ F^{-1}=\mathrm{Id}$ since $F$ is continuous, and we used the symmetry of $N$.
      \item {We can express } 
      {\begin{align*}
          p&=\sup_{\theta_0\in H_0} \EE_{X\sim \theta_0}F(F^{-1}(\phi(X))-T)\\
          &=\sup_{\theta_0\in H_0} P_{X\sim \theta_0, N}(N\leq F^{-1}(\phi(X))-T)\\
          &= \sup_{\theta_0\in H_0} P_{X\sim \theta_0,N}(T\leq F^{-1}(\phi(X))+N),
      \end{align*}}%$p=\sup_{\theta_0\in H_0} \EE_{X\sim \theta_0} F(F^{-1}(\phi(X))-T)=\sup_{\theta_0\in H_0} P_{X\sim \theta_0, N}(N\leq F^{-1}(\phi(X))-T) = \sup_{\theta_0\in H_0} P_{X\sim \theta_0,N}(T\leq F^{-1}(\phi(X))+N)$, 
      {where we used the fact that $N\overset d =-N$. By \citet[Theorem 8.3.27]{casella2002statistical}, this is a valid $p$-value. }
      \item The result follows from Lemma \ref{lem:threshold}.\qedhere
  \end{enumerate}
  \end{proof}
  
  The following lemma is one of several techniques to prove the Neyman Pearson Lemma, and appears as Lemma 4.4 in \citet{awan2018differentially}.

\begin{lemma}\label{lem:npl}%\todo{find lemma to cite}
Let $(\mscr X , \mscr F , \mu)$ be a measure space and let $f$ and $g$ be two densities on X with respect to $\mu$. Suppose that $\phi_1,\phi_2 :\mscr X\rightarrow [0,1]$ are such that $\int \phi_1f\ d\mu\geq \int \phi_2f\ d\mu$, and there exists $k\geq 0$ such that $\phi_1 \geq  \phi_2$ when $g \geq kf$ and $\phi_1 \leq \phi_2$ when $g < kf$. Then $\int \phi_1g\ d\mu \geq \int \phi_2 g\ d\mu$.
\end{lemma}
\begin{proof}
Note that $(\phi_1 - \phi_2)(g - kf) \geq 0$ for all $x\in \mscr X$. This implies that $\int (\phi_1 - \phi_2)(g - kf) \ d\mu \geq 0$. Hence, $\int \phi_1g \ d\mu - \int \phi_2g \ d\mu \geq k\left(\int \phi_1f \ d\mu - \int  \phi_2 f \ d\mu \right)\geq 0$.
\end{proof}

\thmbinary*
\begin{proof}%[Proof of Theorem \ref{thm:binary}.]
First we will establish the equivalence of 1 and 2. Given a test $\phi$ of the form 2, we know by Lemma \ref{lem:recurrence} that $\phi$ satisfies the recurrence in 1. Set $y$ to be the smallest $x\in \{0,1,2,\ldots,n\}$ such that $\phi(x)>0$. Then set $c = \phi(y)$. We have that $\phi$ fits the form of 1. 

Now let $\phi$ be of the form 1. Solve $c=F(y-m)$ for $m$, which has a solution by the Intermediate Value Theorem as $\lim_{m\rightarrow -\infty} F(y-m)=0$ and $\lim_{m\rightarrow \infty}F(y-m)=1$. By Lemma \ref{lem:recurrence}, $F(x-m)=\phi(x)$ for $x\in \{0,1,2,\ldots, n\}$. We conclude that 1 and 2 are equivalent. 

Next we  argue that the prescribed $\phi$ satisfies $f$-DP, using form 1. Note that we have $\phi(x) \leq 1-f(\phi(x-1))$ for all $x=1,\ldots, n$. We also need to show that $\phi(x-1)\leq 1-f(\phi(x))$. To this end, we first observe that $\phi(x-1)\leq \phi(x)$. This follows from the fact that $f(t)\leq 1-t$ or equivalently that $t\leq 1-f(t)$. So, we have $\phi(x-1)\leq \phi(x)\leq 1-f(\phi(x))$. 

Next given $\alpha\in (0,1)$, we need to argue that there exists a test $\phi$ of the prescribed form which has $\EE_{P}\phi(x)=\alpha$. We use form 2 for this part, so we need to show that there exists $m\in \RR$ such that $\alpha=\EE_{X\sim P}F(X-m)$. Note that $\EE_{X\sim P}F(X-m)$ is a continuous function in $m$, where the limit as $m\rightarrow-\infty$ is zero and {the limit as } $m\rightarrow \infty$ is 1. By the Intermediate Value Theorem there exists $m$ such that $\alpha=\EE_{X\sim P}F(X-m)$.

Let $\phi_a(x)$ be a test of form 1 which has $\EE_P\phi_a(x)=\alpha$, and let $\psi$ be another level $\alpha$ $f$-DP test. We will show that $\phi_a$ is more powerful than $\psi$. First we claim that there exists a value $z$ such that $\psi(z)\leq \phi_a(z)$. If this were not the case, then $\psi(x)>\phi_a(x)$ for all $x$, which implies that $\EE_P \psi(x)>\EE_P \phi_a(x)=\alpha$, contradicting the level of $\psi$. 

Now, let $z_m$ be the smallest value such that $\psi(z_m)\leq \phi_a(z_m)$. Then by assumption, for all $x< z_m$, $\psi(x)> \phi_a(x)$. Next, note that $\psi(z_m+1)\leq 1-f(\psi(z_m)\leq 1-f(\phi_a(z_m))\leq \phi_a(z_m+1)$. By induction, we have that  for all $x\geq z_m$, $\psi(x)\leq \phi_a(x)$.

We conclude that $\phi_a(x)\geq \psi(x)$ for all $x\geq z_m$ and $\phi_a(x)\leq \psi(x)$ for all $x<z_m$. In other words, there exists a threshold $t^*$ such that $\phi_a(x)\geq \psi(x)$ when $\frac{q(x)}{p(x)}\geq t^*$ and $\phi_a(x)\leq \psi(x)$ when $\frac{q(x)}{p(x)}\leq t^*$. By Lemma \ref{lem:npl}, we have that $\EE_Q \phi_a(x)\geq \EE_Q\psi(x)$.

Last, we verify the claim of form 3. The variable $T=X+N$ satisfies $f$-DP by Definition \ref{def:CND}, since $X$ has sensitivity 1. {Call $F_{T\sim P}(t)=P_{X\sim P, N\sim F}(X+N\leq t)$ the cdf of $T=X+N$ when $X\sim P$. }The variable $p$ can be expressed as 
{
\begin{align*}
    p &= \EE_{X\sim P}F(X-T)\\
    &=P_{X\sim P,N\sim F}(N\leq X-T)\\
    &=P_{X\sim P,N\sim F}(T\leq X+N)\\
    &=1-P_{X\sim P,N\sim F}(X+N\leq T)\\
    &=1-F_{T\sim P}(T),
\end{align*}
where we used the fact that $N\overset d=-N$. We see that this is a $p$-value by Lemma \ref{lem:pval}.

}

It is easy to verify that $P(I(T\geq m)=1\mid X=x)=P_N(x+N\geq m)=P_N(N\leq x-m)=F(x-m)=\phi(x)$. We then check 
{\begin{align*}
    P(p\leq \alpha\mid X=x) &= P_N(1-F_{T\sim P}(x+N)\leq \alpha)\\
    &=P_N(1-\alpha\leq F_{T\sim P}(x+N))\\
    &=P_{N}(F^{-1}_{T\sim P}(1-\alpha)\leq x+N)\\
    &=P_N(m\leq x+N)\\
    %&=P_N(m\leq x+N)\\
    &=F_N(x-m)\\
    &=\phi(x)
\end{align*}
}
where we used the {Galois inequalities (Lemma \ref{lem:galois}), and} %fact that $F^{-1}_T\circ F_T(T)=T$ with probability one, and 
that $\alpha=\EE_{X\sim P}F(X-m)=P_{X,N}(N\leq X-m)=1-F_{T\sim P}(m)$, which implies that $F^{-1}_{T\sim P}(1-\alpha)=m$. 
\end{proof}

\subsection{Proof of Theorem \ref{thm:semi}}
This section is devoted to the proof of Theorem \ref{thm:semi}. First we need to establish notation and a few lemmas. 

Given $f$ be a symmetric nontrivial tradeoff function, define $\Psi_f$ to be the set of $f$-DP tests on $\{0,1,2,\ldots, m\}$: 
\[\Psi_f = \left\{\psi:\{0,1,2,\ldots,m\}\rightarrow [0,1]\mid \psi(y)\leq 1-f(\psi(y')), \text{ for all } |y-y'|=1\right\}.\]

Given a function $f:\RR\rightarrow \RR$ we define $f^{\circ k}=f\circ f\circ\cdots \circ f$, where there are $k$ appearances of $f$. For example, we write $(1-f)^{\circ 2}(x)=(1-f)\circ (1-f)(x)=1-f(1-f(x))$.

\begin{lem}\label{lem:convert}
  Given $\phi\in \Phi_f$ defined in Equation \eqref{eq:Phi}, and given $z\in \{0,1,\ldots, m+n\}$ define $\psi_z:\{0,1,\ldots,m\}\rightarrow [0,1]$ by 
  \[\psi_z(y) = \begin{cases}
    \phi(z-L,L)&\text{if }y\leq L\\
  \phi(z-y,y)&\text{if } L\leq y\leq U\\
  \phi(z-U,U)&\text{if }y\geq U,
  \end{cases}\]
  where $L=\max\{0,n-z\}$ and $U=\min\{m,z\}$. Then $\psi_z\in \Psi_g$ where $g=1-(1-f)^{\circ 2}$.
\end{lem}
\begin{proof}
Let $y,y'\in \{0,\ldots, m\}$ %$\{L,L+1,\ldots, U-1,U\}$ 
such that $|y-y'|=1$. Define $x=z-y$ and $x'=z-y'$. Then $(x,y)$ and $(x',y')$ both lie in $\{0,1,\ldots,n\}\times \{0,1,\ldots, m\}$. Then $(x',y')$ is adjacent to $(x,y')$ and $(x,y')$ is adjacent to $(x,y)$. 

If $y,y'\geq U$ or $y,y'\leq L$, then $\psi_z(y')=\psi_z(y)\leq (1-g)\circ \psi_z(y)$. If $y,y'\in \{L,L+1,\ldots, U-1,U\}$, then
\begin{align*}
    \psi_z(y')&=\phi(x',y')\\
    &\leq (1-f)\circ \phi(x,y')\\
    &\leq (1-f)^{\circ 2}\circ \phi(x,y)\\
    &=(1-g)\circ \phi(x,y)\\
    &=(1-g)\circ \psi_z(y)
\end{align*}
where $g=1-(1-f)^{\circ 2}$. 
\end{proof}

\begin{lem}[Lemma A.5, \citealp{dong2022gaussian}]\label{lem:A5}
Let $P$, $Q$, and $R$ be distributions, and let $f$ and $g$ be tradeoff functions. If $T(P,Q)\geq f$ and $T(Q,R)\geq g$ then $T(P,R)\geq g\circ(1-f)$.
\end{lem}

%By iterating Lemma \ref{lem:twice}, we obtain almost certainly be extended as follows: for a positive integer $k$, $F(k\cdot)$ is a CND for $g=1-(1-f)^{\circ k}$.
\begin{lem}\label{lem:twice}
Let $f$ be a symmetric nontrivial tradeoff function, and let $F_f$ be a CND for $f$. Then $F(2\cdot)$ is a CND for $g=1-(1-f)^{\circ 2}=f(1-f)$. 
\end{lem}
\begin{proof}
Recall from \citet[Section 2.5]{dong2022gaussian} that $g(\alpha)=1-(1-f)^{\circ 2}(\alpha)$ is also symmetric, and for all $\alpha$, {$g(\alpha)\leq f(\alpha)$, so $g$ is also nontrivial}. By Theorem \ref{thm:canonical} there exists a CND for $g$. 

We drop the subscript and write $F=F_f$. We write $F_2(\cdot) = F(2\cdot)$, and $F_2^{-1}(\cdot) = \frac 12 F^{-1}(\cdot)$, where $F_2^{-1}$ is the quantile function for $F_2$. $F_2$ clearly satisfies property 4 of Definition \ref{def:CND}, since $F$ is a CND. For $\alpha\in (0,1)$, note the following connection between $F$ and $F_2$:
\begin{align*}
    F(F^{-1}(1-\alpha)-2)&=F\left(2 \left[(1/2) F^{-1}(1-\alpha)-1\right]\right)\\
    &=F_2(F_2^{-1}(1-\alpha)-1).
\end{align*}

{
Let $M\defeq \inf\{t\mid 0<F(t)\}$. By symmetry of $F$, we know that $M\leq 0$. If $M\geq -1/2$, then we have that $g(\alpha)\leq f(\alpha)=F(F^{-1}(1-\alpha)-1)=0$ for all $\alpha\in (0,1)$; we also have $F_2(F_2^{-1}(1-\alpha)-1)=F(F^{-1}(1-\alpha)-2)=0=g(\alpha)$ justifying property 3 of Definition \ref{def:CND}. Furthermore, $T(F_2(\cdot),F_2(\cdot-1))=T(F(\cdot),F(\cdot-2))=0=g$, since $F(\cdot)$ and $F(\cdot-2)$ have disjoint support. Finally, note that $T(F_2(\cdot),F_2(\cdot-m))\geq T(F_2(\cdot),F_2(\cdot-1))=0$, since $0$ is a trivial lower bound for any tradeoff function. We conclude that when $M\geq -1/2$, $F_2$ is a CND for $g$.

Now suppose that $M<-1/2$ and let $\alpha\in (0,1)$. If $\alpha\geq f(0)$, then $g(\alpha)=f(1-f(\alpha))=f(1-0)=0$ and $F(F^{-1}(1-\alpha)-2)\leq F(F^{-1}(1-\alpha)-1)=f(\alpha)=0$ because $F$ is increasing. We see that property 3 of Definition \ref{def:CND} holds in this case. Now assume that $\alpha<f(0)$, or equivalently $1-\alpha>1-f(0)$. For the following calculations, we will need to justify that $F$ is invertible at $F^{-1}(1-\alpha)-1$. To see this, note that $F$ is invertible at $F^{-1}(1-\alpha)$, and by Lemma \ref{lem:support} $F$ is also invertible at $F^{-1}(1-\alpha)-1$ unless $F^{-1}(1-\alpha)-1<M$. So, we need to show that $F^{-1}(1-\alpha)\geq M+1$:

\begin{align}
    M+1&=\inf \{t+1\mid 0<F(t)\}\\
    &=\inf \{t\mid 0<F(t-1)\}\\
    &=\inf \{t\mid f(0)>f(F(t-1))\ \& \ 0<F(t-1)\}\label{eq:twice1}\\
    &=\inf\{t\mid 1-f(0)<1-f(F(t-1))\ \&\ 0<F(t-1)\}\\
    &=\inf\{t\mid 1-f(0)<F(t)\ \& \ 0<F(t-1)\}\label{eq:twice2},
\end{align}
where \eqref{eq:twice1} uses the fact that $f$ is strictly decreasing at $0$; \eqref{eq:twice2} uses the fact that $0<F(t-1)$ to apply the recursion of Lemma \ref{lem:recurrence}. Now, suppose that $F(t-1)=0$: then $t-1\leq M$ and because $M<-1/2$, $F(t)<1$. So, $0=F(t-1)$ implies that $0=f(1-F(t))$. But this in turn implies that $1-F(t)\geq f(0)$ or equivalently $1-f(0)\geq F(t)$. We see that $1-f(0)<F(t)$ implies that $0<F(t-1)$. So, 

\begin{align}
   M+1 &=\inf\{t\mid 1-f(0)<F(t)\}\\
    &\leq\inf\{t\mid 1-\alpha\leq F(t)\}\label{eq:twice4}\\
    &=F^{-1}(1-\alpha),
    \end{align}
    where \eqref{eq:twice4} uses the fact that $1-\alpha>1-f(0)$. We are now ready to verify that $F_2$ satisfies property 3 of Definition \ref{def:CND} for $g$ when $M<-1/2$ and $\alpha\in (0,f(0))$:

\begin{align}
    g(\alpha)&=f(1-f(\alpha))\\
    &=f(1-F[F^{-1}\{1-\alpha\}-1])\\
    &=F\{F^{-1}[1-(1-F[F^{-1}\{1-\alpha\}-1])]{ -1}\}\\
    &=F(F^{-1}(1-\alpha)-2)\label{eq:twice5}\\
    &=F\left(2 \left[(1/2) F^{-1}(1-\alpha)-1\right]\right)\\
    &=F_2(F_2^{-1}(1-\alpha)-1),
\end{align}
where \eqref{eq:twice5} used the fact that $F$ is invertible at $F^{-1}(1-\alpha)-1$. 

For property 2 of Definition \ref{def:CND}, we need to show that $\frac{\frac{d}{dx} F(2(x-1))}{F(2x)}$ is increasing in $x$. %This is easily done using similar calculations as in the proof of Theorem \ref{thm:canonical}.
Let $x$ be such that $0<F(2x)<1$ and a point where $\frac d{dx} F(2x)$ {and $\frac{d}{dx} F(2(x-1))$ are} well defined. Then $\frac d{dx} F(2x)>0$. Setting $y=2x$, we have 
}
\begin{align*}
    \frac{\frac{d}{dx} F(2(x-1))}{\frac{d}{dx} F(2x)}
    &= \frac{\frac{d}{dy} F(y-2) 2}{\frac{d}{dy} F(y) 2}\\
    &= \frac{\frac{d}{dy} f(1-F(y-1))}{F'(y)}\\
    &= \frac{\frac{d}{dy} f(1-[f(1-F(y))])}{F'(y)}\\
    &= \frac{\frac{d}{dy} g(1-F(y))}{F'(y)}\\
    &= -g'(1-F(y)),
\end{align*}
which is increasing because $g$ is convex and $1-F(x)$ is decreasing. {If $F(2x)=0$, then $F(2(x-1))=0$ as well, and the ratio of the derivatives is $0/0$. If $F(2x)=1$ and $F(2(x-1))\in \{0,1\}$, then we also get the undefined ratio of $0/0$. In the case that $F(2x)=1$, but $F(2(x-1))<1$, the ratio is $\frac{\frac{d}{dx} F(2(x-1))}{\frac{d}{dx} F(2x)}=+\infty$. In each case, we have that the ratio is increasing, except when it has the form $0/0$ (which has probability zero under either $F(2(\cdot))$ or $F(2(\cdot-1))$}. 

It remains to verify property 1 of Definition \ref{def:CND}. Let $S_0$ and $S_1$ be two real values such that $|S_0-S_1|\leq \Delta$. Let $N\sim F$ and $N_2\sim F_2$. Note that $N_2\overset d = \frac 12 N$. Call $S_2 = (1/2)(S_0+S_1)$, and observe that $|S_2-S_0|\leq \Delta/2$ and $|S_2-S_1|\leq \Delta/2$. Then since $N$ is drawn from a CND for $f$,
\begin{align*}
    T(S_0+\Delta N_2,S_2+\Delta N_2)&=T\left(S_0+\frac{\Delta}{2} N,S_2+\frac{\Delta}{2}N\right)\geq f\\
    T(S_2+\Delta N_2,S_1+\Delta N_2)&=T\left(S_2+\frac{\Delta}{2} N,S_1+\frac{\Delta}{2}N\right)\geq f,
\end{align*}
since $F$ is a CND for $f$. Then by Lemma \ref{lem:A5}, we have $T(S_0+\Delta N_2,S_1+\Delta N_2)\geq f(1-f)=g$.
\end{proof}

Before we finally prove Theorem \ref{thm:semi}, we recall the definition of Neyman structure, and its connection to unbiased tests.

\begin{defn}
  [Definition 4.120 of Schervish]
  Let $G\subset \Theta$. If $T$ is a sufficient statistic for $G$, then a test $\phi$ has \emph{Neyman structure relative to  $G$ and $T$} if $\EE_{\theta}[\phi(X)\mid T=t]$ is constant in $t$ for all $\theta \in G$.
\end{defn}

\begin{thm}
  [Theorem 4.123 of Schervish]\label{thm:Neyman}
  Let $G = \ol \Theta_0 \cap \ol \Theta_1$. Let $T$ be a boundedly complete sufficient statistic for $G$. Assume that the power function is continuous. If there is a UMP unbiased level $\al$ test $\phi$ among those which have Neyman structure relative to $G$ and $T$, then $\phi$ is UMP unbiased level $\al$.
\end{thm}

\thmsemi*
\begin{proof}%[Proof of Theorem \ref{thm:semi}.]
By Theorem \ref{thm:Neyman}, %Proposition \ref{prop:unbiased} and Lemma \ref{lem:Neyman}, 
it suffices to consider tests which have Neyman structure relative to $\{(\theta_x,\theta_y)\mid \theta_x=\theta_y\}$ and $Z=X+Y$. So, we need only consider tests $\phi(x,y)$ that satisfy $\EE_{\theta_x=\theta_y}[\phi(x,y)\mid x+y=Z]=\alpha$ for all $Z\in \{0,1,2,\ldots, m+n\}$. By Theorem \ref{thm:Neyman}, it suffices to show that $\phi^*$ is UMP {among the tests in $\Phi_f^{\text{semi}}$  which also satisfy} $\EE_{\theta_x=\theta_y}[\phi(x,y)\mid x+y=Z]=\alpha$. 

Recall that if $X\sim \mathrm{Binom}(n,\theta_X)$ and $Y\sim \mathrm{Binom}(m,\theta_Y)$, then $(X,Y)\mid X+Y=z$ is equal in distribution to $(z-H,H)$, where $H\sim \mathrm{Hyper}(m,n,z,\omega)$, where $\omega = \frac{\theta_Y/(1-\theta_Y)}{\theta_X/(1-\theta_X)}$ and where $\mathrm{Hyper}(m,n,z,\omega)$ is the Fisher noncentral hypergeometric distribution, which has pmf 
\[P_\omega(H=x)=\frac{\binom{m}{x}\binom{n}{z-x}\omega^x}{\sum_{x=L}^{U} \binom{m}{x}\binom{n}{z-x}\omega^x},\]
with support $\left\{L=\max\{0,z-n\},L+1,\ldots,U-1, U=\min\{z,m\}\right\}$. Then unbiased testing $H_0: \theta_X\leq \theta_Y$ versus $H_1: \theta_Y\geq \theta_X$ in the original model is equivalent to testing $H_0: \omega\leq 1$ versus $H_1: \omega>1$ in the hypergeometric model. 

Next, note that $\mathrm{Hyper}(m,n,z,\omega)$ has an increasing likelihood ratio in $\omega$, meaning that for $\omega_1\leq \omega_2$, we have $\frac{P_{\omega_1}(H=x)}{P_{\omega_2}(H=x)}$ is an increasing function of $x$. Then a test on hypergeometrics for $H_0: \omega\leq 1$ versus $H_1: \omega_1>1$ has size $\alpha$ if and only if $\EE_{H\sim \mathrm{Hyper}(m,n,z,1)}\phi(H)=\alpha$. For the tradeoff function $g=1-(1-f)^{\circ 2}$, by Theorem \ref{thm:binary}, there exists a most powerful $\Psi_g$ test for $H_0: \omega=1$ versus $H_1: \omega=\omega_1$ where $\omega_1>1$, which is of the form $\psi_z^*(x) = F_g(x-m(z))$, where $m(z)$ is chosen such that $\EE_{H\sim \mathrm{Hyper}(m,n,z,1)}F_g(H-m)=\alpha$. Since this test does not depend on the specific alternative, it is UMP for $H_0: \omega\leq 1$ versus $H_1: \omega>1$. 

Now, given $\psi^*_z$ for $z\in\{0,1,2,\ldots, m+n\}$, we will show that $\phi^*$ as defined in the theorem statement is related as follows. Let $z\in \{0,1,2,\ldots, m+n\}$. Then for $x+y=z$, we have 
\[\phi^*(x,y) = F_f(y-x-c(z))=F_f(2y-z-c(z))=F_f(2y-2m(z))=F_g(y-m(z))=\psi^*_z(y),\] 
where $m(z)=(1/2)(z+c(z))$, and where we used Lemma \ref{lem:twice} to justify that $F_f(2\cdot)=F_g(\cdot)$. 

Now, suppose that there is another $\psi_{f}^{\text{semi}}$ test which satisfies $\EE_{H\sim \mathrm{m,n,z,1}} \phi(z-H,H)=\alpha$, and which has higher power than $\psi^*$ for some $\theta_X<\theta_Y$ (call $\omega_1 = \frac{\theta_Y/(1-\theta_Y)}{\theta_X/(1-\theta_X)}$). Because power can be expressed as $\EE_{Z} \EE_{H\sim \mathrm{Hyper}(m,n,Z,\omega_1)} \phi(Z-H,H)$, where the first expectation is over the marginal distribution of $Z$, this implies that there exists $z\in \{0,1,2,\ldots, m+n\}$ such that $\EE_{H\sim \mathrm{Hyper}(m,n,z,\omega_1)} \phi(z-H,H) >\EE_{H\sim \mathrm{Hyper}(m,n,z,\omega_1)} \phi^*(z-H,H)$. However, applying the transformation in Lemma \ref{lem:convert} gives {test in $\Psi_g$ with size $\alpha$} for testing  $H_0: \omega\leq 1$ versus $H_1: \omega>1$ in the family $\mathrm{Hyper}(m,n,z,\omega)$, with power at $\omega_1$ higher than $\psi^*_z$. This contradicts that $\psi^*_z$ is UMP size $\alpha$ in $\Psi_g$. We conclude that $\phi^*$ is UMP unbiased size $\alpha$ among $\Psi_f^{\text{semi}}$ for the hypothesis $H_0: \theta_X\leq \theta_Y$ versus $H_1: \theta_X<\theta_Y$. 

Line 2 of the theorem statement follows from the monotone likelihood ratio property of the Fisher noncentral hypergeometric distribution along with parts 3 and 4 of Theorem \ref{thm:pVal}. 
\end{proof}

\end{appendix}

\end{document}